\documentclass{article}

\usepackage{amsfonts, amssymb, amsmath, amsthm}
\usepackage[top = 2.0cm, bottom = 2.0cm, left = 2.0cm, right = 2.0cm]{geometry}
\usepackage{pgf,tikz}
\usepackage[all]{xy}
\usepackage{cite,hyperref}

\usetikzlibrary{arrows, positioning, patterns, shapes, external}
\definecolor{jgreen}{rgb}{0,0.7,0}

\tikzstyle{c} =	[coordinate]
\tikzstyle{v} = 	[circle, draw=black, line width=.5pt, fill=black, inner sep=0pt, minimum size=1.5mm]
\tikzstyle{vb} =	[circle, draw=black, line width=.5pt, fill=red, inner sep=0pt, minimum size=1.5mm]
\tikzstyle{vh} =	[circle, draw=black, line width=.5pt, fill=blue, inner sep=0pt, minimum size=1.5mm]
\tikzstyle{vs} =	[circle, draw=black, line width=.5pt, fill=jgreen, inner sep=0pt, minimum size=1.5mm]
\tikzstyle{e} =	[draw=red,line width=1.6pt]
\tikzstyle{eb} =	[draw=jgreen,line width=1.6pt]
\tikzstyle{eh} =	[dashed]
\tikzstyle{es} =	[draw=blue,line width=1.6pt]
\tikzstyle{f} = 	[line width=0.01pt,dotted,fill=blue, fill opacity=.1]
\tikzstyle{cs} = 	[draw=none,fill=red, fill opacity=.7]
\tikzstyle{bb} =	[dashed,line width=1.4pt]
\tikzstyle{bh} =	[dotted,line width=1pt]
\tikzstyle{h} = 	[ellipse, inner sep=0.1pt, draw=black]
\tikzstyle{hv} = 	[ellipse, inner sep=0.1pt, draw=red]
\tikzstyle{he} = 	[ellipse, inner sep=0.1pt, draw=blue]
\tikzstyle{hb} = 	[ellipse, inner sep=0.1pt, draw=jgreen]

\newcommand{\simplicial}{{simplicial}}
\newcommand{\trace}{{trace}}

\newcommand{\extd}{\mathrm{d}}
\newcommand{\tr}{\mathrm{tr}}

\newcommand{\gft}{\textsc{gft}}
\newcommand{\mfgft}{\textsc{mf-gft}}
\newcommand{\pgft}{\textsc{p-gft}}
\newcommand{\dwgft}{\textsc{dw-gft}}
\newcommand{\eprl}{\textsc{eprl}}
\newcommand{\kkl}{\textsc{kkl}}
\newcommand{\lqg}{\textsc{lqg}}
\newcommand{\SF}{\textsc{sf}}

\newcommand{\sub}{\textsc{sub}}

\newcommand{\so}{\mathfrak{so}}
\newcommand{\su}{\mathfrak{su}}
\newcommand{\SO}{\mathrm{SO}}
\newcommand{\SU}{\mathrm{SU}}

\newcommand{\bg}{\mathfrak{c}}
\newcommand{\bgs}{\mathfrak{C}}
\newcommand{\bgt}{\widetilde{\mathfrak{c}}}
\newcommand{\bgst}{\widetilde{\mathfrak{C}}}

\newcommand{\bbg}{\mathfrak{b}}
\newcommand{\bbgs}{\mathfrak{B}}
\newcommand{\bbgt}{\widetilde{\mathfrak{b}}}
\newcommand{\bbgst}{\widetilde{\mathfrak{B}}}

\newcommand{\sfa}{\mathfrak{a}}
\newcommand{\sfas}{\mathfrak{A}}
\newcommand{\sfat}{\widetilde{\mathfrak{a}}}
\newcommand{\sfast}{\widetilde{\mathfrak{A}}}

\newcommand{\bp}{{\mathfrak{p}}}
\newcommand{\bps}{\mathfrak{P}}
\newcommand{\bpt}{{\widetilde{\mathfrak{p}}}}
\newcommand{\bpst}{\widetilde{\mathfrak{P}}}

\newcommand{\sfr}{\mathfrak{m}}
\newcommand{\sfrs}{\mathfrak{M}}
\newcommand{\sfrt}{\widetilde{\mathfrak{m}}}
\newcommand{\sfrst}{\widetilde{\mathfrak{M}}}

\newcommand{\sfm}{\mathfrak{m}}
\newcommand{\sfms}{\mathfrak{M}}
\newcommand{\sfmt}{\widetilde{\mathfrak{m}}}
\newcommand{\sfmst}{\widetilde{\mathfrak{M}}}

\newcommand{\sfo}{\mathfrak{o}}
\newcommand{\sfos}{\mathfrak{O}}
\newcommand{\sfl}{\mathfrak{l}}
\newcommand{\sfls}{\mathfrak{L}}

\newcommand{\sta}{\mathfrak{s}}

\newcommand{\vbar}{\bar{v}}
\newcommand{\Vbar}{\overline{\mathcal{V}}}
\newcommand{\vhat}{\hat{v}}
\newcommand{\Vhat}{\widehat{\mathcal{V}}}

\newcommand{\Phit}{\widetilde{\Phi}}
\newcommand{\phit}{\widetilde{\phi}}
\newcommand{\Obt}{\widetilde{\mathcal{O}}}
\newcommand{\obt}{\widetilde{O}}

\newcommand{\m}{n}
\newcommand{\copies}{n}
\newcommand{\loopless}{\textsc{l}}
\newcommand{\primitive}{\textsc{s}}

\newcommand{\primitivedually}{\textsc{s}\textrm{-}\textsc{dw}}
\newcommand{\looplessdually}{\textsc{l}\textrm{-}\textsc{dw}}

\newcommand{\Simplicial}{\textsc{s}}
\newcommand{\Simplicialdually}{\textsc{s}\textrm{-}\textsc{dw}}

\newcommand{\looptosimp}{\textsc{l-s}}

\newcommand{\Acal}{\mathcal{A}}
\newcommand{\Ccal}{\mathcal{C}}
\newcommand{\Dcal}{\mathcal{D}}
\newcommand{\Ecal}{\mathcal{E}}
\newcommand{\Fcal}{\mathcal{F}}

\newcommand{\Ical}{\mathcal{I}}

\newcommand{\Mcal}{\mathcal{M}}
\newcommand{\Ocal}{\mathcal{O}}
\newcommand{\Rcal}{\mathcal{R}}
\newcommand{\Vcal}{\mathcal{V}}

\newcommand{\Nbb}{\mathbb{N}}

\newcommand{\Rbb}{\mathbb{R}}

\newcommand{\bbB}{\mathbb{B}}  

\newcommand{\Dbb}{\mathbb{D}}
\newcommand{\Ibb}{\mathbb{I}}
\newcommand{\Kbb}{\mathbb{K}}
\newcommand{\Pbb}{\mathbb{P}}
\newcommand{\Sbb}{\mathbb{S}}
\newcommand{\Vbb}{\mathbb{V}}

\newcommand{\bbBt}{\widetilde{\mathbb{B}}}
\newcommand{\Kbbt}{\widetilde{\mathbb{K}}}

\newcommand{\Vbbt}{\widetilde{\mathbb{V}}}

\newcommand{\Kbbb}{\overline{\mathbb{K}}}
\newcommand{\Vbbb}{\overline{\mathbb{V}}}
\newcommand{\bbBb}{\overline{\mathbb{B}}}

\newcommand{\dbb}{\mathbf{d}}
\newcommand{\ibb}{\mathbf{i}}
\newcommand{\pbb}{\mathbf{p}}
\newcommand{\vbb}{\mathbf{v}}

\newcommand{\fdiagram}{\Gamma}

\newcommand{\ra}{\rightarrow}

\newcommand{\us}{\underset}

\newcommand{\In}{\subset}

\newcommand{\Cc}{\mathcal{C}}
\newcommand{\SC}{\Cc^\mathrm{sim}}
\newcommand{\PC}{\Cc^\mathrm{poly}}
\newcommand{\ssub}{\Delta}
\newcommand{\br}{\partial}

\newcommand{\Vb}{\overline{\mathcal{V}}}
\newcommand{\vb}{\bar v}
\newcommand{\Vh}{\widehat{\mathcal{V}}}
\newcommand{\vh}{\hat v}
\newcommand{\Eb}{\overline{\mathcal{E}}}

\newcommand{\cGFT}{Freidel:2005jy,Oriti:2006ts,Oriti:2012wt,Krajewski:2012wm,Baratin:2012ge,Oriti:2009ur,Oriti:2014wf}
\newcommand{\ckkl}{Kaminski:2010ba}
\newcommand{\cklp}{Kisielowski:2012bo}
\newcommand{\ceprl}{Engle:2007em,Engle:2008ka,Engle:2008fj}
\newcommand{\cfk}{Freidel:2008fv}
\newcommand{\cbo}{Baratin:2012br}
\newcommand{\cOperator}{Bahr:2011ey,Bahr:2012iu}
\newcommand{\cguraulost}{Gurau:2010iu}
\newcommand{\csmerlaklost}{Smerlak:2011ea}
\newcommand{\cgftrenorm}{Freidel:2009ek,BenGeloun:2013fw,BenGeloun:2013dl,BenGeloun:2013ek,BenGeloun:2013uf,Samary:2014bs}
\newcommand{\cCOR}{Carrozza:2014ee,Carrozza:2014bh,Carrozza:2014tf}

\newcommand{\cGFC}{Gielen:2013cr,Gielen:2014gv,Gielen:2014ca,Calcagni:2014uz,Gielen:2014vk,Sindoni:2014vs}
\newcommand{\clargeN}{Bonzom:2011cs,Gurau:2012ek,Gurau:2012hl,Bonzom:2012bg,Gurau:2012td,Gurau:2012wj,Baratin:2014bea}
\newcommand{\cdouble}{Gurau:2011sk}

\theoremstyle{definition}
\newtheorem{defin}{Definition}[section]
\newtheorem{remark}[defin]{Remark}

\theoremstyle{plain}
\newtheorem{proposition}[defin]{Proposition}

\newtheorem{conjecture}[defin]{Conjecture}
\newtheorem{corollary}[defin]{Corollary}

\begin{document}



\title{\bf Group field theories for all loop quantum gravity}

\author{Daniele\ Oriti, James\ P.\ Ryan, Johannes\ Th\"urigen%
\footnote{{daniele.oriti@aei.mpg.de}, {james.ryan@aei.mpg.de}, {johannes.thuerigen@aei.mpg.de}}
\\[0.2cm]
\small MPI f\"ur Gravitationsphysik, Albert Einstein Institut, Am M\"uhlenberg 1, D-14476 Potsdam, Germany
}

\date{}

\maketitle

\begin{abstract}
Group field theories represent a 2nd quantized reformulation of the loop quantum gravity state space and a completion of the spin foam formalism. States of the canonical theory, in the traditional continuum setting, have support on graphs of arbitrary valence. On the other hand, group field theories have usually been defined in a simplicial context, thus dealing with a restricted set of graphs. In this paper, we generalize the combinatorics of group field theories to cover all the loop quantum gravity state space. As an explicit example, we describe the group field theory formulation of the \kkl\ spin foam model, as well as a particular modified version. We show that the use of tensor model tools allows for the most effective construction. In order to clarify the mathematical basis of our construction and of the formalisms with which we deal, we also give an exhaustive description of the combinatorial structures entering spin foam models and group field theories, both at the level of the boundary states and of the quantum amplitudes.
\end{abstract}



\section{Introduction}

\newcommand{\fk}{\textsc{fk}}
\newcommand{\bo}{\textsc{bo}}

The field of non-perturbative, background independent,  quantum gravity has witnessed several important developments in the last decades. 

In particular, loop quantum gravity \cite{Thiemann:2007wt,Ashtekar:2004fk,Rovelli:2004wb} emerged as a prominent candidate in the endeavour to fully describe the kinematics of quantum geometry. At its base lie quantum states that can be defined in a purely algebraic and combinatorial manner, a complete basis for which is provided by spin networks: graphs labelled by irreducible representations of the Lorentz or the rotation group. Furthermore, a quantum dynamics for such quantum geometric states can be rigorously defined, although both its solution and the extraction of effective classical dynamics are fraught with difficulties. 

On the covariant side, spin foam models \cite{Baez:2000kp, Oriti:2003uw, Perez:2003wk,Rovelli:2004wb,Perez:2013uz,Rovelli:2011tk,Bianchi:2013fh} rose to prominence both as a new approach to lattice gravity path integrals and as a covariant definition of the dynamics of loop quantum gravity states. 
They are similarly based on combinatorial and algebraic structures. Spacetime is replaced by a (simplicial) 
complex and discrete quantum geometric data. This data comes in the form of group/algebra elements or representations, labelling various components of the 
complex. It plays the role of the discrete metric, reproducing at the covariant level the histories for quantum states. In the case of 4--dimensional quantum gravity, the most actively studied models are the \eprl\ model \cite{Engle:2008ka,Engle:2008fj}, the \fk\ model \cite{Freidel:2008fv} and the \bo\ model \cite{\cbo}. 

Group field theory (\gft) \cite{\cGFT} also took a more central role in the quantum gravity landscape, in connection to loop quantum gravity and spin foam models. These are quantum field theories on group manifolds characterized by a peculiar non--local pairing of field variables in their interactions and motivated from both the canonical and the covariant perspective. 
In fact, on the one hand, they represent a 2nd quantized, Fock space reformulation of the loop quantum gravity state space. In this capacity, spin network vertices play the role of fundamental quanta, created/annihilated by field operators. 
Meanwhile, their canonical quantum equations of motion (e.g. the Hamiltonian constraint equation)  are encoded in (a sector of) the quantum equations of motion for the $n$--point functions of the corresponding field theory \cite{Oriti:2013vv}. 
On the other hand, they provide a completion of the spin foam formalism. 
A spin foam model, defined on a given (simplicial) 
complex, encodes a finite number of degrees of freedom. This presents the issue of defining a quantum dynamics for the infinite degrees of freedom that one expects a quantum gravity theory to possess. One strategy, following the lattice gravity interpretation of spin foam models, is to define some refinement procedure for the spin foam complex. Thereafter, one looks for fixed points as the renormalized amplitudes flow under coarse graining \cite{Dittrich:2012ba}. 
A second strategy focusses on defining an appropriate sum over spin foams, including a sum over the 
complexes themselves \cite{Reisenberger:2001hd,Baez:2000kp,Oriti:2003uw, Perez:2003wk}. This is more directly in line with the interpretation of spin foams as histories of spin networks.  \gft s provide a natural and elegant way to define this sum. In fact, for any given spin foam model, there is a \gft\ model, whose perturbative expansion around the Fock vacuum, generates a series catalogued by spin foam complexes, weighted by the appropriate amplitude. In other words, the spin foam complexes arise as \gft\ Feynman diagrams and spin foam amplitudes as \gft\ Feynman amplitudes. Thus, \gft\ models complete the spin foam picture and moreover, \emph{any} \gft\  model defines a complete spin foam model. 
If one then keeps in mind that \gft s give a 2nd quantized formulation of canonical \lqg, one obtains a direct link between the canonical and covariant approaches \cite{Oriti:2013vv}.

Such sums over complexes is not reserved solely for the \gft\ formalism.  
Tensor models \cite{Gurau:2012hl} are a generalization of matrix models \cite{Francesco:1995ih} to dimensions greater than two. They can be seen as providing a stripped-down version of \gft s, reducing them to purely combinatorial models. Indeed, the group--theoretic data are dropped altogether; equivalently, this can be seen as restricting the discrete geometric data to the graph--distance metric and working with equilateral triangulations. 
As a result, the amplitudes depend only on the combinatorics of the simplicial 
complexes. This allows one to focus principally on the sum over complexes. Indeed, many of the recent advances in tensor models exert increasing analytic control over such series.  Some of these advances have been already extended to the more involved \gft\ framework. 
Thus, one should expect that techniques from tensor models could play a greater role in the context of spin foam models and loop quantum gravity, since one needs to exert analytical control over the spin foam sum, as well as the combinatorial structure of quantum states. This paper provides one example of this fruitful exchange.

An important issue concerns the combinatorial structure of graphs and complexes and directly affects the relation between the canonical and covariant approaches. On the one hand, the set of graphs supporting quantum states of the canonical theory includes graphs of arbitrary valence. 
This stems partly from its historic origin as a direct quantization of a continuum gravity theory. 
On the other hand, spin foam models have often been defined to evolve quantum states with support on a restricted set of graphs, those that may be endowed with a simplicial interpretation. 
Such a choice has several motivations: it facilitates calculations; it is in this restricted context that their discrete geometric properties are best understood, in particular,  the so--called simplicity constraints that reduce topological \textsc{bf} theory to gravity \cite{Barbieri:1998kx,Barrett:1998fp,Barrett:2000fr}; such states arise as a superselection sector for certain \lqg\  Hamiltonian constraints.  This also implies that the boundary data in the covariant setting 
are simplicial. Importantly, current \gft s share the same type of boundary states and amplitudes.

To ensure a better matching between canonical \lqg\ and covariant spin foam models, as well as to have a \gft\ formulation for both approaches, one may want to generalize the combinatorial structures appearing in spin foam models and \gft\ to arbitrary graphs and complexes. A second motivation arises from the study of physical applications and the continuum limit, where it may be worth possessing a larger class of models at the outset. Afterwards,  physical rather than aesthetic or mathematical reasons restrict the combinatorial structures taking part. This has been already done, at least partially, in the context of \gft\ renormalization \cite{\cgftrenorm, \cCOR}, following developments in tensor models \cite{Gurau:2012hl, Bonzom:2012bg}. 

From another perspective, the matching with canonical \lqg\ could also be achieved the other way around, i.e.~by working with a simplicial version of the canonical theory. 
Rather than dealing with a quantization of continuum general relativity (\textsc{gr}), one views continuum \textsc{gr} arising only as the effective theory of the quantum dynamics for fundamental structures that are intrinsically discrete.
From this point of view, it makes sense to start with the simplest possible discrete structures, provided they are general enough to recover continuum manifolds in some approximation, and generalize them only if and when necessary. With respect to this criterion, simplicial complexes are sufficiently general. Let us emphasize that the other reason for restricting to simplicial complexes (and fixed--valence graphs) is practical. Controlling the sums over complexes and dealing with arbitrary superpositions of graph--based states is complicated enough when their combinatorics is restricted. A generalization would seem, a priori, to make things worse. 

While the reluctance to complicate things may explain the delay in developing combinatorial generalizations of spin foam models and \gft s, there is no obstruction, mathematical or conceptual, to doing so. In fact, a generalization of current spin foam models, in particular the \eprl\ model, to a larger set of 
complexes was provided in \cite{\ckkl}. The aim of this article is to show that the \gft\ framework is also completely comfortable with the generation of complexes evolving arbitrary spin network states. 

More precisely, our results are the following: 

We give an extensive and exhaustive description of the combinatorial structures entering spin foam models and \gft s, both at the level of the boundary states and of the quantum amplitudes. 
To this end, we define spin foam atoms and molecules as structures that are directly adapted to the needs of spin foams and \gft.   
In addition, we also investigate and detail the fundamental properties of combinatorial complexes in a more precise, mathematical sense, introducing the concept of abstract polyhedral complexes as a generalization of abstract simplicial complex to include abstract polytopes.
This prepares the mathematical foundation and the intuition for the \gft\ construction.  Moreover, we believe it is of intrinsic value, clarifying, relating and extending several results in the literature. 

We generalize the combinatorics of \gft s through two mechanisms.  The first proposal constitutes a very formal (and thereby somewhat trivial) generalization of the \gft\ formalism to one based on an infinite number of fields. Having said that, this multi--field \gft\  generates series catalogued by arbitrary 2--complexes, while arbitrary graphs label quantum states. This is a direct counterpart of the \kkl--extension of gravitational spin foam models. It shows the absence of any fundamental obstruction to accommodating arbitrary combinatorial structures. However, one does not expect such a field theory to be 
useful since the sum over complexes appears no more tamed than before. 
 
  More interesting and much more manageable is our second construction. 
  As is well known, the standard simplicial \gft\ contains an interaction based upon a 2--complex that may be interpreted as the dual 2--skeleton of a $D$--simplex. 
Remarkably, at the 2--complex level, arbitrary 2--complexes can be decomposed in terms of this simplicial 2--complex. 
Thus, the standard \gft\ is sufficient to generate arbitrary 2--complexes. However, the subtle issue is to assign correct amplitudes.  This is solved by a mild extension, wherein one augments the data set over which the \gft\ field is defined, so as to exert more sensitive control over the combinatorial structures generated by the theory. This is known as dual--weighting and permits one to tune the theory to a regime, in which the perturbative series are catalogued by appropriately weighted arbitrary spin foams (not just simplicial spin foams). 
Here we accomplish two things. First, we give an explicit \gft\ formulation of the \kkl--extension of the \eprl\ model (and of other similar spin foam models).  Second, we propose a new (set of) model(s) incorporating similar constraints that are arguably better motivated from the geometric point of view. 

The presentation of these results is structured as follows. In Section \ref{sec:comb}, we set the stage for defining the generalized \gft s, discussing the combinatorics structures upon which they are supported, as well as introducing all the relevant concepts for the constructive way spin foam molecules (combinatorial 2--complexes) are generated in \gft s as a bonding of atoms determined by their boundary graphs. In particular, we show how these graphs and molecules can be decomposed into graphs and atoms of a simplicial kind. In the appendix, we show that these combinatorics are indeed the ones of combinatorial 2--complexes, that is, the $\m=2$ case of abstract polyhedral $\m$-complexes. 
Then, in Section \ref{sec:gft}, we review the definition of \gft s and rephrase them in terms of the generalized combinatorics language. The definition of multi--field \gft\ which generates arbitrary spin foam molecules is thereafter straightforward. For the implementation of dually--weighted \gft s, we detail the dual--weighting mechanism that realizes, in a dynamical manner, the decomposition of generic molecules in terms of simplicial building blocks. 
Finally, in Section \ref{sec:sfm} we show how gravitational spin foam models incorporating the relevant simplicity constraints in their operators can be generalized to both the multi-field \gft\ as well as the dually weighted \gft. As an example we present the details for the \eprl--simplicity constraints.


\newpage


\section{Combinatorics of spin foams} \label{sec:comb}

The graphical and topological structures, upon which spin foam models have support, tend to have a rather molecular structure. This has been noted and explained in detail in \cite{\cklp}.  The coming section includes a self--contained description of these structures, one that increases its utility within the group field theory framework. Moreover, while we have consciously opted for a physicochemical naming convention,  rather than the cephalopodal counterpart used in \cite{\cklp}, we stress that its use is for purely intuitive purposes.

 Given the technical nature of the coming section, we present a synopsis of the main points.

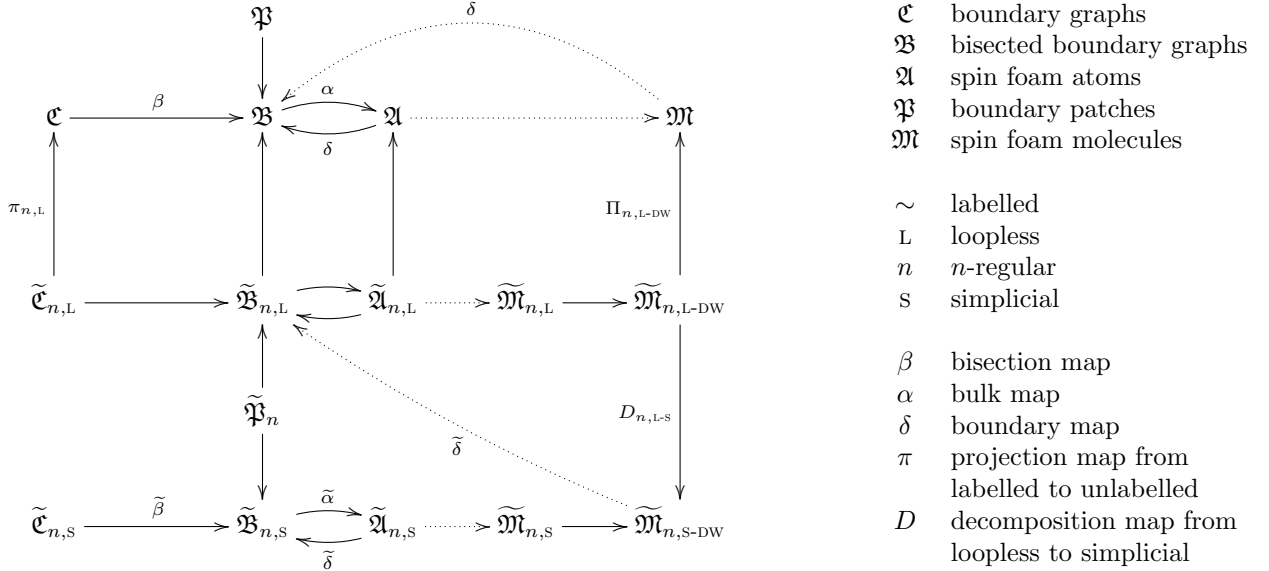
\begin{minipage}{0.7\textwidth}
\begin{displaymath}
  \xymatrix{
   & &\bps \ar[d] \\ 
   \bgs \ar[rr]^{\beta}& &\bbgs \ar@/^/[r]^{\alpha}  &\sfas \ar@/^/[l]^{\delta} \ar@{.>}[rr]&&\sfrs\ar@/_{3pc}/@{.>}[lll]_{\delta}\\
   &\\
   \bgst_{n,\loopless} \ar[rr]\ar[uu]^{\pi_{n,\loopless}} & &\bbgst_{n,\loopless} \ar@/^/[r]\ar[uu]  &\sfast_{n,\loopless} \ar@/^/[l]\ar[uu]\ar@{.>}[r] &\sfrst_{n,\loopless}\ar[r]&\sfrst_{n,\looplessdually}\ar[uu]^{\Pi_{n,\looplessdually}}\ar[dd]_{D_{n,\looptosimp}}&\\
   & &\bpst_{n
   }\ar[u] \ar[d] & & & &\\
 \bgst_{n,\primitive}\ar[rr]^{\widetilde{\beta}}&&
    \bbgst_{n,\primitive}\ar@/^/[r]^{\widetilde{\alpha}}& \sfast_{n,\primitive}\ar@/^/[l]^{\widetilde{\delta}}\ar@{.>}[r]&\sfrst_{n,\primitive}\ar[r]& \sfrst_{n,\primitivedually}\ar@/^/@{.>}[uulll]^{\widetilde{\delta}}
 }
\end{displaymath}
\end{minipage}
\begin{minipage}{0.3\textwidth}
\begin{tabular}{cl}
 $\bgs$ 	& boundary graphs \\
$\bbgs$	& bisected boundary graphs \\
$\sfas$ 	& spin foam atoms \\
$\bps$	& boundary patches \\
$\sfrs$	& spin foam molecules \\
\\
$\sim$	& labelled \\
$\loopless$& loopless \\
$n$		& $n$-regular \\
$\primitive$& simplicial \\ 
\\
$\beta$	& bisection map\\
$\alpha$	& bulk map \\
$\delta$	& boundary map\\
$\pi$		& projection map from \\
		& labelled to unlabelled\\
$D$		& decomposition map from\\
		& loopless to simplicial

\end{tabular}
\end{minipage}
  
One starts with a set of boundary graphs $\bgs$ that provide support for 
loop quantum gravity states. For a graph $\bg\in\bgs$, one arrives at the corresponding bisected boundary graph  $\bbg= \beta(\bg)\in\bbgs$ by bisecting each of its edges. The graph $\bbg$ can be augmented to arrive at the corresponding 2--dimensional spin foam atom $\sfa = \alpha(\bbg) \in\sfas$.  This spin foam atom $\sfa$ is the simplest spin foam structure with $\bbg$ as a boundary: $\bbg= \delta(\sfa)$.  Moreover, the bisected boundary graph $\bbg$ can be decomposed into boundary patches $\bp\in\bps$.  The boundary patches are important because it is along these patches that atoms are bonded to form composite structures, known as spin foam molecules $\sfrs$. The boundary ($\delta$) of these molecules are (generically a collection of) graphs in $\bbgs$. Moreover, the molecules are the objects generated in the perturbative expansion of the group field theory. 

From the \gft\ perspective, however,  one looks for as concise a way as possible to generate such structures.  It emerges that the complexity of the \gft\ generating function can be infinitely reduced by considering labelled ($\sim$), $n$--regular, loopless ($\loopless$) graphs $\bgst_{n,\loopless}$. The labels are associated to each edge and drawn from the set $\{real, virtual\}$, while loopless means that the terminus of any edge does not coincide with its source. For this set of objects, one can then follow an analogous procedure to generate $\bbgst_{n,\loopless}$, $\sfast_{n,\loopless}$ and $\sfrst_{n,\loopless}$.    

There is a surjection $\pi_{n,\loopless}:\bgst_{n,\loopless}\longrightarrow\bgs$, meaning that each graph in $\bgs$ is represented by a class of graphs in $\bgst_{n,\loopless}$.  This surjection can be extended to $\bbgst_{n,\loopless}$ and $\sfast_{n,\loopless}$ but not the molecules $\sfrst_{n,\loopless}$.  However, one can identify a subset $\sfrst_{n,\looplessdually}\subset\sfrst_{n,\loopless}$, for which one can extend $\pi_{n,\loopless}$ to a surjection $\Pi_{n,\looplessdually}:\sfrst_{n,\looplessdually}\longrightarrow\sfrs$.  Thus, every molecule in $\sfrs$ is represented by a class of molecules in $\sfrst_{n,\looplessdually}$.

The key now is that the patches making up any graph in $\bbgst_{n,\loopless}$ come from a finite set of patches $\bpst_{n}$, called $n$--patches. Using these patches one can pick out a finite subset of \simplicial\ $n$--graphs $\bgst_{n,\primitive}\subset\bgst_{n,\loopless}$,
that are based on the complete graph over $n+1$ vertices.
$\bbgst_{n,\primitive}$, $\sfast_{n,\primitive}$ and $\sfrst_{n,\primitive}$ follow as before. 

While $\bgst_{n,\primitive}$, $\bbgst_{n,\primitive}$ and $\sfast_{n,\primitive}$ are finite sets, the set of \simplicial\ spin foam molecules $\sfrst_{n,\primitive}$ is infinite and contains a subset $\sfrst_{n,\Simplicialdually}$ whose elements reduce properly to molecules in $\sfrs$.  But the set $\sfrst_{n, \Simplicialdually}$ does not cover $\sfrs$ through some surjection, but maps onto a subset.  To cover all of $\sfrs$, one needs $\sfrst_{n,\looplessdually}$. Having said that, \textit{i}) there is a decomposition map $D_{n,\looptosimp}:\sfrst_{n,\looplessdually}\longrightarrow\sfrst_{n,\Simplicialdually}$ and \textit{ii}) every graph or collection of graphs from $\bbgst_{n,\loopless}$ arises as the boundary of some molecule in $\sfrst_{n,\primitivedually}$. As a result, $\sfmst_{n,\Simplicialdually}$ is sufficient to support a spin foam dynamics for arbitrary \lqg\ quantum states. 

\

The forthcoming construction is separated into six parts. The first and second catalogue the basic building blocks or atoms, along with the set of possible bonds that may arise between pairs of atoms.  These structures are drawn directly from those used in loop quantum gravity. Both the set of atoms and the set of their bonds are very large and inspire an attempt to find smaller subsets, introduced in the third and forth part, that still probe the whole space of graphical structures in some precisely defined sense which is explained and proven in the fifth and sixth part. 

After all these technicalities  we will discuss the relation of the 2-dimensional spin foam atoms and molecules to higher dimensional topologies in a seventh subsection.
Finally we will close this section emphasizing that the whole construction can be equivalently carried out in the language of stranded diagrams which is the usual one used in the \gft\ literature and is totally equivalent to the more \lqg\ oriented language of boundary graphs and spin foam atoms used in this work.


\subsection{Part 1: catalogue the basic building blocks} 
\label{ssec:stepOne}

This part focusses on defining the structure underlying loop quantum gravity and spin foams:
\begin{displaymath}
  \xymatrix{
   \textsc{unlabelled} &\bgs \ar[r]^{\beta} &\bbgs \ar@/^/[r]^{\alpha}  &\sfas \ar@/^/@{.>}[l]^{\delta}
   }
  \end{displaymath}

\begin{defin}[{\bf boundary graph}]
  \label{def:atomicbg}
  A \emph{boundary graph} is a double $\bg = (\overline{\mathcal{V}},\overline{\mathcal{E}})$, where $\overline{\mathcal{V}}$ is the vertex set and $\overline{\mathcal{E}}$ is the edge (multi)set,%
\footnote{A multiset is an extension of set concept, in which elements are allowed to occur multiple times.}
comprising of unordered two--element subsets of $\overline{\mathcal{V}}$,%
 \footnote{For a loop, the two--element subset is itself a multiset $(\bar{v}\bar{v})$.}
  subject to the condition that the graph is connected. 
\end{defin}
The set of boundary graphs is denoted by $\bgs$. 
Indeed this is just the set of connected multigraphs.
\begin{remark}
  One should note here that multi--edges (multiple edges joining two vertices), loops (edges whose two vertices coincide) and even 1--valent vertices (vertices with only one incident edge) are allowed.  Thus, $\bgs$ constitutes a very large set.  However, such graphs arise within loop quantum gravity, can be incorporated within the group field theory framework and so, in principle, serve as an appropriate starting point.  Later, this set can be whittled down to a more manageable subset. 
\end{remark}
\begin{defin}[{\bf bisected boundary graph}]
  \label{def:bisectedbg}
  A \emph{bisected boundary graph} is a double, $\bbg = (\mathcal{V}_{\bbg}, \mathcal{E}_{\bbg})$, constituting a bipartite graph with vertex partition $\mathcal{V}_{\bbg}=\overline{\mathcal{V}}\cup\widehat{\mathcal{V}}$, such that the vertices $\hat{v}\in\widehat{\mathcal{V}}$ are bivalent. 
\end{defin}
The set of bisected boundary graphs is denoted by $\bbgs$.  
\begin{proposition}
  \label{prop:bisectedcor}
  There is a bijection $\beta: \bgs \longrightarrow \bbgs$.
\end{proposition}
\begin{proof}
Given a boundary graph $\bg\in \bgs$, the \emph{bisection map} $\beta$ acts on each edge $\bar{e}=(\bar{v}_1\bar{v}_2)\in\overline{\mathcal{E}}$, replacing it by a pair of edges $\{(\bar{v}_1\hat{v}),(\bar{v}_2\hat{v})\}$, where $\hat{v}$ is a newly created bivalent vertex effectively bisecting the original edge. Thus, under the action of $\beta$:
  \begin{description}
    \item[--] $\overline{\mathcal{V}}\longrightarrow \mathcal{V}_{\bbg} = \overline{\mathcal{V}}\cup\widehat{\mathcal{V}}$, where $\widehat{\mathcal{V}}$ is the set of vertices bisecting the original edges of $\bg$;
    \item[--] $\overline{\mathcal{E}}\longrightarrow \mathcal{E}_{\bbg} =  \bigcup_{\bar{e}\in\overline{\mathcal{E}}}\{(\bar{v}_1 \hat{v}),\,(\bar{v}_2 \hat{v}): \bar{e} =  (\bar{v}_1\bar{v}_2)\}$ is the multiset of newly bisected edges.\footnote{\label{fn:self}Note that a loop $\bar{e} = (\bar{v}\bar{v})\in\overline{\mathcal{E}}$ is replaced by the multiset of edges $\{(\bar{v}\hat{v}),(\bar{v}\hat{v})\}$ and thus $\mathcal{E}_{\bbg}$ is a multiset.} 
  \end{description}
  This clearly results in an element of $\bbgs$ and the constructive nature of the map assures its injectivity. 

  Given a graph $\bbg\in\bbgs$, removing the vertex subset $\widehat{\mathcal{V}}$ and replacing the edge pair $\{(\bar{v}_1\hat{v}),(\bar{v}_2\hat{v})\}$ by $(\bar{v}_1\bar{v}_2)$ results in an element $\bg\in\bgs$ such that $\beta(\bg) = \bbg$. Thus, $\beta$ is surjective.  
\end{proof}
  A graph $\bg\in\bgs$ and its bisected counterpart $\bbg\in\bbgs$ are presented in Figure \ref{fig:boundary}. 

\begin{figure}[htb]
  \centering
  \tikzsetnextfilename{boundary}

\begin{tikzpicture}[scale=1.5]

\draw [|->] (2,0) -- node[label=above:$\beta$]{} (3,0);

\begin{scope}
\draw [eb] (1.5,1) circle (0.25cm);
\node [vb]		(a)	at (0,-1)		{};
\node [vb]		(b)	at (1.25,1)		{}; 
\node [vb]		(c)	at (0.5,0.5) 	{}; 
\node [vb]		(d)	at (-0.5,0.5)	{}; 
\node [vb]		(e)	at (-1.25,1)	{}; 
\node [vb]		(f)	at (1,-1)		{}; 
\foreach \i/\j in {a/b,a/c,a/d,a/e,b/c,b/e,c/d,d/e,a/f}{
 \draw [eb] (\i) -- (\j);
  }
\end{scope}

\begin{scope}[xshift=4.5cm]
\draw [eb] (1.5,1) circle (0.25cm);
\node [vb]		(a)	at (0,-1)		{};
\node [vb]		(b)	at (1.25,1)		{}; 
\node [vb]		(c)	at (0.5,0.5) 	{}; 
\node [vb]		(d)	at (-0.5,0.5)	{}; 
\node [vb]		(e)	at (-1.25,1)	{}; 
\node [vb]		(f)	at (1,-1)		{}; 
\node [vh]		(g)	at (1.75,1)		{}; 
\foreach \i/\j in {a/b,a/c,a/d,a/e,b/c,b/e,c/d,d/e,a/f}{
 \draw [eb] (\i) -- node[vh] {} (\j);
  }
\end{scope}

\end{tikzpicture}
  \caption{\label{fig:boundary} A boundary graph $\bg$ and its bisected counterpart $\bbg$.}
\end{figure}

\begin{remark}
  The bipartite property of the graphs $\bbg\in\bbgs$ means that the pairs $(\bar{v}\hat{v})\in\mathcal{E}_{\bbg}$ are ordered and thus, $\bbg$ is quite naturally a directed graph. 
\end{remark}

\begin{defin}[{\bf spin foam atom}]
  \label{def:sfatom}
  A \emph{spin foam atom} is a triple, $\sfa = (\mathcal{V}_{\sfa},\mathcal{E}_{\sfa},\mathcal{F}_{\sfa})$,  of vertices, edges and faces. It is constructed from the pair $(\bbg, \alpha)$, where ${\bbg}\in{\bbgs}$ and $\alpha$ is a \emph{bulk map} sending $\bbg$ to: 
  \vspace{-0.3cm}
  
\begin{description}
  \item[--] $\mathcal{V}_{\sfa} = \mathcal{V}\cup\mathcal{V}_{\bbg}$, where $\mathcal{V} = \{v\}$ is a one--element vertex set, containing the \emph{bulk} vertex; 
    
  \item[--] $\mathcal{E}_{\sfa} = \mathcal{E}\cup \mathcal{E}_{\bbg}$, where $\mathcal{E} = \bigcup_{u\in\Vcal_{\bbg}} \{ (vu) : v\in\mathcal{V}\}$. $\mathcal{E}$ contains precisely one edge for each vertex in $\Vcal_{\bbg}$, joining it to the bulk vertex $v$. Thus, $u$ takes values in $\Vbar$ and $\Vhat$. 
    
  \item[--] $\mathcal{F}_{\sfa} = \bigcup_{\hat{v}\in\widehat{\mathcal{V}}}\{(v\bar v\hat v): (\bar v\hat v)\in \mathcal{E}_{\bbg}\}$, where $(v\bar{v}\hat{v})$ is the prescription for a face in terms of the three vertices on its boundary. 
\end{description}
\end{defin}

One denotes the set of spin foam atoms by $\sfas$. 

\begin{remark}[{\bf boundary map}]
  \label{rem:bdycor}
  By construction $\alpha:\bbgs\longrightarrow\sfas$ is a bijection.  Moreover, one may define a \emph{boundary map} $\delta:\sfas\longrightarrow\bbgs$, such that for $\sfa$ constructed from $(\bbg,\alpha)$, this map is defined as $\delta(\sfa) = \alpha^{-1}(\sfa) = \bbg$. 
\end{remark}
Thus, as a result of the bijective property of the maps $\alpha$ and $\beta$, the following proposition holds:

\fbox{
  \begin{minipage}[c][][c]{0.97\textwidth}
\begin{proposition}
  \label{prop:bijectionAtom}
  The set $\sfas$ of spin foam atoms is catalogued precisely by the set $\bgs$ of boundary graphs. 
\end{proposition}
\end{minipage}
}

An illustrative example of such a structure is presented in Figure \ref{fig:atom}.

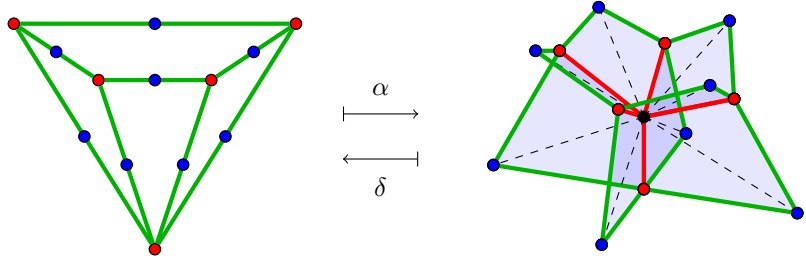
\begin{figure}[htb]
  \centering
  \tikzsetnextfilename{apyr}
  \begin{tikzpicture}

\draw [|->] (-4,0.3) -- node[label=above:$\alpha$] {} (-3,0.3);
\draw [<-|] (-4,-.3) -- node[label=below:$\delta$] {} (-3,-.3);

\begin{scope}[xshift=-6.5cm, scale=1.5]
\node [vb]		(a)	at (0,-1)		{};
\node [vb]		(b)	at (1.25,1)		{}; 
\node [vb]		(c)	at (0.5,0.5) 	{}; 
\node [vb]		(d)	at (-0.5,0.5)	{}; 
\node [vb]		(e)	at (-1.25,1)	{}; 
\foreach \i/\j in {a/b,a/c,a/d,a/e,b/c,b/e,c/d,d/e}{
 \draw [eb] (\i) -- node[vh] {} (\j);
  }
\end{scope}

\begin{scope}[yshift=.5cm, scale=2]
\node [c]		(v)	at (0,-.12)		{};
\node [c]		(1)	at (-0.56,0.32) 	{}; 
\node [c]		(2)	at (-0.17,-0.07)	{}; 
\node [c]		(3)	at (0.6,0) 		{}; 
\node [c]		(4)	at (.14,.37)	{}; 
\node [c]		(5)	at (0,-.6)		{};
\node [c]		(12)	at (-.72,.32)	{};
\node [c]		(23)	at (.44,.09)	{};
\node [c]		(34)	at (.57,.52)	{};
\node [c]		(14)	at (-.3,.61)		{};
\node [c]		(15)	at (-1,-.44)		{};
\node [c]		(25)	at (-.28,-.97)	{};
\node [c]		(35)	at (1.02,-.76)	{};
\node [c]		(45)	at (.28,-.23)	{};
\foreach \i/\j in {1/2,1/4,2/3,3/4,1/5,2/5,3/5,4/5}{
 \path	[f] 	(\i) -- (\i\j) -- (\j) -- (v) -- cycle;
 }
 \foreach \i in {1,2,3,4,5}{
  \draw [e] (\i) -- (v);
  }
\foreach \i/\j in {1/2,1/4,2/3,3/4,1/5,2/5,3/5,4/5}{
 \draw 	[eh]	(v)		-- (\i\j);
 \draw	[eb] 	(\i) node[vb] {} -- (\i\j) node[vh] {};
 \draw 	[eb] 	(\j) node[vb] {} -- (\i\j) node[vh] {};
 }
\draw [e] (3) node[vb] {} -- (v) node[v] {};
%
\end{scope}

\end{tikzpicture}
  \caption{\label{fig:atom} A spin foam atom and its (bisected) boundary graph.}
\end{figure}


\begin{remark}
  \label{rem:retraction}
  A neat alternative to the above construction is given in \cite{\cklp}.  One embeds the graph $\bbg\in\bbgs$ in the bounding 3--sphere of a 4--ball. One performs a radial deformation retraction of this ball to a point, denoted by $v\in \mathcal{V}$.  This retraction restricts to the graph, where one denotes the path traced out by the vertex $\bar{v}\in\overline{\mathcal{V}}$ and $\vh\in\Vh$ as edges $(v\bar{v}),(v\vh)\in \mathcal{E}$ respectively, while the surface traced out by an edge in $(\bar v\hat v)\in \mathcal{E}_{\bbg}$ is interpreted as a face $f=(v\bar v \hat v)\in \mathcal{F}_{\sfa}$.  
  In contrast, the definition given earlier was chosen to be purely combinatorial.
\end{remark}


\subsection{Part 2: bonding atoms to build molecules}
\label{sssec:bonding}

The second step is to describe the procedure by which these atoms bond  to form composite structures, thus completing the unlabelled part of the diagram:  
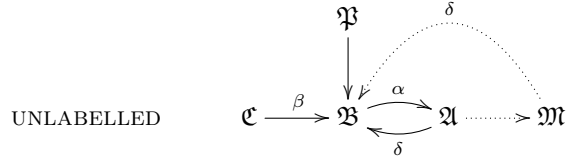
\begin{displaymath}
  \xymatrix{
   & &\bps \ar[d] \\ 
   \textsc{unlabelled} &\bgs \ar[r]^{\beta} &\bbgs \ar@/^/[r]^{\alpha}  &\sfas \ar@/^/[l]^{\delta} \ar@{.>}[r]&\sfrs\ar@/_{3pc}/@{.>}[ll]_{\delta} 
   }
  \end{displaymath}

\begin{defin}[{\bf boundary patch}]
  A \emph{boundary patch} is a double $\bp = \bp_{\vbar} = (\mathcal{V}_{\bp},\mathcal{E}_{\bp})$, where:
  \vspace{-0.2cm}

  \begin{description}
    \item[--] $\mathcal{V}_{\bp} = \{\vbar\} \cup  \Vhat_{\bp}$, $\Vhat_{\bp}\ne\emptyset$;
    \item[--] $\mathcal{E}_{\bp} = \{(\vbar\vhat) : \vhat\in\Vhat_{\bp}\}$ is a multiset of edges where each $(\vbar\vhat)$ occurs at least once and at most twice.
  \end{description} 
\end{defin}
  \begin{remark}
  Boundary patches are useful since they arise as the doubles $\bp_{\vbar}(\bbg) = (\mathcal{V}_{\bar{v}}, \mathcal{E}_{\bar{v}})$, formed as the closure of the star of $\bar{v}\in\overline{\mathcal{V}}$, within $\bbg\in\bbgs$.

  Thus, $\mathcal{V}_{\bar{v}} = \{\bar{v}\} \cup  \{\hat{v}\in\mathcal{V}_{\bbg}:(\bar{v}\hat{v})\in\mathcal{E}_{\bbg}\}$, and $\mathcal{E}_{\bar{v}} = \{e = (\bar{v}\hat{v})\in\mathcal{E}_{\bbg}\}$. In words, a boundary patch $\bp_{\vbar}(\bbg)$ is a graph containing $\bar{v}$ itself, all boundary edges containing $\bar{v}$ (the result of the star operation),  as well as the endpoints of these edges (the result of the closure operation). A simple example is depicted in Figure \ref{fig:star}. 

  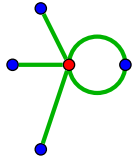
\begin{figure}[htb]
  \centering
  \tikzsetnextfilename{patch}
  \begin{tikzpicture}[scale=1.5]
  \draw [eb] (0.75,0.5) circle (.25cm);
  \node [vb]		(vb)	at (.5,.5) 	 	{}; 
  \node [vh]		(1)	at (0,.5)		{};
  \node [vh]		(2)	at (.25,1)		{}; 
  \node [vh]		(3)	at (.25,-.25)	{}; 
  \node [vh]		(4)	at (1,.5)		{};
  \foreach \i/\j in {vb/1,vb/2,vb/3}{
  \draw [eb] (\i) -- (\j);
  }
  \end{tikzpicture}
  \caption{\label{fig:star} A boundary patch.}
\end{figure}

\end{remark}
The set of boundary patches is denoted by $\bps$.
\begin{remark}[{\bf generators}]
  \label{rem:generator}
  For some subset of patches, $\bps_{\sub}\subseteq\bps$,  the set of graphs \emph{generated} by $\bps_{\sub}$, denoted $\sigma(\bps_{\sub})$, is the set of all possible graphs that are composed only of patches from $\bps_{\sub}$.
\end{remark}

Then, it is quite clear that:
\begin{proposition}
$\bbgs = \sigma(\bps)$.
\end{proposition}
\begin{remark}[{\bf bondable}]
  \label{rem:bondable}
  Two patches, $\bp_{\bar{v}_1}(\bbg_1)$ and $\bp_{\bar{v}_2}(\bbg_2)$,  whether or not $\bbg_1$ and $\bbg_2$ are distinct,  are said to be \emph{bondable}, if $|\mathcal{V}_{\bar{v}_1}|= |\mathcal{V}_{\bar{v}_2}|$ and $|\mathcal{E}_{\bar{v}_1}|= |\mathcal{E}_{\bar{v}_2}|$ (and thus, they have the same number of loops).  
\end{remark}

\begin{defin}[{\bf bonding map}]
  \label{def:bonding}
  A \emph{bonding map}, $\gamma:\bp_{\bar{v}_1}(\bbg_1)\longrightarrow \bp_{\bar{v}_2}(\bbg_2)$, is a map identifying, elementwise, two bondable patches such that:
  \begin{equation}
    \bar{v}_1    \longrightarrow    \bar{v}_2\,,\quad\quad
    \mathcal{V}_{\bar{v}_1} - \{\bar{v}_1\}   \longrightarrow  \mathcal{V}_{\bar{v}_2} - \{\bar{v}_2\}\,,\quad\quad
  \mathcal{E}_{\bar{v}_1}  \longrightarrow  \mathcal{E}_{\bar{v}_2}
  \end{equation}
  with the compatibility condition that for each identified pair $\hat{v}_1\in\mathcal{V}_{\bar{v}_1}\longrightarrow\hat{v}_2\in\mathcal{V}_{\bar{v}_{2}}$, then $e_1 = (\bar{v}_1\hat{v}_1)\in\mathcal{E}_{\bar{v}_1}\longrightarrow e_2 = (\bar{v}_2\hat{v}_2)\in\mathcal{E}_{\bar{v}_2}$.  
\end{defin}
A simple example is illustrated in Figure \ref{fig:bonding}. 
\begin{remark}
The compatibility condition ensures that loops are bonded to loops. In principle, slightly more general gluing maps can be incorporated within the group field theory framework, corresponding to loop edges bonding to non-loop edges.  However, these gluings are absent from the loop quantum gravity and spin foam theories.  Thus, there is no motivation to include them here.
\end{remark}
\begin{remark}
Certainly, for two bondable patches, there are many bonding maps that satisfy the compatibility condition.   However, all may be obtained from a given one by applying compatible permutations to the sets $\mathcal{V}_{\bar{v}_1}$ and $\mathcal{E}_{\bar{v}_1}$.   
\end{remark}

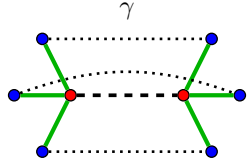
\begin{figure}[htb]
  \centering
  \tikzsetnextfilename{bonding}
 \begin{tikzpicture}[scale=1.5]
\node [vb]		(a)	at (-0.5,0)		{}; 
\node [vh]		(1)	at (-.75,.5)		{};
\node [vh]		(2)	at (-.75,-.5)	{}; 
\node [vh]		(3)	at (-1,0)		{}; 
\node [vb]		(b)	at (0.5,0)		{}; 
\node [vh]		(4)	at (.75,.5)		{};
\node [vh]		(5)	at (.75,-.5)		{}; 
\node [vh]		(6)	at (1,0)		{}; 
\foreach \i/\j in {a/1,a/2,a/3,b/4,b/5,b/6}{
 \draw [eb] (\i) --  (\j);
  }
\path	(1) 	edge [bh] node[label=above:$\gamma$] {} (4)
	(2)	edge [bh]		 		(5)
	(3)	edge [bh, bend left=20]	(6)
	(a)	edge [bb]				(b);
\end{tikzpicture}
  \caption{\label{fig:bonding} A bonding map $\gamma$ identifying two bondable patches.} 
\end{figure}

\begin{defin}[{\bf spin foam molecule}]
  \label{def:molecule}
  A \emph{spin foam molecule} is a triple, $\sfr = (\mathcal{V}_{\sfr}, \mathcal{E}_{\sfr}, \mathcal{F}_{\sfr})$,  constructed from a collection of spin foam atoms quotiented by a set of bonding maps.
\end{defin}

\begin{remark}[{\bf bonding example}]
  \label{rem:bondingexample}
  It is worth considering the simple example of two spin foam atoms $\sfa_1$ and $\sfa_2$, with respective bisected boundary graphs $\bbg_1 = \alpha^{-1}(\sfa_1)$ and $\bbg_2 = \alpha^{-1}(\sfa_2)$ and two bondable patches $\bp_{\bar{v}_1}(\bbg_1)$ and $\bp_{\bar{v}_2}(\bbg_2)$. Quotienting the pair $\sfa_1$, $\sfa_2$ by a bonding map $\gamma:\bp_{\bar{v}_1}(\bbg_1)\longrightarrow \bp_{\bar{v}_2}(\bbg_2)$ results in a spin foam molecule $\sfr  \equiv \sharp_\gamma\,\{\sfa_1,\sfa_2\}$:
  \begin{equation}  
    \mathcal{V}_{\sfr} = \sharp_\gamma\,\{\mathcal{V}_{\sfa_1},\mathcal{V}_{\sfa_2}\}\;,
     \quad\quad
     \mathcal{E}_{\sfr}=\sharp_\gamma\,\{\mathcal{E}_{\sfa_1},\mathcal{E}_{\sfa_2}\}\;,
     \quad\quad
     \mathcal{F}_{\sfr}=\sharp_\gamma\,\{\mathcal{F}_{\sfa_1},\mathcal{F}_{\sfa_2}\}\;,
  \end{equation}
where $\sharp_\gamma$ denotes the union of the relevant sets \emph{after} the identification of the elements of $(\mathcal{V}_{\bar{v}_1}\subset \mathcal{V}_{\sfa_1}, \mathcal{E}_{\bar{v}_1}\subset\mathcal{E}_{\sfa_1})$ with those of $(\mathcal{V}_{\bar{v}_2}\subset \mathcal{V}_{\sfa_2}, \mathcal{E}_{\bar{v}_2}\subset\mathcal{E}_{\sfa_2})$.
Thus, there exists still structure at the interface between the two bonded atoms, specifically, $\bp_{\bar{v}}(\bbg) \equiv \bp_{\bar{v}_1}(\bbg_1) = \bp_{\bar{v}_2}(\bbg_2)$. A realization of the above example is presented in Figure \ref{fig:molecule}. 

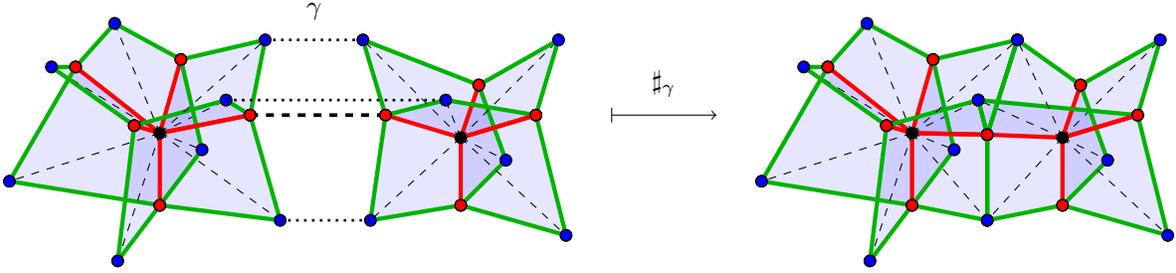
\begin{figure}[htb]
  \centering
  \tikzsetnextfilename{molecule}
\begin{tikzpicture}[scale=2]
\path (.6,0) edge [bb]	(1.5,0);

\node [c]		(v)	at (0,-.12)		{};
\node [c]		(6)	at (-0.56,0.32) 	{}; 
\node [c]		(7)	at (-0.17,-0.07)	{}; 
\node [c]		(3)	at (0.6,0) 		{}; 
\node [c]		(8)	at (.14,.37)	{}; 
\node [c]		(5)	at (0,-.6)		{};
\node [c]		(67)	at (-.72,.32)	{};
\node [c]		(37)	at (.44,.1)		{};
\node [c]		(38)	at (.7,.5)		{};
\node [c]		(68)	at (-.3,.61)		{};
\node [c]		(56)	at (-1,-.44)		{};
\node [c]		(57)	at (-.28,-.97)	{};
\node [c]		(35)	at (.8,-.7)		{};
\node [c]		(58)	at (.28,-.23)	{};
\foreach \i/\j in {6/7,6/8,3/7,3/8,5/6,5/7,3/5,5/8}{
 \path	[f] 	(\i) -- (\i\j) -- (\j) -- (v) -- cycle;
 }
 \foreach \i in {6,7,3,8,5}{
  \draw [e] (\i) -- (v);
  }
\foreach \i/\j in {6/7,6/8,3/7,3/8,5/6,5/7,3/5,5/8}{
 \draw 	[eh]	(v)	-- (\i\j);
 \draw	[eb] 	(\i) node[vb] {} -- (\i\j) node[vh] {};
 \draw 	[eb] 	(\j) node[vb] {} -- (\i\j) node[vh] {};
 }
\draw [e] (3) node[vb] {} -- (v) node[v] {};

\begin{scope}[xshift=2cm]
\node [c]		(v)	at (0,-.15)		{};
\node [c]		(1)	at (.5,0)	 	{}; 
\node [c]		(2)	at (.12,.2)		{}; 
\node [c]		(3)	at (-.5,0)		{}; 
\node [c]		(4)	at (0,-.6)		{};
\node [c]  		(12)	at (.65,.5)		{};
\node [c]	 	(13)	at (-.1,.1)		{};
\node [c]		(14)	at (.7,-.8)	{};
\node [c]		(23)	at (-.65,.5)		{};
\node [c] 		(24)	at (.3,-.3)		{};
\node [c]		(34)	at (-.6,-.7)		{};
\foreach \i/\j in {1/2,1/3,1/4,2/3,2/4,3/4}{
 \path	[f] 	(\i) -- (\i\j) -- (\j) -- (v) -- cycle;
 }
 \foreach \i in {1,2,3,4}{
  \draw [e] (\i) -- (v);
  }
\foreach \i/\j in {2/4,1/2,1/3,1/4,2/3,3/4}{
 \draw 	[eh]	(v)		-- (\i\j);
 \draw	[eb] 	(\i) node[vb] {} -- (\i\j) node[vh] {};
 \draw 	[eb] 	(\j) node[vb] {} -- (\i\j) node[vh] {};
 }
\draw [e] (1) node[vb] {} -- (v) node[v] {};
\end{scope}

\path	(38)	edge [bh] node[label=above:$\gamma$] {} (23)
	(37)	edge [bh]		(13)
	(35)	edge [bh] 		(34);

\draw [|->] (3,0) -- node[label=above:$\sharp_{\gamma}$] {} (3.7,0);

\begin{scope}[xshift=5cm]
\node [c]		(w)	at (0,-.12)		{};
\node [c]		(6)	at (-0.56,0.32) 	{}; 
\node [c]		(7)	at (-0.17,-0.07)	{}; 
\node [c]		(8)	at (.14,.37)	{}; 
\node [c]		(5)	at (0,-.6)		{};
\node [c]		(67)	at (-.72,.32)	{};
\node [c]		(37)	at (.44,.1)		{};
\node [c]	 	(13)	at (.44,.1)		{};
\node [c]		(38)	at (.7,.5)		{};
\node [c]		(23)	at (.7,.5)		{};
\node [c]		(68)	at (-.3,.61)		{};
\node [c]		(56)	at (-1,-.44)		{};
\node [c]		(57)	at (-.28,-.97)	{};
\node [c]		(35)	at (.5,-.7)		{};
\node [c]		(34)	at (.5,-.7)		{};
\node [c]		(58)	at (.28,-.23)	{};
\begin{scope}[xshift=1cm]
\node [c]		(v)	at (0,-.15)		{};
\node [c]		(1)	at (.5,0)	 	{}; 
\node [c]		(2)	at (.12,.2)		{}; 
\node [c]		(3)	at (-.5,-.13)	{}; 
\node [c]		(4)	at (0,-.6)		{};
\node [c]  		(12)	at (.65,.5)		{};
\node [c]		(14)	at (.7,-.8)		{};
\node [c] 		(24)	at (.3,-.3)		{};
\end{scope}
\foreach \i/\j in {6/7,6/8,3/7,3/8,5/6,5/7,3/5,5/8}{
 \path	[f] 	(\i) -- (\i\j) -- (\j) -- (w) -- cycle;
 }
 \foreach \i in {6,7,3,8,5}{
  \draw [e] (\i) -- (w);
  }
\foreach \i/\j in {6/7,6/8,3/7,3/8,5/6,5/7,3/5,5/8}{
 \draw 	[eh]	(w)	-- (\i\j);
 \draw	[eb] 	(\i) node[vb] {} -- (\i\j) node[vh] {};
 \draw 	[eb] 	(\j) node[vb] {} -- (\i\j) node[vh] {};
 }
\draw [e] (3) node[vb] {} -- (w) node[v] {};

\foreach \i/\j in {1/2,1/3,1/4,2/3,2/4,3/4}{
 \path	[f] 	(\i) -- (\i\j) -- (\j) -- (v) -- cycle;
 }
 \foreach \i in {1,2,3,4}{
  \draw [e] (\i) -- (v);
  }
\foreach \i/\j in {2/4,1/2,1/3,1/4,2/3,3/4}{
 \draw 	[eh]	(v)		-- (\i\j);
 \draw	[eb] 	(\i) node[vb] {} -- (\i\j) node[vh] {};
 \draw 	[eb] 	(\j) node[vb] {} -- (\i\j) node[vh] {};
 }
\draw [e] (1) node[vb] {} -- (v) node[v] {};

\end{scope}

\end{tikzpicture} 
\caption{\label{fig:molecule} The bonding $\sharp_{\gamma}$ of two atoms along an identification of patches $\gamma$.}
\end{figure}

\end{remark}

\begin{remark}[{\bf molecule boundary}]
  \label{rem:moleculeboundary}
  The boundary map $\delta$ can be extended to the spin foam molecule $\sfr = \sharp_{\{\gamma\}_I}\,\{\sfa\}_{J}$, where $I,J$ are index sets.  $\delta(\sfr)$ is identified as the subset of constituent boundary graphs, $\cup_{j\in J}\delta(\sfa_j)$ formed from the edges that remain unbonded, along with their vertices. In symbols:
  \begin{equation}
    \label{eq:moleculeboundary}
    \mathcal{E}_{\delta(\sfr)} = \bigcup_{j\in J}\mathcal{E}_{\delta(\sfa_j)} - \bigcup_{i\in I}\mathcal{E}_{\gamma_{i}}\;,\quad\quad
    \mathcal{V}_{\delta(\sfr)} = \{ \bar{v},\hat{v}: (\bar{v}\hat{v})\in\mathcal{E}_{\delta(\sfr)}\}\;. 
  \end{equation}
  where $\gamma:\bp_{\bar{v}_{i_1}}(\bbg_{i_1})\longrightarrow \bp_{\bar{v}_{i_2}}(\bbg_{i_2})$ and $\mathcal{E}_{\gamma_i} = \mathcal{E}_{\bar{v}_{i_1}}\cup\mathcal{E}_{\bar{v}_{i_2}}$.
  In general, $\delta(\sfr)$ need not be connected, but it will be the disjoint union of some set of bisected boundary graphs.  Moreover, these boundary graphs will very rarely coincide with the boundary graphs associated to any of the constituent atoms.

If a spin foam molecule $\sfr$ has a non-vanishing boundary $\delta(\sfr)\ne\emptyset$, one might also term it as a \emph{spin foam radical}.  

On the other hand, if $\delta(\sfr)=\emptyset$, $\sfr$ can be called a \emph{saturated} or \emph{closed} spin foam molecule.
\end{remark}


\subsection{Part 3: specifying to loopless, regular and simplicial structures}
\label{sssec:special}

There are few obvious restrictions one can have on graphs, atoms and molecules which will become important later.
These are loopless and regular structures as well as the restriction to a single type of spin foam atom which we shall call simplicial.
All of them mirror exactly the structure of the most general case. For example,  loopless structures are related in the following way:
\begin{displaymath}
  \xymatrix{
   \textsc{loopless} &\bgs_{\loopless} \ar[r]^{\beta} &\bbgs_{\loopless} \ar@/^/[r]^{\alpha}  &\sfas_{\loopless} \ar@/^/[l]^{\delta} \ar@{.>}[r]&\sfrs_{\loopless}\ar@/_{3pc}/@{.>}[ll]_{\delta} 
   }
  \end{displaymath}

\begin{defin}[{\bf loopless structures}]
\label{def:loopless}
Loopless structures are specified by:

A \emph{loopless boundary graph}, $\bg\in\bgs_{\loopless}$, is a $\bg=(\Vb,\Eb)\in\bgs$ without edges from any vertex $\vb\in\Vb$ to itself, that is for every $\vb\in\Vb$:  $(\vb\vb)\not\in\Eb$. 

Their images under the bisection map $\beta$ and thereafter the bulk map $\alpha$ straightforwardly define \emph{loopless bisected boundary graphs} $\bbgs_{\loopless}$ and \emph{loopless atoms} $\sfas_{\loopless}$, respectively.

For a graph in $\bbgs_{\loopless}$, all of its patches are obviously loopless. In fact, the \emph{loopless patches} are uniquely specified by $n$, the number of edges. 
Therefore, we call it an \emph{n--patch}, $\bp_n$, and we have that $\bps_{\loopless}=\bigcup_{\copies=1}^{\infty}\{\bp_{\copies}\}$. 
Moreover, $\bbgs_{\loopless} = \sigma(\bps_{\loopless})$, the loopless graphs are generated by loopless patches. 

Through the bonding maps $\gamma$, one constructs \emph{loopless spin foam molecules} $\sfms_{\loopless}$.

\end{defin}

\begin{remark}[{\bf loopless molecules}]
\label{rem:LMolecules}
Loopless molecules are indeed the most natural class of 2--dimensional combinatorial objects, since they are triangulations of a certain kind of abstract (i.e. combinatorial) polyhedral 2--complexes.
We provide the definition of abstract polyhedral complexes in the appendix and prove their precise relation to $\sfrs_{\loopless}$ in Proposition \ref{prop:LMolecules}.

This also means that arbitrary spin foam molecules $\sfrs$,  
do not correspond naturally to 2--complexes in a combinatorial sense, exactly because they are containing loops. Nevertheless, from the quantum gravity viewpoint, these structures are necessary to provide dynamics for the most general graph, upon which \lqg\ states are based. Moreover, abstract polyhedral complexes can be generalized to match $\sfrs$ (Proposition \ref{prop:Molecules}).
\end{remark}

Another important restriction concerns the valency of boundary graph vertices: 
\begin{defin}[{\bf $n$--regular structures}]
An \emph{$\copies$-regular boundary graph} $\bg\in\bgs_{\copies}$ is a double $\bg=(\Vb,\Eb)\in\bgs$, for which every vertex $\vb\in\Vb$ is $\copies$--valent.  In other words, there are exactly $\copies$ edges $(\vb\vh)\in\Eb$ containing $\vb$. Analogous to Definition \ref{def:loopless}, the notion of their bisected counterparts $\bbgs_{\copies}$, the related \emph{$\copies$--regular atoms} $\sfas_{\copies}$, as well as $\copies$--regular molecules $\sfrs_{\copies}$, is straightforward.
\end{defin}
\begin{remark}[{\bf $\copies$--regular and loopless}]
	Combining these restrictions, one arrives at much simpler sets of graphs $\bbgs_{n,\loopless}$, atoms $\sfas_{n,\loopless}$ and molecules $\sfms_{n,\loopless}$. In particular, $\bbgs_{n,\loopless} = \sigma(\bp_n)$, a single patch generates the whole set. 
Since the structure of a \gft\ field is determined by a patch, these structures will play a role in single field \gft s, explained in detail in section \ref{sec:gft}.
\end{remark}

Nevertheless, the simplest \gft\ is not only defined in terms of one field, but also only one interaction term of simplicial type. This motivates the following definition:

\begin{defin}[{\bf $\copies$--simplicial molecules}]
The set of \emph{$\copies$--simplicial molecules} $\sfrs_{\copies,\primitive}$ consists of all molecules, which are bondings of the single spin foam atom $\sfa_{\copies,\primitive}$ obtained from the complete graph with $\copies +1$ vertices $K_{\copies+1}$,
\[
\sfa_{\copies,\primitive}:=\alpha(\bbg_{\copies,\primitive}):=\alpha(\beta(\bg_{\copies,\primitive})):=\alpha(\beta(K_{\copies+1})).
\]
\end{defin}

A complete graph is displayed in Figure \ref{fig:complete}.
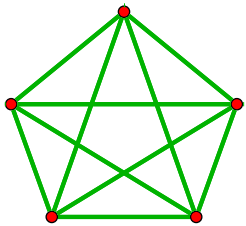
\begin{figure}[htb]
  \centering
    \tikzsetnextfilename{complete}
\begin{tikzpicture}[scale=1.5]
\node [vb]		(a)	at (0,0.82)		{};
\node [vb]		(b)	at (-1,0) 		{}; 
\node [vb]		(c)	at (-0.64,-1) 	{}; 
\node [vb]		(d)	at (0.64,-1)	{}; 
\node [vb]		(e)	at (1,0)		{}; 
\foreach \i in {a,b,c,d,e}{
  \foreach \j in {a,b,c,d,e}{
    \draw [eb] (\i) -- (\j);
    }
  }
\end{tikzpicture}
  \caption{\label{fig:complete} The complete graph over $n+1$ vertices (n = 4).}
\end{figure}

\begin{remark}[{\bf clarification on the notion \lq\simplicial\rq}]
	It must be emphasized that the special class of $\copies$--simplicial molecules $\sfrs_{\copies,\primitive}\In\sfrs_{n,\loopless}\In\sfrs_{\loopless}$, like all other loopless  molecules, are polyhedral 2--complexes.
We call them simplicial because each spin foam atom in itself can be canonically understood as the dual 2--skeleton of an $n$--simplex (cf. Figure \ref{fig:dualtet} and \ref{fig:hassedual}, and the Appendix).
But this can be done only locally, since
it has been proven in \cite{\cguraulost} that not every simplicial spin foam molecule (referred to as \gft--gluing therein) can be assigned a simplicial complex, for which the molecule arises as the dual 2--skeleton.
\end{remark}

\begin{remark}
  As mentioned at the outset of this section, the construction presented here is effectively very similar to the \emph{operator spin network} approach devised in \cite{\cklp}, which in turn is based upon the language of \emph{operator spin foams} \cite{\cOperator}.

  For clarity, it is worth setting up a small dictionary between the two descriptions. To begin, loopless boundary patches  correspond to \emph{squids}. Then \emph{squid graphs} are defined as gluings of such patches where gluing vertices of a patch to itself is allowed. Thus, these are what we call bisected boundary graphs. Our definition of patches including loops in general is necessary from a GFT perspective.
  Moreover, the set of squid graphs considered in \cite{\cklp} corresponds to that subset of boundary graphs without 1--valent vertices $\bar{v}\in\overline{\mathcal{V}}$. However, this is a choice and is easily generalized. 

Squid graphs encode \emph{1--vertex spin foams} through a retraction, which was mentioned above in Remark \ref{rem:retraction} (in \cite{\cklp} also a more combinatorial definition is given), just as boundary graphs encode spin foam atoms.  After that,  1--vertex spin foams are glued together by identifying pairs of squids, just like boundary patches are bonded during the construction of spin foam molecules.

\end{remark}


\subsection{Part 4: labelled structures}
\label{sssec:virtual}

The set of spin foam atoms $\mathfrak{A}$ is efficiently catalogued by their boundary graphs $\bgs$. However, this is a large collection of objects and thus motivates one to seek out sub--atomic building blocks that are more concisely presented but can nevertheless resemble all of $\bgs$.  

This search is divided into two stages.  This first stage examines the boundary graphs in terms of their constituent boundary patches. The set of such patches is very large. 
Thus, the first stage will focus on manufacturing a manageable\footnote{A set with a (small) finite number of elements.} set of patches, with which, none the less, one may encode all the boundary graphs in $\bgs$.  

Having accomplished this, the next stage examines the boundary graphs from the perspective of generating them by bonding boundary graphs from a more manageable set. 


To set the stage, in this part we introduce labelled structures:
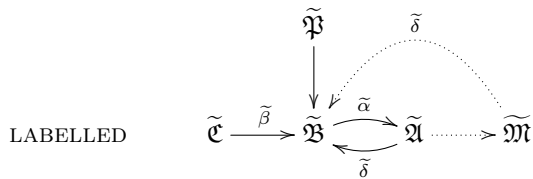
\begin{displaymath}
  \xymatrix{
   & &\bpst \ar[d] \\ 
   \textsc{labelled} &\bgst \ar[r]^{\widetilde{\beta}} &\bbgst \ar@/^/[r]^{\widetilde{\alpha}}  &\sfast \ar@/^/[l]^{\widetilde{\delta}} \ar@{.>}[r]&\sfrst\ar@/_{3pc}/@{.>}[ll]_{\widetilde{\delta}} 
   }
  \end{displaymath}

\begin{defin}[{\bf labelled boundary graph}]
  A \emph{labelled boundary graph}, $\bgt$ is a boundary graph augmented with a label for each edge drawn from the set $\{real,virtual\}$. 
\end{defin}
The set of such graphs is denoted by $\bgst$ and is much larger than the set $\bgs$, since for a graph $\bg = (\overline{\mathcal{V}},\overline{\mathcal{E}})\in\bgs$, there are $2^{|\overline{\mathcal{E}}|}$ labelled counterparts in $\bgst$.
\begin{remark}[{\bf labelled structures}]
  \label{rem:generalization}
  There are some trivial generalizations:
  
  \vspace{-0.2cm}

  \begin{description}
    \item[--] The \emph{labelled bisected boundary graphs}, denoted by $\bbgt\in\bbgst$, are obtained using a bisection map $\widetilde{\beta}$ that maintains edge labelling. Thus, if $(\bar{v}_1\bar{v}_2)\in\bgt$ is a real (virtual) edge, then $\{\hat{v},(\bar{v}_1\hat{v}),(\bar{v}_2\hat{v})\}\subset\bbgt=\widetilde{\beta}(\bgt)$ is a real (resp.\ virtual) subset, where $\hat{v}$ is the bisecting vertex.
    \item[--] The \emph{labelled spin foam atoms}, denoted by $\sfat\in\sfast$, are obtained using a bulk map $\widetilde{\alpha}$, such that if $(\bar{v}\hat{v})$ is real  (virtual), then so is $(v\bar{v}\hat{v})$. In other words,  the faces inherit their label from the boundary $\delta(\sfat)=\bbgt$, where $\widetilde{\delta} = \widetilde{\alpha}^{-1}$.
    \item[--] The \emph{labelled boundary patches}, denoted by $\bpt\in\bpst$, are bonded pairwise using bonding maps $\widetilde{\gamma}$ that ensure real (virtual) elements bonded to real (resp.\ virtual) elements.
    \item[--] With these bonding maps, \emph{labelled spin foam molecules} $\sfrt\in\sfrst$ follow immediately. 
  \end{description}
\end{remark}


\subsection{Part 5: molecules from labelled, $n$--regular, loopless structures}
\label{sssec:labelledloopless}

This part focusses on defining a projection $\pi$ which relates labelled graphs to unlabelled ones by contracting and deleting the virtual edges, as well as its restriction to the labelled, $n$--regular, loopless structures,
$\pi_{n,\loopless}$, $\Pi_{n,\loopless}$ and $\Pi_{n,\looplessdually}$, which can be shown to still map surjectively to arbitrary graphs and molecules:
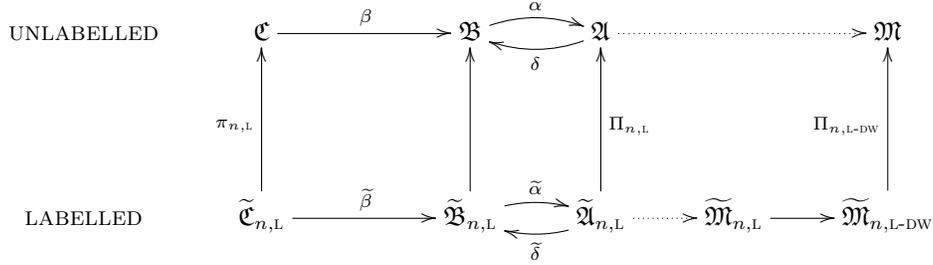
\begin{displaymath}
  \xymatrix{
   \textsc{unlabelled} &\bgs \ar[rr]^{\beta}& &\bbgs \ar@/^/[r]^{\alpha}  &\sfas \ar@/^/
   [l]^{\delta} \ar@{.>}[rr]&&\sfrs \\
   &\\
   \textsc{labelled} &  \bgst_{n,\loopless} \ar[rr]^{\widetilde{\beta}} \ar[uu]^{\pi_{n,\loopless}} & &\bbgst_{n,\loopless}  \ar@/^/[r]^{\widetilde{\alpha}}\ar[uu]  
   &\sfast_{n,\loopless} \ar@/^/[l]^{\widetilde{\delta}}\ar[uu]_{\Pi_{n,\loopless}}\ar@{.>}[r] &\sfrst_{n,\loopless}\ar[r] &\sfrst_{n,\looplessdually}\ar[uu]^{\Pi_{n,\looplessdually}}
 }
  \end{displaymath}

    One can naturally identify the unlabelled boundary graphs $\bgs$ with the subset of labelled graphs that possess only real edges $\bgst_{real}\subset\bgst$. However, one would like to go further and utilize the unlabelled graphs to mark classes of labelled graphs.  From another aspect, one would think of this class of labelled graphs as encoding an underlying (unlabelled) subgraph $\bg\in\bgs$. 

    To uncover this structure, one defines certain moves on the set of labelled graphs:
\begin{defin}[{\bf reduction moves}]
  \label{def:moves}
  Given a graph $\tilde{\bg}\in\bgst$, there are two moves that reduce the virtual edges of the graph:

  \vspace{-0.2cm}

  \begin{description}
    \item[--] given two vertices, $\bar{v}_1$ and $\bar{v}_2$, such that $(\bar{v}_1\bar{v}_2)$ is a virtual edge of $\tilde{\bg}$, a \emph{contraction move}, removes this virtual edge and identifies the vertices $\bar{v}_1$ and $\bar{v}_2$;
    \item[--] given a vertex $\bar{v}$ such that $(\bar{v}\bar{v})$ is a virtual loop, a \emph{deletion move} is simply the removal of this edge.
  \end{description}

  \vspace{-0.2cm}

  These inspire two counter moves:

  \vspace{-0.2cm}

  \begin{description}
    \item[--] given a vertex $\bar{v}$, an \emph{expansion move} partitions the edges, incident at $\bar{v}$, into two subsets. In each subset, $\bar{v}$ is replaced by two new vertices $\bar{v}_1$ and $\bar{v}_2$, respectively, and a virtual edge $(\bar{v}_1\bar{v}_2)$ is added to the graph.\footnote{There is subtlety for loops, in that both ends are incident at $\bar{v}$ and may (or may not) be separated by the partition.} 
    \item[--] given a vertex $\bar{v}$, a \emph{creation move} adds a virtual loop to the graph at $\bar{v}$.
  \end{description}
  These moves are illustrated in Figure \ref{fig:moves}.

  \begin{figure}[htb]
\centering
 \tikzsetnextfilename{moves}
\begin{tikzpicture}[scale=1.5]

\draw [eb, dotted] (0,.25) circle (0.25cm);
\node [vb]  (v) 	at (0,0) {};
\draw [eb]  (v) -- (-.7,0);
\draw [eb]  (v) -- (-.5,-.5);
\draw [eb]  (v) -- (.5,-.5);
\draw [eb]  (v) -- (.7,0);

\draw [|->] (1.5,0) -- node[label=above:$\pi$] {} (2,0);

\begin{scope}[xshift=3.5cm]
\node [vb]  (v) 	at (0,0) {};
\draw [eb]  (v) -- (-.7,0);
\draw [eb]  (v) -- (-.5,-.5);
\draw [eb]  (v) -- (.5,-.5);
\draw [eb]  (v) -- (.7,0);
\end{scope}

\begin{scope}[xshift=0cm,yshift=1.5cm]
\node [vb]  (v) 	at (0,0) {};
\node [vb]  (v2) 	at (.5,0) {};
\draw [eb, dotted] (v) -- (v2);
\draw [eb]  (v) -- (-.7,0);
\draw [eb]  (v) -- (-.5,-.5);
\draw [eb]  (v) -- (-.5,.5);
\draw [eb]  (v2) -- (1,-.5);
\draw [eb]  (v2) -- (1,.5);

\draw [|->] (1.5,0) -- node[label=above:$\pi$] {} (2,0);
\end{scope}

\begin{scope}[xshift=3.5cm,yshift=1.5cm]
\node [vb]  (v) 	at (0,0) {};
\draw [eb]  (v) -- (-.7,0);
\draw [eb]  (v) -- (-.5,-.5);
\draw [eb]  (v) -- (.5,-.5);
\draw [eb]  (v) -- (-.5,.5);
\draw [eb]  (v) -- (.5,.5);
\end{scope}

\end{tikzpicture}
\caption{\label{fig:moves} Contraction/expansion and deletion/creation moves.}
\end{figure}
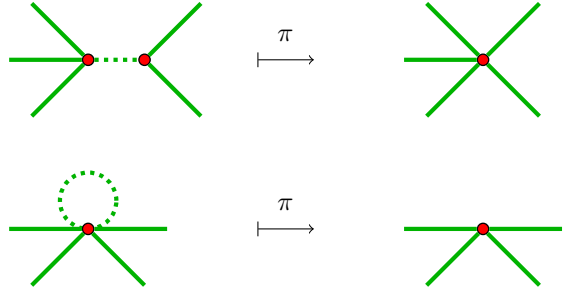

\end{defin}
\begin{remark}[{\bf projector}]
  \label{rem:projector}
This allows one to define a projection $\pi:\bgst\longrightarrow\bgs$, which captures the complete removal of virtual edges through contraction and deletion. It is well--defined, in the sense that contraction and deletion eventually map to an element of $\bgs$ (that is, the graph remains connected) and the element $\bg\in\bgs$ acquired from $\bgt\in\bgst$ is independent of the sequence of contraction and deletion moves used to reduce the graph. In turn, this means that the $\pi^{-1}(\bg)$ partition $\bgst$ into classes. 
\end{remark}

  In fact, one is interested only in the $n$--regular ($n>2$) subset $\bgst_n$.
One denotes the restriction of $\pi$ to these subsets as $\pi_n$. Note that the $\pi_n$ are no longer projections, since $\pi_n(\bgt)$ with $\bgt\in\bgst_{n}$ need no longer be $n$--valent. 

\begin{proposition}[{\bf surjections}]
  \label{prop:surjectors}
  The maps $\pi_n$ have the following properties:
  \vspace{-0.2cm}
  \begin{description}
    \item[--]  The map $\pi_n:\bgst_{n}\longrightarrow \bgs$ is surjective, for $n$ odd.
    \item[--]  The map $\pi_n:\bgst_{n}\longrightarrow \bgs_{even}\subset\bgs$, is surjective for $n$ even, where $\bgs_{even}$ is the subset of boundary graphs with only even--valent vertices. 
  \end{description}
\end{proposition}
\begin{proof}
    First, one proves the results for the lowest values of $n$. 
    For $n=3$, consider a graph $\bg\in\bgs$ and say it possesses an $m$--valent vertex ($m>3$).  Then, one may expand such a vertex to a sequence of 3--valent vertices joined by a string of virtual edges.  For a 2--valent vertex, one first creates a virtual loop and then expand the resulting 4--valent vertex. For a 1--valent vertex, one simply creates a virtual loop. See Figure \ref{fig:3decomp} for an illustration of these three cases processes.

    \begin{figure}[htb]
      \centering
      \tikzsetnextfilename{3val}

\begin{tikzpicture}[scale=1.5]

\draw [eb] (0.7,0.0)-- (0.0,0.0);
\draw [eb] (-0.7,0.0)-- (0.0,0.0);
\draw [eb] (0.0,0.0)-- (-0.5,-0.5);
\draw [eb] (0.5,-0.5)-- (0.0,0.0);
\draw [eb] (0.0,0.25) circle (0.25cm);

\draw [fill=red] (0.0,0.0) circle (1.5pt);

\draw [<-|] (1.5,0) -- node[label=above:$\pi_3$] {} (2,0);

\begin{scope}[xshift=4cm]

\draw [dotted, eb] (-0.5,0.0)-- (0.0,0.0);
\draw [dotted, eb] (1.0,0.0)-- (0.0,0.0);
\draw [eb] (-1.0,-0.5)-- (-0.5,0.0);
\draw [eb] (-0.5,0.0)-- (-1.2,0.0);
\draw [eb] (1.7,0.0)-- (1.0,0.0);
\draw [eb] (1.5,-0.5)-- (1.0,0.0);
\draw [shift={(0.25,0.1875)},eb]  plot[domain=-0.64:3.785,variable=\t]({1.0*0.3125*cos(\t r)+-0.0*0.3125*sin(\t r)},{0.0*0.3125*cos(\t r)+1.0*0.3125*sin(\t r)});
\begin{scriptsize}
\draw [fill=red] (-0.5,0.0) circle (1.5pt);
\draw [fill=red] (0.0,0.0) circle (1.5pt);
\draw [fill=red] (1.0,0.0) circle (1.5pt);
\draw [fill=red] (0.5,0.0) circle (1.5pt);
\draw [fill=red] (-0.5,0.0) circle (1.5pt);
\draw [fill=red] (1.0,0.0) circle (1.5pt);
\end{scriptsize}
\end{scope}
\begin{scope}[xshift=0cm,yshift=-1.5cm]
\draw [eb] (0.7,0.0)-- (0.0,0.0);
\draw [eb] (-0.7,0.0)-- (0.0,0.0);
\draw [fill=red] (0.0,0.0) circle (1.5pt);

\draw [<-|] (1.5,0) -- node[label=above:$\pi_3$] {} (2,0);
\end{scope}

%

\begin{scope}[xshift=3.5cm,yshift=-1.5cm]
\draw [eb] (-0.7,0.0)-- (0.0,0.0);
\draw [dotted, eb] (0.5,0.0)-- (0.0,0.0);
\draw [eb] (1.2,0.0)-- (0.5,0.0);
\draw [dotted, shift={(0.25,0.1875)},eb]  plot[domain=-0.64:3.785,variable=\t]({1.0*0.3125*cos(\t r)+-0.0*0.3125*sin(\t r)},{0.0*0.3125*cos(\t r)+1.0*0.3125*sin(\t r)});
\draw [fill=red] (0.5,0.0) circle (1.5pt);
\draw [fill=red] (0.0,0.0) circle (1.5pt);
\end{scope}

\begin{scope}[xshift=0cm,yshift=-3cm]
\draw [eb] (-0.7,0.0)-- (0.0,0.0);
\draw [fill=red] (0.0,0.0) circle (1.5pt);
\draw [<-|] (1.5,0) -- node[label=above:$\pi_3$] {} (2,0);
\end{scope}

\begin{scope}[xshift=3.5cm,yshift=-3cm]
\draw [eb] (-0.7,0.0)-- (0.0,0.0);
\draw [dotted, eb] (0.25,0) circle (0.25cm);
\draw [fill=red] (0.0,0.0) circle (1.5pt);
\end{scope}

\end{tikzpicture}
      \caption{\label{fig:3decomp} The expansion and creation moves to arrive at a 3--valent graph.}
    \end{figure}
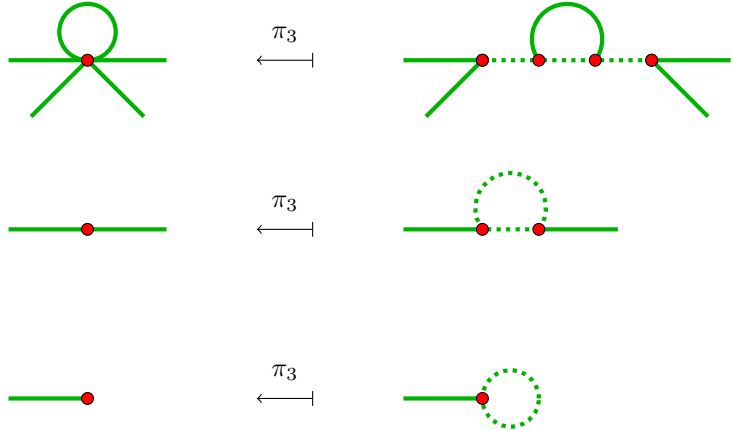

    For $n$ even, one notes that the $\pi_n$ maps into $\bgs_{even}$ since contraction and deletion both preserve the evenness of the vertex valency.  Specializing for a moment to the case of $n=4$, consider a graph $\bg\in\bgs_{even}$.  Once again, examining an $m$--valent vertex in $\bg$ ($m$ even), one may expand such a vertex to a sequence of 4--valent vertices joined by a string of virtual edges.  For a 2--valent vertex, one may simply add a virtual loop. See Figure \ref{fig:4decomp} for an illustration. 

    \begin{figure}[htb]
      \centering
      \tikzsetnextfilename{4val}

\begin{tikzpicture}[scale=1.5]

\draw [eb] (0.7,0.0)-- (0.0,0.0);
\draw [eb] (-0.7,0.0)-- (0.0,0.0);
\draw [eb] (0.0,0.0)-- (-0.5,-0.5);
\draw [eb] (0.5,-0.5)-- (0.0,0.0);
\draw [eb] (0.0,0.25) circle (0.25cm);

\draw [fill=red] (0.0,0.0) circle (1.5pt);

\draw [<-|] (1.5,0) --  node[label=above:$\pi_{4}$] {} (2,0);

\begin{scope}[xshift=3.5cm]

\draw [dotted, eb] (0.5,0.0)-- (0.0,0.0);
\draw [eb] (-0.7,0.0)-- (0.0,0.0);
\draw [eb] (0.0,0.0)-- (-0.5,-0.5);
\draw [eb] (1.0,-0.5)-- (0.5,0.0);
\draw [eb] (0.5,0.0)-- (1.2,0.0);
\draw [shift={(0.25,0.1875)},eb]  plot[domain=-0.64:3.785,variable=\t]({1.0*0.3125*cos(\t r)+-0.0*0.3125*sin(\t r)},{0.0*0.3125*cos(\t r)+1.0*0.3125*sin(\t r)});
\begin{scriptsize}
\draw [fill=red] (0.0,0.0) circle (1.5pt);
\draw [fill=red] (0.5,0.0) circle (1.5pt);
\end{scriptsize}
\end{scope}

\begin{scope}[xshift=0cm,yshift=-1.5cm]
\draw [eb] (0.7,0.0)-- (0.0,0.0);
\draw [eb] (-0.7,0.0)-- (0.0,0.0);
\draw [fill=red] (0.0,0.0) circle (1.5pt);

\draw [<-|] (1.5,0) -- node[label=above:$\pi_{4}$] {} (2,0);
\end{scope}

\begin{scope}[xshift=3.5cm,yshift=-1.5cm]
\draw [eb] (0.7,0.0)-- (0.0,0.0);
\draw [eb] (-0.7,0.0)-- (0.0,0.0);
\draw [dotted, eb] (0.0,0.25) circle (0.25cm);
\draw [fill=red] (0.0,0.0) circle (1.5pt);

\end{scope}

\end{tikzpicture}
      \caption{\label{fig:4decomp} The expansion and creation moves to arrive at a 4--valent graph.}
    \end{figure}
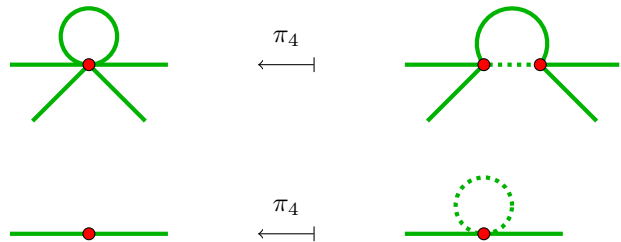

To generalize to arbitrary $n$ odd (even), then one need only to create $(n-3)/2$ (resp.\ $(n-4)/2$) virtual loops at each vertex.
\end{proof}

\begin{remark}
  In effect, one has encoded the unlabelled graphs in $\bgs$ in terms of labelled $n$--regular graphs in $\bgst_{n}$. The surjectivity result above implies that for $n$ odd (even), each graph $\bg\in\bgs$ (resp.\ $\bgs_{even}$) labels an class $\pi_n^{-1}(\bg)$ of graphs in $\bgst_{n}$. 
\end{remark}

One can go even a step further, encoding $\bgs$ in terms of loopless, $n$-regular labelled graphs:
\begin{remark}[{\bf surjection:} $\pi_{n,\loopless}$]
  \label{rem:looplessg}
  There exists a sequence of expansion and creation moves that effect a \emph{(1--n)--move}.  Consider an element of $\bgst_{n}$ that has (up to $\lfloor n/2\rfloor$) loops at some vertex $\bar{v}$.  Then, applying a (1--$n$)--move to this vertex, one can remove all loops. The effect of this transformation is depicted in Figure \ref{fig:onefour} for  a vertex with $n=4$ and one loop.   Thus, in each class $\pi_n^{-1}(\bg)$, there is a loopless graph. As for $n$ odd (even), there is a projection $\pi_{n,\loopless}:\bgst_{n,\loopless}\longrightarrow\bgs$ (resp.\ $\bgs_{even}$) such that $\pi_{n,\loopless}$ is surjective. Thus again, the boundary graphs $\bg$ label classes $\pi_{n,\loopless}^{-1}(\bg)$ in $\bgst_{n,\loopless}$.

  \begin{figure}[htb]
    \centering
     \tikzsetnextfilename{14move}

\begin{tikzpicture}[scale=1.5]

\draw [eb] (0.0,0.0)-- (-0.5,-0.5);
\draw [eb] (0.5,-0.5)-- (0.0,0.0);
\draw [eb] (0.0,0.25) circle (0.25cm);
\draw [fill=red] (0.0,0.0) circle (1.5pt);

\draw [<-|] (1.5,0) --  node[label=above:$\pi_{4,\loopless}$] {} (2,0);

\begin{scope}[xshift=3.5cm,yshift=.25cm]

\draw [dotted, eb] (0.5,0.0)-- (0.0,0.0);
\draw [dotted, eb] (0.5,-0.5)-- (0.0,-0.5);
\draw [dotted, eb] (0,-0.5) -- (0.0,0.0);
\draw [dotted, eb] (0.5,0.0)-- (0.5,-0.5);
\draw [dotted, eb] (0.5,-0.5)-- (0.0,0.0);
\draw [dotted, eb] (0.5,0.0)-- (0.0,-0.5);

\draw [eb] (0.0,-0.5)-- (-0.5,-1);
\draw [eb] (1.0,-1)-- (0.5,-0.5);
\draw [shift={(0.25,0.1875)},eb]  plot[domain=-0.64:3.785,variable=\t]({1.0*0.3125*cos(\t r)+-0.0*0.3125*sin(\t r)},{0.0*0.3125*cos(\t r)+1.0*0.3125*sin(\t r)});
\begin{scriptsize}
\draw [fill=red] (0.0,0.0) circle (1.5pt);
\draw [fill=red] (0.5,0.0) circle (1.5pt);
\draw [fill=red] (0.0,-0.5) circle (1.5pt);
\draw [fill=red] (0.5,-0.5) circle (1.5pt);
\end{scriptsize}
\end{scope}

\end{tikzpicture}
   \caption{\label{fig:onefour} Use of a $(1-n)$--move on an $n$-valent vertex with loop to create a loopless graph ($n=4$ in the example)}
  \end{figure}
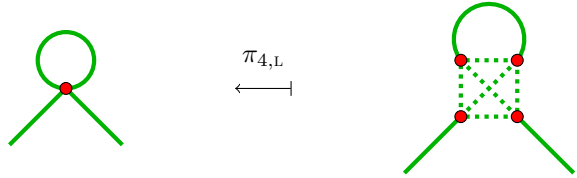

 \end{remark}

\begin{remark}[{\bf atomic reduction}]

  \label{rem:atomicred}
  There is an obvious and natural extension of the contraction/expansion and deletion/creation moves, defined for $\bgt\in\bgst$ in Definition \ref{def:moves}, to labelled spin foam atoms $\sfat\in\sfast$:

  \vspace{-0.2cm}

  \begin{description}
    \item[--]  {A contraction move on the virtual edge $(\bar{v}_1\bar{v}_2)\in \bgt$ translates to: \textit{i}) the deletion of the virtual subset \\
     $\{\hat{v}$, $(\bar{v}_1\hat{v})$, $(\bar{v}_2\hat{v})$, $(v\vhat)$,  $(v\bar{v}_1\hat{v})$, $(v\bar{v}_2\vhat)\}\subset\sfat$, as well as \textit{ii}) the identifications $\bar{v}_1=\bar{v}_2$ and   $(v\bar{v}_1)= (v\bar{v}_2)$.  
     }
    \item[--] {A deletion move on a virtual loop $(\bar{v}\bar{v})\in\bgt$ translates to the deletion of the virtual subset\\
     $\{\hat{v}, (\bar{v}\hat{v}), (\bar{v}\hat{v}), (v\bar{v}\hat{v}), (v\bar{v}\hat{v})\}\subset\sfat$. 
    }
  \end{description}
  The expansion and creation moves are similarly extended and they are illustrated in Figure \ref{fig:atommoves}.  

  \begin{figure}[htb]
    \centering
     \tikzsetnextfilename{apyr2}
     \begin{tikzpicture}[scale=2]

\draw [|->] (1.5,0) --  node[label=above:$\Pi_{3,\loopless}$] {} (2,0);

\begin{scope}
\node [c]		(v)	at (0,-.12)		{};
\node [c]		(1)	at (-0.56,0.32) 	{}; 
\node [c]		(2)	at (-0.17,-0.07)	{}; 
\node [c]		(3)	at (0.6,0) 		{}; 
\node [c]		(4)	at (.14,.37)	{}; 
\node [c]		(5a)	at (-.1,-.6)		{};
\node [c]		(5b)	at (.3,-.5)		{};
\node [c]		(12)	at (-.72,.32)	{};
\node [c]		(23)	at (.44,.09)	{};
\node [c]		(34)	at (.57,.52)	{};
\node [c]		(14)	at (-.3,.61)		{};
\node [c]		(15a)	at (-1,-.44)		{};
\node [c]		(25a)	at (-.28,-.97)	{};
\node [c]		(35b)	at (1.02,-.76)	{};
\node [c]		(45b)	at (.28,-.23)	{};
\node [c]		(5a5b)at (.1,-.55)	{};
\foreach \i/\j in {1/2,1/4,2/3,3/4,1/5a,2/5a,3/5b,4/5b,5a/5b}{
 \path	[f] 	(\i) -- (\i\j) -- (\j) -- (v) -- cycle;
 }
 \foreach \i in {1,2,3,4,5a,5b}{
  \draw	[e] 	(\i) -- (v);
  }
\foreach \i/\j in {1/2,1/4,2/3,3/4,1/5a,2/5a,3/5b,4/5b}{
 \draw 	[eh]	(v)		--	 (\i\j);
 \draw	[eb] 	(\i) node[vb] {} -- (\i\j) node[vh] {};
 \draw 	[eb] 	(\j) node[vb] {} -- (\i\j) node[vh] {};
 }
\draw 	[eh]	(v)		 	 -- (5a5b);
\draw[dotted,eb](5a) node[vb] {} -- (5a5b) node[vh] {};
\draw[dotted,eb](5b) node[vb] {} -- (5a5b) node[vh] {};
\draw 	[e] 	(3) node[vb] {} -- (v) node[v] {};
\end{scope}
\begin{scope}[xshift=3.5cm]
\node [c]		(v)	at (0,-.12)		{};
\node [c]		(1)	at (-0.56,0.32) 	{}; 
\node [c]		(2)	at (-0.17,-0.07)	{}; 
\node [c]		(3)	at (0.6,0) 		{}; 
\node [c]		(4)	at (.14,.37)	{}; 
\node [c]		(5)	at (0,-.6)		{};
\node [c]		(12)	at (-.72,.32)	{};
\node [c]		(23)	at (.44,.09)	{};
\node [c]		(34)	at (.57,.52)	{};
\node [c]		(14)	at (-.3,.61)		{};
\node [c]		(15)	at (-1,-.44)		{};
\node [c]		(25)	at (-.28,-.97)	{};
\node [c]		(35)	at (1.02,-.76)	{};
\node [c]		(45)	at (.28,-.23)	{};
\foreach \i/\j in {1/2,1/4,2/3,3/4,1/5,2/5,3/5,4/5}{
 \path	[f] 	(\i) -- (\i\j) -- (\j) -- (v) -- cycle;
 }
 \foreach \i in {1,2,3,4,5}{
  \draw [e] (\i) -- (v);
  }
\foreach \i/\j in {1/2,1/4,2/3,3/4,1/5,2/5,3/5,4/5}{
 \draw 	[eh]	(v)		-- 	(\i\j);
 \draw	[eb] 	(\i) node[vb] {} -- (\i\j) node[vh] {};
 \draw 	[eb] 	(\j) node[vb] {} -- (\i\j) node[vh] {};
 }
\draw [e] (3) node[vb] {} -- (v) node[v] {};
\end{scope}

\end{tikzpicture}
    \caption{\label{fig:atommoves} A contraction move on an atom.}
  \end{figure}
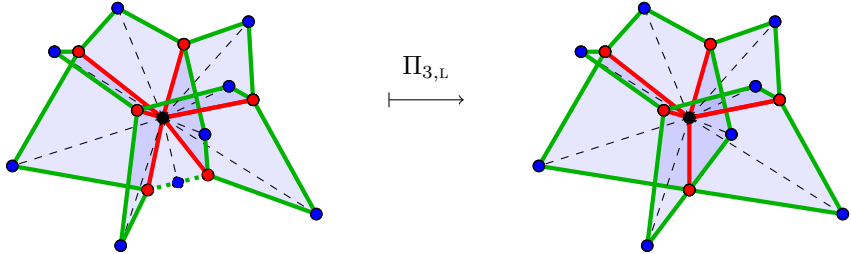

  Quite trivially, one may extend the map $\pi_{n,\loopless}$ of Remark \ref{rem:looplessg} to $\sfat\in\sfast_{n,\loopless}$. This map $\Pi_{n,\loopless}:\sfast_{n,\loopless}\longrightarrow\sfas$, $\Pi_{n,\loopless} \equiv \delta^{-1}\circ\beta\circ\pi_{n,\loopless}\circ\widetilde{\beta}^{-1}\circ\widetilde{\delta}$ is surjective.  Thus, each $\sfa\in\sfas$ marks a non--trivial class $\Pi_{n,\loopless}^{-1}(\sfa)\in\sfast_{n,\loopless}$.
\end{remark}

\begin{remark}[{\bf molecule reduction}]
  \label{rem:moleculered}
  While the bonding of atoms in $\sfast$ just follows the procedure laid out in Remark \ref{rem:generalization},  the reduction of a labelled spin foam molecule possesses certain subtleties. 
  Within a spin foam molecule, two scenarios arise for a virtual vertex $\hat{v}\in\sfrt$:
  \begin{description}
    \item[$\hat{v}\notin\widetilde{\delta}(\sfrt)$:]  
      Consider a virtual vertex $\hat{v}$ with the virtual edges and faces incident at $\hat{v}$, here denoted by $\{(\bar{v}_1\hat{v}), \dots, (\bar{v}_{k}\hat{v})\}$ and 
$\{(v_{12}\bar{v}_1\hat{v}),(v_{12}\bar{v}_2\hat{v}),\dots,(v_{k1}\bar{v}_k\hat{v}),(v_{k1}\bar{v}_1\hat{v})\}$, respectively. Following the rules laid out in Remark \ref{rem:atomicred}, a contraction move applied to that virtual substructure \textit{i}) deletes $\{\hat{v}\}$, as well as all edges and faces incident at $\hat{v}$ and \textit{ii}) identifies $\bar{v}\equiv\bar{v}_i$ and pairwise $(v_{ii+1}\bar{v}) \equiv (v_{ii+1}\bar{v}_i) = (v_{ii+1}\bar{v}_{i+1})$, for all $i\in\{1,\dots,k\}$.
  
      As illustrated in Figures \ref{fig:moleculecontraction} and \ref{fig:moleculecontraction2}, this contraction only behaves well when $k = 2$, that is, there are two virtual edges of type $(\vbar\vhat)$ incident at $\hat{v}$. For other values of $k$, the resulting structure does not lie within $\sfrst$ and therefore ultimately, it lies  outside $\sfrs$; the reason is that in a $\sfr\in\sfrs$ there are precisely two edges of type $e=(v\bar{v})\in\mathcal{E}$ incident at each vertex $\bar{v}\notin\delta(\sfr)$
while in the reduction of a $\sfrt\in\sfrst$ in general there occur any $k\ge 2$ edges at a vertex $\bar{v}\notin\widetilde{\delta}(\sfrt)$.

Moreover, the above condition ensures good behaviour under deletion moves as well.

    \item[$\hat{v}\in\widetilde{\delta}(\sfrt)$:]  In this case, a similar argument reveals the necessity for precisely one virtual edge of type $(\vbar\vhat)$ incident at $\hat{v}$ to obtain a molecule $\sfr\in\sfrs$ upon reduction.
  \end{description}
\end{remark}

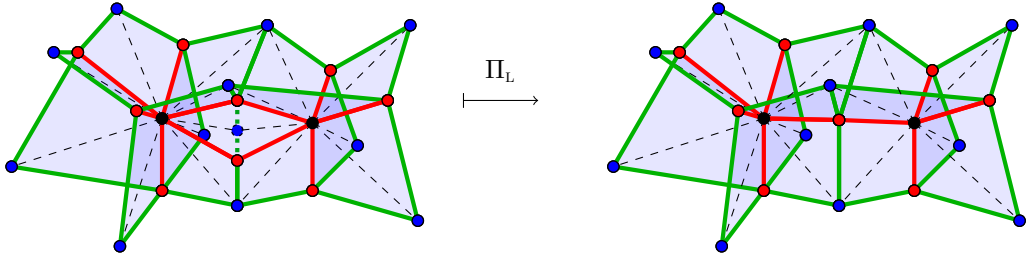
\begin{figure}[htb]
    \centering
        \tikzsetnextfilename{moleculecontraction}
        \begin{tikzpicture}[scale=2]

\draw [|->] (2,0) --  node[label=above:$\Pi_{\loopless}$] {} (2.5,0);

\node [c]		(w)	at (0,-.12)		{};
\node [c]		(6)	at (-0.56,0.32) 	{}; 
\node [c]		(7)	at (-0.17,-0.07)	{}; 
\node [c]		(8)	at (.14,.37)	{}; 
\node [c]		(5)	at (0,-.6)		{};
\node [c]		(67)	at (-.72,.32)	{};
\node [c]		(3a7)	at (.44,.1)		{};
\node [c]	 	(13a)	at (.44,.1)		{};
\node [c]		(3a8)	at (.7,.5)		{};
\node [c]		(23a)	at (.7,.5)		{};
\node [c]		(68)	at (-.3,.61)		{};
\node [c]		(56)	at (-1,-.44)		{};
\node [c]		(57)	at (-.28,-.97)	{};
\node [c]		(3b5)	at (.5,-.7)		{};
\node [c]		(3b4)	at (.5,-.7)		{};
\node [c]		(58)	at (.28,-.23)	{};
\begin{scope}[xshift=1cm]
\node [c]		(v)	at (0,-.15)		{};
\node [c]		(1)	at (.5,0)	 	{}; 
\node [c]		(2)	at (.12,.2)		{}; 
\node [c]		(3a)	at (-.5,0)		{}; 
\node [c]		(3b)	at (-.5,-.4)		{};
\node [c]		(3a3b)at (-.5,-.2)	{}; 
\node [c]		(4)	at (0,-.6)		{};
\node [c]  		(12)	at (.65,.5)		{};
\node [c]		(14)	at (.7,-.8)		{};
\node [c] 		(24)	at (.3,-.3)		{};
\end{scope}

\path [eb,dotted] (3a) -- node[vh] {} (3b);

\foreach \i/\j in {6/7,6/8,3a/7,3a/8,5/6,5/7,3b/5,5/8,3a/3b}{
 \path	[f] 	(\i) -- (\i\j) -- (\j) -- (w) -- cycle;
 \draw 	[eh]	(w) -- (\i\j);
 }
 \foreach \i in {6,7,3a,3b,8,5}{
  \draw [e] (\i) -- (w);
  }
\foreach \i/\j in {6/7,6/8,3a/7,3a/8,5/6,5/7,3b/5,5/8}{
 \draw	[eb] 	(\i) node[vb] {} -- (\i\j) node[vh] {};
 \draw 	[eb] 	(\j) node[vb] {} -- (\i\j) node[vh] {};
 }
\draw [e] (3a) node[vb] {} -- (w) node[v] {};
\draw [e] (3b) node[vb] {} -- (w) node[v] {};

\foreach \i/\j in {1/2,1/3a,1/4,2/3a,2/4,3b/4,3a/3b}{
 \path	[f] 	(\i) -- (\i\j) -- (\j) -- (v) -- cycle;
 \draw 	[eh]	(v) -- (\i\j);
 }
 \foreach \i in {1,2,3a,3b,4}{
  \draw [e] (\i) -- (v);
  }
\foreach \i/\j in {1/2,1/3a,1/4,2/3a,2/4,3b/4}{
 \draw	[eb] 	(\i) node[vb] {} -- (\i\j) node[vh] {};
 \draw 	[eb] 	(\j) node[vb] {} -- (\i\j) node[vh] {};
 }
\draw [e] (1) node[vb] {} -- (v) node[v] {};

\begin{scope}[xshift=4cm]
\node [c]		(w)	at (0,-.12)		{};
\node [c]		(6)	at (-0.56,0.32) 	{}; 
\node [c]		(7)	at (-0.17,-0.07)	{}; 
\node [c]		(8)	at (.14,.37)	{}; 
\node [c]		(5)	at (0,-.6)		{};
\node [c]		(67)	at (-.72,.32)	{};
\node [c]		(37)	at (.44,.1)		{};
\node [c]	 	(13)	at (.44,.1)		{};
\node [c]		(38)	at (.7,.5)		{};
\node [c]		(23)	at (.7,.5)		{};
\node [c]		(68)	at (-.3,.61)		{};
\node [c]		(56)	at (-1,-.44)		{};
\node [c]		(57)	at (-.28,-.97)	{};
\node [c]		(35)	at (.5,-.7)		{};
\node [c]		(34)	at (.5,-.7)		{};
\node [c]		(58)	at (.28,-.23)	{};
\begin{scope}[xshift=1cm]
\node [c]		(v)	at (0,-.15)		{};
\node [c]		(1)	at (.5,0)	 	{}; 
\node [c]		(2)	at (.12,.2)		{}; 
\node [c]		(3)	at (-.5,-.13)	{}; 
\node [c]		(4)	at (0,-.6)		{};
\node [c]  		(12)	at (.65,.5)		{};
\node [c]		(14)	at (.7,-.8)		{};
\node [c] 		(24)	at (.3,-.3)		{};
\end{scope}
\foreach \i/\j in {6/7,6/8,3/7,3/8,5/6,5/7,3/5,5/8}{
 \path	[f] 	(\i) -- (\i\j) -- (\j) -- (w) -- cycle;
 }
 \foreach \i in {6,7,3,8,5}{
  \draw [e] (\i) -- (w);
  }
\foreach \i/\j in {6/7,6/8,3/7,3/8,5/6,5/7,3/5,5/8}{
 \draw 	[eh]	(w)	-- (\i\j);
 \draw	[eb] 	(\i) node[vb] {} -- (\i\j) node[vh] {};
 \draw 	[eb] 	(\j) node[vb] {} -- (\i\j) node[vh] {};
 }
\draw [e] (3) node[vb] {} -- (w) node[v] {};

\foreach \i/\j in {1/2,1/3,1/4,2/3,2/4,3/4}{
 \path	[f] 	(\i) -- (\i\j) -- (\j) -- (v) -- cycle;
 }
 \foreach \i in {1,2,3,4}{
  \draw [e] (\i) -- (v);
  }
\foreach \i/\j in {2/4,1/2,1/3,1/4,2/3,3/4}{
 \draw 	[eh]	(v)		-- (\i\j);
 \draw	[eb] 	(\i) node[vb] {} -- (\i\j) node[vh] {};
 \draw 	[eb] 	(\j) node[vb] {} -- (\i\j) node[vh] {};
 }
\draw [e] (1) node[vb] {} -- (v) node[v] {};

\end{scope}

\end{tikzpicture}
    \caption{\label{fig:moleculecontraction} Contraction move with respect to a vertex $\vh$ incident to two virtual edges in a molecule.}
 \end{figure}

\begin{figure}[htb]
    \centering
        \tikzsetnextfilename{moleculecontraction2}
        \begin{tikzpicture}

\draw [|->] (2,0) --  node[label=above:$\sharp_{\{\gamma_1,\gamma_2,\gamma_3\}}$] {} (3.5,0);
\draw [|->] (6.5,0) --  node[label=above:$\Pi_{\loopless}$] {} (7.5,0);

\begin{scope}[xshift=.5cm,yshift=.4cm]
\node [c]		(fu)	at (0,0)		{};
\node [c]		(u)	at (1,.58)	 	{};  
\node [c]		(uv)	at (0,.58)		{};
\node [c]		(uw)	at (.5,-.3)		{};
\end{scope}
\begin{scope}[xshift=-.5cm,yshift=.4cm]
\node [c]		(fv)	at (0,0)		{};
\node [c]		(v)	at (-1,.58)		{};
\node [c]		(vu)	at (0,.58)		{};
\node [c]	 	(vw)	at (-.5,-.3)		{};
\end{scope}
\begin{scope}[yshift=-.5cm]
\node [c]		(fw)	at (0,0)		{};
\node [c]		(w)	at (0,-1.16)	{}; 
\node [c]		(wu)	at (.5,-.3)		{};
\node [c]		(wv)	at (-.5,-.3)		{};
\end{scope}
\foreach \i/\j/\k in {u/v/w,v/w/u,w/u/v}{
\path		[f] 	(\i) 	-- (\i\j) -- (f\i) -- (\i\k) -- cycle;
\draw[dotted,eb](f\i) 	-- (\i\j);
\draw[dotted,eb](f\i) 	-- (\i\k);
\path		(\j\k) 	edge [bb]  (\k\j)
		(f\i)	edge [bh] 	(f\j);
\draw	[e]	(\i)	-- (\i\j) node[vb]{};
\draw	[e]	(\i)  node[v]{} 	-- (\i\k) node[vb]{};
\draw	[eh]	(\i) 	-- (f\i) node[vh]{};
}

\begin{scope}[xshift=5cm]
\node [c]		(fu)	at (0,0)		{};
\node [c]		(u)	at (1,.58)	 	{};  
\node [c]		(uv)	at (0,.58)		{};
\node [c]		(uw)	at (.5,-.3)		{};
\node [c]		(fv)	at (0,0)		{};
\node [c]		(v)	at (-1,.58)		{};
\node [c]		(vu)	at (0,.58)		{};
\node [c]	 	(vw)	at (-.5,-.3)		{};
\node [c]		(fw)	at (0,0)		{};
\node [c]		(w)	at (0,-1.16)	{}; 
\node [c]		(wu)	at (.5,-.3)		{};
\node [c]		(wv)	at (-.5,-.3)		{};
\foreach \i/\j/\k in {u/v/w,v/w/u,w/u/v}{
\path		[f] 	(\i) 	-- (\i\j) -- (f\i) -- (\i\k) -- cycle;
\draw[dotted,eb](f\i) 	-- (\i\j);
\draw[dotted,eb](f\i) 	-- (\i\k);
\draw	[e]	(\i)	-- (\i\j) node[vb]{};
\draw	[e]	(\i)  node[v]{} 	-- (\i\k) node[vb]{};
\draw	[eh]	(\i) 	-- (f\i) node[vh]{};
 }
\end{scope}

\begin{scope}[xshift=9cm]
\node [vb]		(f)	at (0,0)		{};
\node [v]		(u)	at (1,.58)	 	{};  
\node [v]		(v)	at (-1,.58)		{};
\node [v]		(w)	at (0,-1.16)	{}; 
\foreach \i in {u,v,w}{
\draw [e] (f) -- (\i);
}
\end{scope}
\end{tikzpicture}
    \caption{\label{fig:moleculecontraction2} Contraction move with respect to a vertex $\vh$ adjacent to three virtual edges as consequence of three bondings. The contraction identifies three boundary vertices and the resulting vertex is incident to three bulk edges. Such a situation is not possible in a molecule $\sfr\in\sfrs$.}
 \end{figure}
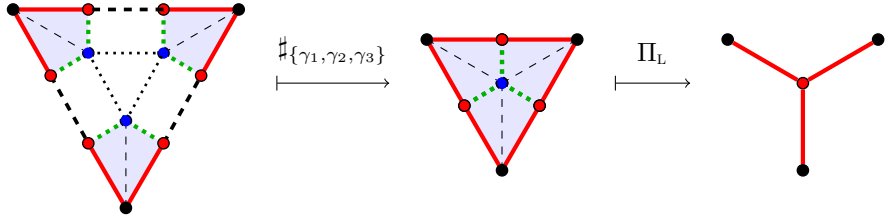
 
\begin{remark}[{\bf dually--weighted molecules}]
\label{rem:dually} 
  Remark \ref{rem:moleculered} instructs that one is not interested in the whole of $\sfrst_{n,\loopless}$, but rather in the subset that possesses vertices $\hat{v}\in\widehat{\mathcal{V}}$ with at most two (resp.\ precisely two) virtual edges incident at a vertex $\hat{v}$.
This set is denoted by $\sfrst_{n,\looplessdually}$. 
  The reason for this nomenclature will become clear in Section \ref{sec:gft}. Fortunately, the expansion/creation moves act each time on a single vertex $\bar{v}$, so that one may define a surjective map $\Pi_{n,\looplessdually}:\sfrst_{n,\looplessdually}\longrightarrow\sfrs$.  In words, each unlabelled spin foam molecule is represented in $\sfrs_{n,\looplessdually}$.
\end{remark}

%
\begin{remark}
  \label{prop:generated}
  Anticipating the \gft\ application, it should be emphasized that the whole construction is based on only one single kind of labelled patches,  the $n$-patch.
In the  labelled case this is not unique but there are $2^n$ $n$-patches and we denote their set as $\bpst_n$.
Thus we have that  $\bbgst_{n,\loopless} = \sigma(\bpst_{n})$.
\end{remark}


\subsection{Part 6: molecules from simplicial structures}
\label{sssec:primitivea}

Finally, we can show that it is even possible to use only molecules obtained from bonding labelled atoms of simplicial type to recover all arbitrary unlabelled molecules in terms of reduction:
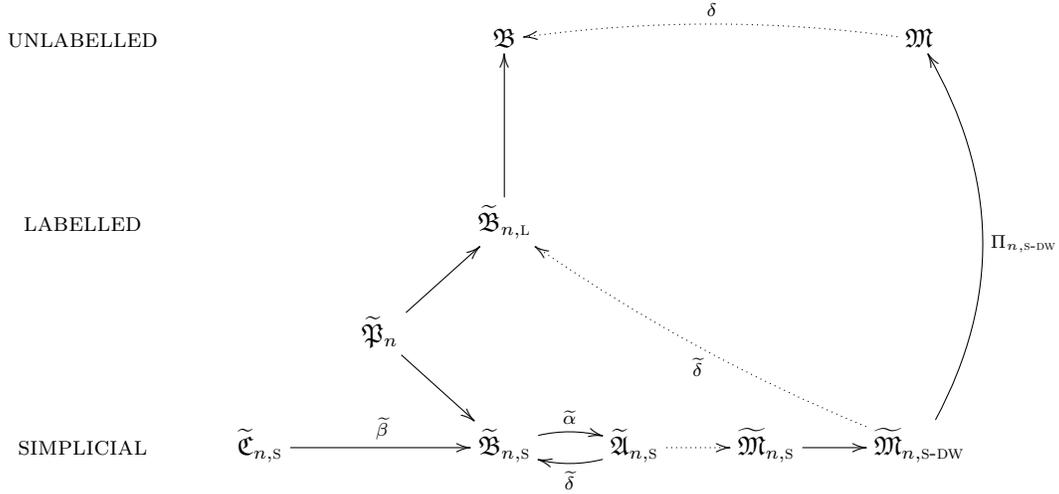
\begin{displaymath}
  \xymatrix{
    \textsc{unlabelled} & & &\bbgs &&&\sfrs\ar@/_/@{.>}[lll]_{\delta} \\
   &\\
   \textsc{labelled} &  & &\bbgst_{n,\loopless} \ar[uu]  & &&\\
   && \bpst_{n
  }\ar[ur] \ar[dr] & & & &\\
   \textsc{
   simplicial} &\bgst_{n,\primitive}\ar[rr]^{\widetilde{\beta}}&& \bbgst_{n,\primitive}\ar@/^/[r]^{\widetilde{\alpha}}& \sfast_{n,\primitive}\ar@/^/[l]^{\widetilde{\delta}}\ar@{.>}[r]&\sfrst_{n,\primitive}\ar[r] & \sfrst_{n,\primitivedually}\ar@/^/@{.>}[uulll]^{\widetilde{\delta}}\ar@/_{2pc}/[uuuu]_{\Pi_{n,\primitivedually}}
 }
  \end{displaymath}

In Propositions \ref{prop:surjectors}
, it was shown that all boundary graphs could be encoded in terms of labelled, $n$--regular, loopless graphs.  Moreover, from the spin foam point of view these graphs occur as the boundaries of labelled spin foam atoms $\sfast_{n,\loopless}$ \emph{and} labelled spin foam molecules $\sfrst_{n,\looplessdually}$ (see Remarks \ref{rem:moleculered} and \ref{rem:dually}).  However, one would also like to show that all possible boundary graphs arise as the boundary of molecules composed of atoms drawn from a small finite set of types. 

This can be achieved using the labelled version of simplicial graphs and atoms.
\begin{remark}[{\bf labelled $n$--simplicial structures}]
Due to the label on each edge, there are $2^{(n+1)(n+2)/2}$ \emph{labelled $n$--simplicial boundary graphs}, denoted $\bgst_{n,\primitive}$. 

Through the maps $\widetilde{\beta}$ and $\widetilde{\alpha}$, defined in Remark \ref{rem:generalization}, one can rather easily obtain the \emph{labelled bisected $n$--simplicial graphs} $\bbgst_{n,\primitive}$ and \emph{labelled $n$--simplicial atoms} $\sfast_{n,\primitive}$, respectively. 

Furthermore, label--preserving bonding maps $\widetilde{\gamma}$ give rise to \emph{labelled $n$--simplicial molecules} $\sfrst_{n,\primitive}$, and their subclass $\sfrst_{n,\primitivedually}$ according to Remark \ref{rem:dually}.
\end{remark}

\begin{remark}[{\bf atoms from patches}]
  \label{rem:patchatom}
  One can use an $n$--patch $\bpt_n\in\bpst_{n}$ as the foundation for a bisected \simplicial\ $n$--graph $\bbgt\in\bbgst_{n,\primitive}$ in the following manner:  
  \begin{description} 
    \item[--]  A $n$--patch consists of a single $n$--valent vertex $\bar{v}$,  1--valent vertices $\hat{v}^i$ with $i\in I_{\bar{v}}$ an $n$--element index set, and labelled edges $(\bar{v}\hat{v}^i)$.  
    \item[--] For each $i$, one creates a new vertex $\bar{v}^i$, along with an edge $(\bar{v}^i\hat{v}^i)$ with the same label as $(\bar{v}\hat{v}^i)$. 
    \item[--] For each pair of new vertices $\bar{v}^i$ and $\bar{v}^j$ with $i\neq j$, one creates a new vertex $\hat{v}^{ij}$, along with a pair of {\bf real} 
    edges $(\bar{v}^i\hat{v}^{ij})$ and $(\bar{v}^j\hat{v}^{ij})$.
  \end{description}
  The result is a \simplicial\ $n$--graph. In a moment, it will be useful to distinguish the constructed \simplicial\ $n$--graph by $\bbgt_{\bar{v}}$, the original 
   $n$--patch by $\bpt(\bar{v}) \equiv \bpt_{ \bar{v}}(\bbgt_{\bar{v}})$,  and new patches by $\bpt_{\bar{v}_i}(\bbgt_{\bar{v}_i})$ for $i\in I$.
\end{remark}

The aim is summarized in the statement:

\fbox{
  \begin{minipage}[c][][c]{0.97\textwidth}
\begin{proposition}
  \label{prop:encodeproof}
  Every graph in $\bgst_{n,\loopless}$ arises as the boundary graph of a dually--weighted molecule composed of \simplicial\ $n$--atoms.  
\end{proposition}
\end{minipage}
}

\begin{proof}
  The basic argument is fairly straightforward and goes as follows: given a graph $\bgt\in\bgst_{n,\loopless}$, one bisects it and thereafter cuts it into its constituent patches; one uses Remark \ref{rem:patchatom} to construct a \simplicial\ $n$--atom from each patch: one supplements this set of atoms with bonding maps that yield a molecule with $\bgt$ as boundary.  
The procedure is also sketched in Figure \ref{fig:decomposition}.
\begin{figure}[htp]
    \centering
        \tikzsetnextfilename{decomposition}
        \begin{tikzpicture}[scale=1.5]
\node [vb,label=left:$\vb_i$]	(1)	at (0,-.2)	{};
\node [vh,label=above:$\vh_{ij}$](12)	at (1,-.1)	{};
\node [vh]					(13)	at (.4,.4)	{};
\node [vh]					(14)	at (.5,-.5)	{};
\node [vb,label=right:$\vb_j$]	(2)	at (2,0)	{};
\node [vh]					(21)	at (1,-.1)	{};
\node [vh]					(23)	at (1.6,.5)	{};
\node [vh]					(24)	at (1.5,-.4)	{};
\foreach \i/\j in {1/2,1/3,1/4,2/1,2/3,2/4}{
    \draw [eb] (\i) -- (\i\j);
    }
\draw [->] (3,0) -- (4,0);

\begin{scope}[xshift=6cm,yshift=-.2cm]
  \node [vb,label=left:$\vb_i$]	(3)	at (-1.1,-.2){};
  \node [vb]				(4)	at (0,-.8)	{};
  \node [vb,label=above:$\vb_i^j$](1)	at (1,0) 	{};
  \node [vb]				(2)	at (.2,1)	{};
  \foreach \i/\j in {1/2,2/4,3/4,2/3,1/3,1/4}{
    \draw [eb] (\i) -- node[vh] {} (\j);
    }
\end{scope}

\begin{scope}[xshift=9cm]
  \node [vb,label=above:$\vb_j^i$](5)	at (-1.1,-.2){};
  \node [vb]				(4)	at (0,-.8)	{};
  \node [vb,label=right:$\vb_j$]	(6)	at (1,0) 	{};
  \node [vb]				(2)	at (.2,1)	{};
  \foreach \i/\j in {6/2,2/4,5/4,2/5,6/5,6/4}{
    \draw [eb] (\i) -- node[vh] {} (\j);
    }
\end{scope}

\path	(1) 	edge [bb] (5)
	(6,-.3) node[label=100:$\vh_i^j$] {} edge [bh, bend right=15] node[label=below:$\gamma_{ij}$] {} (9,-.1);
\node  [label=100:$\vh_j^i$]  at (9,-.1) {};

\end{tikzpicture}
    \caption{\label{fig:decomposition} Decomposition of an atom with boundary graph $\bgt\in\bgst_{n,\loopless}$ into simplicial atoms, sketched for the patches of two connected vertices in $\bgt$ and $n=3$.}
 \end{figure}

  \begin{description}
    \item[{\it index:}]
  More precisely, consider a labelled, loopless, $n$--regular graph $\bgt\in\bgst_{n,\loopless}$, with $\bgt = (\overline{\mathcal{V}},\overline{\mathcal{E}})$. It is useful to index the vertex set by $\bar{v}_i$ with $i\in\{1,\dots,|\overline{\mathcal{V}}|\}$. This induces an index for the edges; an edge joining $\bar{v}_i$ to $\bar{v}_j$ is indexed by $e_{ij}^{(a)}$, where a non--trivial index $(a)$ arises should multiple edges join the two vertices. 
  
\item[{\it bisect:}] The graph $\bgt$ has a bisected counterpart $\widetilde\beta(\bgt) = \bbgt = (\mathcal{V}_{\bbgt}, \mathcal{E}_{\bbgt})$.  
  The vertex set $\mathcal{V}_{\bbgt} = \overline{\mathcal{V}}\cup\widehat{\mathcal{V}}$, where $\widehat{\mathcal{V}}$ is the set of bisecting vertices.
  A vertex in $\widehat{\mathcal{V}}$ is indexed by $\hat{v}_{ij}^{(a)}$ if it bisects the edge $e_{ij}^{(a)}$ of $\bgt$.
   
\item[{\it cut:}]
  The boundary patches in $\bbgt$ are $\bpt_{\bar{v}_i}(\bbgt)$ with $i\in\{1,\dots, |\overline{\mathcal{V}}|\}$. The patch $\bpt_{\bar{v}_i}(\bbgt)$ is comprised of the vertex $\bar{v}_i$, the  $n$ vertices $\hat{v}_{ij}^{(a)}$ and $n$ edges $(\bar{v}_i\hat{v}_{ij}^{(a)})$. The indices of type $j(a)$, attached to the $n$ elements $\hat{v}_{ij}^{(a)}$, form an $n$--element index set $I_{\bar{v}_i}$. 
  
  
  Each bisecting vertex $\hat{v}_{ij}^{(a)}\in\widehat{\mathcal{V}}$ is shared by precisely two patches.

  Now one cuts the graph along each bisecting vertex and considers each patch in isolation.
  This cutting procedure sends each $\bpt_{\bar{v}_i}(\bbgt) \longrightarrow\bpt(\bar{v}_i)$, where $\bpt(\bar{v}_i)$ is a 
   $n$--patch
 comprising of a vertex $\bar{v}_i$, $n$ vertices $\hat{v}_{i}^{j(a)}$ and $n$ edges $(\bar{v}_i\hat{v}_{i}^{j(a)})$.

 Thus, after cutting, a bisecting vertex $\hat{v}_{ij}^{(a)}$ is represented by $\hat{v}_{i}^{j(a)}$ in $\bpt(\bar{v}_i)$ and $\hat{v}_{j}^{i(a)}$ in $\bpt(\bar{v}_j)$.

\item[{\it atoms:}]
  For the patch $\bpt(\bar{v}_i)$, the $n$ superscript indices $j(a)$ are that indexing set $I_{\bar{v}_i}$, defined a moment ago.   Thus, one may use Remark \ref{rem:patchatom} to construct, from $\bpt(\bar{v}_i)$, a \simplicial\  $n$--graph $\bbgt_{\bar{v}_i}$  and there after a \simplicial\  $n$--atom $\sfat_{\bar{v}_i}$.

  Through this process, one obtains a set of \simplicial\  $n$--atoms, $\sfat_{\bar{v}_i}$ with $i\in\{1,\dots,|\overline{\mathcal{V}}|\}$.  This set is denoted by $\sfat_{\overline{\mathcal{V}}}$, since the atoms are in one--to--one correspondence with the vertices $\overline{\mathcal{V}}$ of $\bbgt$.   They will be used to form a spin foam molecule whose bisected boundary graph is $\bbgt$.
 
\item[{\it bonding maps:}] 

  For each pair $\bar{v}_{i}^{j(a)}\in\bbgt_{\bar{v}_i}$, $\bar{v}_j^{i(a)}\in\bbgt_{\bar{v}_j}$, define a bonding map
  \begin{eqnarray}
\gamma_{ij}^{(a)}:\bpt_{\bar{v}_{i}^{j(a)}}(\bbgt_{\bar{v}_i}) & \longrightarrow & \bpt_{ \bar{v}_j^{i(a)}}(\bbgt_{\bar{v}_j})\\
\bar{v}_i^{j(a)} & \longrightarrow & \bar{v}_j^{i(a)} \\
\hat{v}_i^{j(a)} & \longrightarrow & \hat{v}_j^{i(a)}
\end{eqnarray}
while the remaining $n-1$ vertices in each patch are paired in an arbitrary way:\footnote{As an aside, the bonding maps are specified only up to permutations of these $n-1$ vertex pairings, leading to ${n-1 \choose 2}$ choices for each bonding map. However, the resulting spin foam molecules possess the same boundary.}
\begin{eqnarray}
  \left\{ \hat{v}_i^{j(a)k(b)} : k(b)\in I_{\bar{v}_i} - \{j(a)\}\right\} & \longrightarrow & \left\{ \hat{v}_j^{i(a)l(c)} : l(c)\in I_{\bar{v}_j} - \{i(a)\}\right\}\;.
\end{eqnarray}
The set of bonding maps is denoted $\gamma_{\widehat{\mathcal{V}}}$, since the maps are in one--to--one correspondence with the bisecting vertices $\widehat{\mathcal{V}}$ of  $\bbgt$.

Then, in the molecule $\sfrt = \sharp_{\gamma_{\Vhat}}\sfat_{\Vbar}$, the only patches that remain unbonded are the original $\bpt(\bar{v}_i)$ for $i\in\{1,\dots,|\overline{\mathcal{V}}|\}$.  Moreover, after one relabels the identified vertices $\hat{v}_{ij}^{(a)} \equiv \hat{v}_{i}^{j(a)} = \hat{v}_{j}^{i(a)}$, one has truly come full circle:  the boundary of $\sfrt$, which may be extracted using Remark \ref{rem:moleculeboundary},  satisfies the relation $\widetilde{\delta}(\sfrt) = \bbgt$.

\item[{\it dually--weighted:}]
  From Remark \ref{rem:patchatom}, one notices that all edges added in the construction are {\bf real}.  Thus, the molecule $\sfrt\in\sfrst_{n,\primitivedually}$.

\end{description}
\end{proof}

Proposition \ref{prop:encodeproof} has the following consequence:  
\begin{corollary}[{\bf molecule decomposition}]
\label{rem:refmolreduction}
There is a decomposition map $D_{n,\looptosimp}:\sfrst_{n,\looplessdually}\longrightarrow\sfrst_{n,\Simplicialdually}$. 
\end{corollary}
\begin{proof}
	Consider $\sfrt\in\sfrst_{n,\looplessdually}$. 	By Proposition \ref{prop:encodeproof}, one can decompose each of its atoms, leading to the image of the molecule $\sfrt$ itself under decomposition map $D_{n,\looptosimp}$. 
\end{proof}
We note an important limitation.  
\begin{proposition}
	\label{prop:limitation}
	The projection $\Pi_{n,\Simplicialdually}:\sfrst_{n,\Simplicialdually}\longrightarrow\sfrs$ is {\bf not} surjective.
\end{proposition}
	We sketch our reasoning here. 
	Consider a generic $\sfr\in\sfrs$ and let $\sfrt$ be a representative in the class $\Pi_{n,\looplessdually}^{-1}(\sfr)$. Then, $\sfrt$ consists of bonded spin foam atoms drawn from the set $\sfast_{n,\loopless}$. According to Proposition \ref{prop:encodeproof}, every atom $\sfat\in\sfast_{n,\loopless}$ has a decomposition into simplicial atoms of $\sfast_{n,\primitive}$.  Just like in the decomposition utilized in \ref{prop:encodeproof},  it is possible to show that any decomposition requires one to add {\bf real} structures in order to maintain the integrity of the boundary graph under reduction.  However, if one adds in real structures, then one does not arrive back to the original atom/molecule after reduction, since reduction just amounts to contraction and deletion of virtual structures.


\subsection{Enhancing with higher--dimensional information}
\label{sssec:enhancing}

 We pause to remark on the relationship between these molecular spin foam structures and $D$--dimensional topologies. For clarity, we shall concentrate on $n$--regular structures. 

The elements of $\sfast_{n,\primitive}$  possess at most 2--dimensional components, and so in principle have no information about any higher--dimensional embedding.  Such higher--dimensional components must be added by some mechanism.  There exist two paths%
\footnote{The two paths mentioned above are the ones most often used in the quantum gravity literature. From the mathematical perspective there is an interesting third way, detailed in \cite{Denicola:2010fa}. Therein the authors extend embedded discrete structures to include topological data that encode the underlying $D$-manifold as a branched cover.
}
 that one may follow, both of which set $n=D$.   

\begin{remark}[{\bf $D$--dimensional structure by hand}]
In the first approach, one notes that spin foam atoms $\sfast_{D,\primitive}$ form the dual 2--skeleton to a $D$--simplex. Thus, at the atomic level, the $D$--dimensional structure can be defined by hand \emph{once} at the outset. As the result, the \simplicial\  $D$--graphs implicitly encode the $(D-1)$--dimensional boundary of a $D$--simplex, while the \simplicial\  $D$--patches are enhanced to $(D-1)$--simplices.  The tricky issue, of course, comes when one bonds \simplicial\  $D$--patches. These bonding maps should be augmented to identify $(D-1)$--dimensional information.  With many subtleties, these enhanced bonding maps can be defined \emph{once} at the start and applied mechanically throughout the bonding process.  However, the spin foam molecules, reconstructed in the manner, will generically encode $D$--dimensional objects that are very ill--behaved from a topological viewpoint \cite{\cguraulost, \csmerlaklost}. 
\end{remark}

\begin{remark}[{\bf $D$--dimensional structure from colouring}]
A second approach, which has gained a lot of traction in recent years,  is based upon so--called \emph{$D$--coloured graphs} \cite{Bonzom:2012bg}. Of course, this means defining yet another set of boundary graphs, with yet more labels, their associated spin foam atoms, bonding maps and so on. However, the definitions are like those given above, so we concentrate on their properties. Consider the set of labelled loopless $D$--regular boundary graphs $\bgst_{D,\loopless}$. Look for the subset that are \emph{$D$--colourable}, in the sense that one may assign to each edge another label drawn from the set $\{1,\dots, D\}$, such that the $D$ edges of each \simplicial\  $D$--patch have distinct colour s. This subset is called $\bgst_{D,coloured}$.  It emerges that the \simplicial\  $(D+1)$--graphs lie in this subset \emph{and} they generate, when coloured and accompanied by bonding maps that conserve edge colour , the whole of $\bgst_{D,coloured}$. 
Remarkably, this colour  information ensures that one can reconstruct an abstract simplicial pseudo-manifold \cite{Gurau:2010iu}.  While not all graphs in $\bgst_{D,\loopless}$ are $D$--colourable,  the $D$--dimensional topologies encoded by such spin foam molecules are much better behaved than those reconstructed using the first approach. 
\end{remark}

  One could in principle attempt to make a more ambitious statement. By showing the existence, for $D$ odd (even), of a surjective map $\pi_{D,coloured}:\bgst_{D,coloured}\longrightarrow\bgs$ (resp.\ $\bgs_{even}$), one could conjecture the following:
\begin{conjecture}
$D$--coloured graphs capture all of $\bgs$  ($\bgs_{even}$).  
\end{conjecture}
In essence, all one would need to show is that in every class $\pi_{D,\loopless}^{-1}(\bg) \subset \bgst_{D,\loopless}$, there is a graph that is $D$--colourable.

The benefit would be that in this way one could, for arbitrary molecules $\sfr$, specify the subclass whose molecules allow for a subdivision into the colourable subclass of $\sfrst_{n,\primitivedually}$.
Thus, all these molecules would have a well-behaved topological structure as pseudo-$D$-manifolds. 
In particular, their atoms would carry the structure of $D$-dimensional polytopes (cf. \ref{sec:polytopes}).

\subsection{Stranded diagrams}
\label{ssec:stranded}

One might wonder at this stage how the structures above match the usual stranded graph description utilized in group field theory. It emerges that stranded graphs can easily incorporate the information pertaining to generic spin foam atoms and molecules, as well as virtual and \simplicial\  structures.  Moreover, stranded diagrams provide a more succinct graphical representation for molecular spin foams.  With this aim in mind, we provide here a dictionary between the two descriptions.
\begin{defin}
  A \emph{stranded atom} is the double, $\mathfrak{s} = (\mathcal{C}, \mathcal{R})$, such that:
  \begin{description}
    \item[--] $\mathcal{C}$ is a set of vertices partitioned into subsets known as \emph{coils}. This set $\mathcal{C}$ has an even number of elements and coils are denoted by $c$.  
    \item[--] $\mathcal{R}$ is the set of \emph{reroutings}, where a rerouting is an edge, refered to quite frequently as a \emph{strand}, joining a pair of {\bf distinct} vertices in $\mathcal{C}$.  This set of reroutings saturates the set of vertices,  in the sense that each vertex is an endpoint of exactly one strand.  
  \end{description}
\end{defin}
We denote the set of stranded atoms by $\mathfrak{S}$. 
\begin{remark}
  One must take note of a particular type of rerouting, known as a \emph{retracing}.  This refers to a strand joining two vertices in the same coil.  One will see in moment that a retracing corresponds to a loop in the associated boundary graph.
\end{remark}
\begin{remark}
  \label{rem:explosion}
  Consider a spin foam atom $\sfa=(\mathcal{V},\mathcal{E},\mathcal{F})\in\sfas$.  As was shown in Proposition \ref{rem:bdycor}, it is completely determined by its boundary graph $\bg = (\overline{\mathcal{V}}, \overline{\mathcal{E}})\in\bgs$. From $\bg$, one constructs a stranded graph $\mathfrak{s}=(\mathcal{C},\mathcal{R})$ by \lq\lq exploding\rq\rq\ the vertices $\bar{v}\in\overline{\mathcal{V}}$. More precisely, for each edge $\bar{e} = (\bar{v}_1\bar{v}_2)\in\overline{\mathcal{E}}$, one creates two vertices in $\mathcal{C}$ (one for each endpoint) and a strand in $\mathcal{R}$ joining them. The subset of vertices in $\mathcal{C}$ created from a given endpoint vertex in $\overline{\mathcal{V}}$ constitutes a coil. 

  The reverse operation is equally simple. Given a stranded diagram $\mathfrak{s}$, one constructs a boundary graph $\bg$ by identifying the vertices within each coil. 

  These operations are clearly inversely related and are illustrated for a simple example in Figure \ref{fig:explosion}.
 
  \begin{figure}[htb]
    \centering
     \tikzsetnextfilename{strandedvertex}
      \begin{tikzpicture}[scale=1.5]

\draw [<->] (2,0)-- (3,0);
\begin{scope}

\node [vb]		(a)	at (0,-1)		{};
\node [vb]		(b)	at (1.25,1)		{}; 
\node [vb]		(c)	at (0.5,0.5) 	{}; 
\node [vb]		(d)	at (-0.5,0.5)	{}; 
\node [vb]		(e)	at (-1.25,1)	{}; 
\foreach \i/\j in {a/b,a/c,a/d,a/e,b/c,b/e,c/d,d/e}{
 \draw [eb] (\i) -- node[vh] {} (\j);
  }
\end{scope}

\begin{scope}[xshift=5cm]
\foreach \i in {0,60,120,180}{
\draw [cs, rotate=\i] (1.1,-.3) rectangle (.9,.3);
}
\draw [cs]	 (-.4,-1.1) rectangle (.4,-.9);
\node [vs]		(14)	at (-.35,.95)	{};
\node [vs]		(15)	at (-.5,.86)		{};
\node [vs]		(12)	at (-.65,.77)	{};
\node [vs]		(21)	at (-1,.2)		{};
\node [vs]		(23)	at (-1,0)		{};
\node [vs]		(25)	at (-1,-.2)		{};
\node [vs]		(52)	at (-.3,.-1)		{};
\node [vs]		(51)	at (-.1,-1)		{};
\node [vs]		(54)	at (.1,-1)		{};
\node [vs]		(53)	at (.3,-1)		{};
\node [vs]		(34)	at (1,.2)		{};
\node [vs]		(32)	at (1,0)		{};
\node [vs]		(35)	at (1,-.2)		{};
\node [vs]		(41)	at (.35,.95)	{};
\node [vs]		(45)	at (.5,.86)		{};
\node [vs]		(43)	at (.65,.77)	{};
\path
\foreach \i/\j in {14/41,21/12,43/34}{
  (\i) edge [es,bend right=70] (\j)
  }
\foreach \i/\j in {52/25,35/53}{
  (\i) edge [es,bend right=50] (\j)
  }
\foreach \i/\j in {51/15,45/54}{
  (\i) edge [es,bend right=16] (\j)
  }
(23) edge [es] (32);
\end{scope}
\end{tikzpicture}
      \caption{\label{fig:explosion} An example of the relation between (bisected) boundary graphs and stranded diagrams. While faces of atoms (and molecules) are in 1-to-1 correspondence to bisection vertices in the graph description, in the stranded diagrams they are uniquely represented by the strands.
            }
  \end{figure}
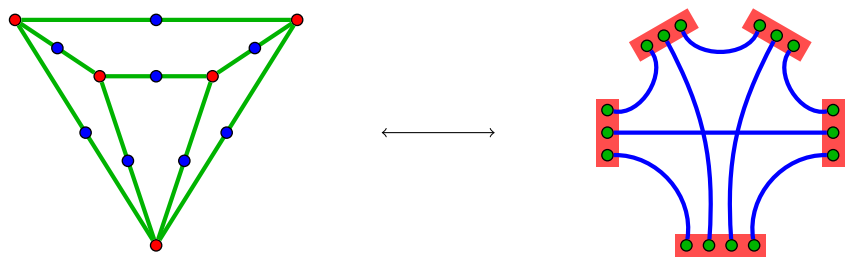
  
\end{remark}
From the Remark \ref{rem:explosion}, the following holds:
\begin{proposition}
  There exists a bijection between the set of spin foam atoms $\sfas$ and the set of stranded atoms $\mathfrak{S}$.
\end{proposition}
One can also bond stranded atoms to form stranded molecules. 
\begin{remark}[{\bf stranded counterparts}]
	The stranded counterparts of various objects take the form:
\begin{itemize}
  \item[--] A \emph{stranded patch} is a coil $c\subset\mathcal{C}$ along with retracings within that coil. 
  \item[--] Two stranded patches are \emph{bondable} if they have the same number of vertices and the same number of retracings.  Knowledge of the retracing are necessary to capture the loop information of a boundary patch. 
  \item[--]A \emph{stranded bonding map} identifies the vertices within two bondable stranded patches, with the compatibility condition that the vertices associated to a retracing in one patch are identified with the vertices associated to a retracing in the other. This is illustrated in Figure \ref{fig:StrandedBonding}.
	 
  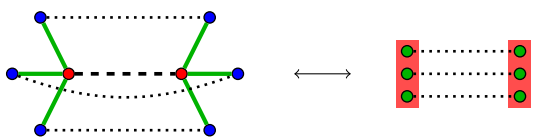
\begin{figure}[htb]
	  \centering
     \tikzsetnextfilename{strandedbonding}
      \begin{tikzpicture}[scale=1.5]
  \draw [<->] (1.5,0) -- (2,0);
\node [vb]		(a)	at (-0.5,0)		{}; 
\node [vh]		(1)	at (-.75,.5)		{};
\node [vh]		(2)	at (-.75,-.5)	{}; 
\node [vh]		(3)	at (-1,0)		{}; 
\node [vb]		(b)	at (0.5,0)		{}; 
\node [vh]		(4)	at (.75,.5)		{};
\node [vh]		(5)	at (.75,-.5)		{}; 
\node [vh]		(6)	at (1,0)		{}; 
\foreach \i/\j in {a/1,a/2,a/3,b/4,b/5,b/6}{
 \draw [eb] (\i) --  (\j);
  }
\path	(1) 	edge [bh] 				(4)
	(2)	edge [bh]		 		(5)
	(3)	edge [bh, bend right=20]	(6)
	(a)	edge [bb]				(b);
  \begin{scope}[xshift=3cm]
  \foreach \i in {0,180}{
    \draw [cs, rotate=\i] (.6,-.3) rectangle (.4,.3);
    }
  \node [vs]	(11)	at (-.5,.2)	{};
  \node [vs]	(12)	at (-.5,0)	{};
  \node [vs]	(13)	at (-.5,-.2)	{};    
  \node [vs]	(21)	at (.5,.2)	{};
  \node [vs	]	(22)	at (.5,0)	{};
  \node [vs]	(23)	at (.5,-.2)	{};  
  \path
  \foreach \i in {1,2,3}{
    (1\i) edge [bh]  (2\i)
    };
  \end{scope}
  \end{tikzpicture}		
	  \caption{Stranded bonding map.}
	  \label{fig:StrandedBonding}
  \end{figure}

  \item[--] A \emph{stranded molecule} is a set of strand atoms quotiented by a set of stranded bonding maps, as drawn in Figure \ref{fig:StrandedMolecule}. 

	  \begin{figure}[htb]
		  \centering
		  \tikzsetnextfilename{strandedmolecule}
      \begin{tikzpicture}[scale=1.5]

\draw [|->] (4.5,0) -- node[label=above:$\sharp_{\gamma}$] {} (5.5,0);

\begin{scope}
\foreach \i in {0,60,120,180}{
\draw [cs, rotate=\i] (1.1,-.3) rectangle (.9,.3);
}
\draw [cs]	 (-.4,-1.1) rectangle (.4,-.9);
\node [vs]		(14)	at (-.35,.95)	{};
\node [vs]		(15)	at (-.5,.86)		{};
\node [vs]		(12)	at (-.65,.77)	{};
\node [vs]		(21)	at (-1,.2)		{};
\node [vs]		(23)	at (-1,0)		{};
\node [vs]		(25)	at (-1,-.2)		{};
\node [vs]		(52)	at (-.3,.-1)		{};
\node [vs]		(51)	at (-.1,-1)		{};
\node [vs]		(54)	at (.1,-1)		{};
\node [vs]		(53)	at (.3,-1)		{};
\node [vs]		(34)	at (1,.2)		{};
\node [vs]		(32)	at (1,0)		{};
\node [vs]		(35)	at (1,-.2)		{};
\node [vs]		(41)	at (.35,.95)	{};
\node [vs]		(45)	at (.5,.86)		{};
\node [vs]		(43)	at (.65,.77)	{};
\path
\foreach \i/\j in {14/41,21/12,43/34}{
  (\i) edge [es,bend right=70] (\j)
  }
\foreach \i/\j in {52/25,35/53}{
  (\i) edge [es,bend right=50] (\j)
  }
\foreach \i/\j in {51/15,45/54}{
  (\i) edge [es,bend right=16] (\j)
  }
(23) edge [es] (32);
\end{scope}

\begin{scope}[xshift=1.5cm]
  \node [vs]	(11)	at (-.5,.2)	{};
  \node [vs]	(12)	at (-.5,0)	{};
  \node [vs]	(13)	at (-.5,-.2)	{};    
  \node [vs]	(21)	at (.5,.2)	{};
  \node [vs	]	(22)	at (.5,0)	{};
  \node [vs]	(23)	at (.5,-.2)	{};  
  \path
  \foreach \i in {1,2,3}{
    (1\i) edge [bh]  (2\i)
    };
   \node at (0,.5) {$\gamma$} ;
\end{scope}

\begin{scope}[xshift=3cm] 
  \foreach \i in {0,90,180,270}{
    \draw [cs, rotate=\i] (1.1,-.3) rectangle (.9,.3);
    }
  \node [vs]	(12)	at (1,.2)	{};
  \node [vs]	(13)	at (1,0)	{};
  \node [vs]	(14)	at (1,-.2)	{};    
  \node [vs]	(21)	at (.2,1)	{};
  \node [vs]	(24)	at (0,1)	{};
  \node [vs]	(23)	at (-.2,1)	{};    
  \node [vs]	(32)	at (-1,.2)	{};
  \node [vs]	(31)	at (-1,0)	{};
  \node [vs]	(34)	at (-1,-.2)	{};    
  \node [vs]	(41)	at (.2,-1)	{};
  \node [vs]	(42)	at (0,-1)	{};
  \node [vs]	(43)	at (-.2,-1)	{};    
  \path
  (13) edge [es] 	 (31)
  (24) edge [es]  (42)  
  \foreach \i/\j in {14/41,21/12,32/23,43/34}{
  (\i) edge [es,bend right=50]  (\j)
  };
  \end{scope}
  
\begin{scope}[xshift=7cm]
\foreach \i in {60,120,180}{
\draw [cs, rotate=\i] (1.1,-.3) rectangle (.9,.3);
}
\draw [cs]	 (-.4,-1.1) rectangle (.4,-.9);
\node [vs]		(14)	at (-.35,.95)	{};
\node [vs]		(15)	at (-.5,.86)		{};
\node [vs]		(12)	at (-.65,.77)	{};
\node [vs]		(21)	at (-1,.2)		{};
\node [vs]		(23)	at (-1,0)		{};
\node [vs]		(25)	at (-1,-.2)		{};
\node [vs]		(52)	at (-.3,.-1)		{};
\node [vs]		(51)	at (-.1,-1)		{};
\node [vs]		(54)	at (.1,-1)		{};
\node [vs]		(53)	at (.3,-1)		{};
\node [vs]		(34)	at (1,.2)		{};
\node [vs]		(32)	at (1,0)		{};
\node [vs]		(35)	at (1,-.2)		{};
\node [vs]		(41)	at (.35,.95)	{};
\node [vs]		(45)	at (.5,.86)		{};
\node [vs]		(43)	at (.65,.77)	{};
\path
\foreach \i/\j in {14/41,21/12,43/34}{
  (\i) edge [es,bend right=70] (\j)
  }
\foreach \i/\j in {52/25,35/53}{
  (\i) edge [es,bend right=50] (\j)
  }
\foreach \i/\j in {51/15,45/54}{
  (\i) edge [es,bend right=16] (\j)
  }
(23) edge [es] (32);
\end{scope}

\begin{scope}[xshift=9cm] 
  \foreach \i in {0,90,180,270}{
    \draw [cs, rotate=\i] (1.1,-.3) rectangle (.9,.3);
    }
  \node [vs]	(12)	at (1,.2)	{};
  \node [vs]	(13)	at (1,0)	{};
  \node [vs]	(14)	at (1,-.2)	{};    
  \node [vs]	(21)	at (.2,1)	{};
  \node [vs]	(24)	at (0,1)	{};
  \node [vs]	(23)	at (-.2,1)	{};    
  \node [vs]	(32)	at (-1,.2)	{};
  \node [vs]	(31)	at (-1,0)	{};
  \node [vs]	(34)	at (-1,-.2)	{};    
  \node [vs]	(41)	at (.2,-1)	{};
  \node [vs]	(42)	at (0,-1)	{};
  \node [vs]	(43)	at (-.2,-1)	{};    
  \path
  (13) edge [es] 	 (31)
  (24) edge [es]  (42)  
  \foreach \i/\j in {14/41,21/12,32/23,43/34}{
  (\i) edge [es,bend right=50]  (\j)
  };
  \end{scope}

\end{tikzpicture}
		  \caption{Stranded molecule.}
		  \label{fig:StrandedMolecule}
	  \end{figure}
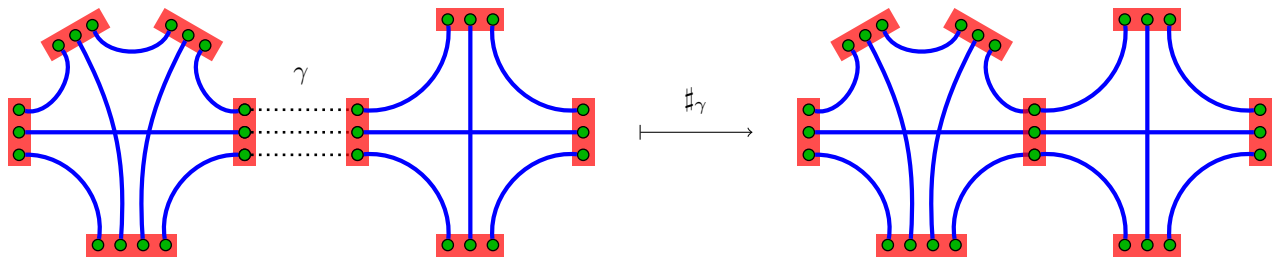

\end{itemize}
\end{remark}
One can translate the concepts such as labelled, loopless, simplicial to the stranded diagram realization. This is left to the interested reader since these structures are not extensively used in the remaining sections.  Having said that, we should also mention that stranded graphs are a natural and powerful tool in the \gft\ formalism.
One particular advantage of stranded diagrams as compared to bondings of boundary graphs is that the full internal bonding structure, including the ordering of bondings of faces along patches, is represented in these diagrams in terms of the strands. This is not possible in bondings of boundary graphs.


\newpage

\section{Group field theories: generating spin foam molecules}
\label{sec:gft}

\renewcommand{\pgft}{\textsc{s-gft}}

Having laid the combinatorial foundations, let us now turn to our main goal: 
\begin{center}
  defining a \gft\ framework that can accommodate, both kinematically and dynamically,  all the states and histories that one might expect to appear in loop quantum gravity.
\end{center}
The route is divided into three parts.  First, we shall summarize some generalities of the \gft\ set--up, with respect to its definition as a quantum field theory generating spin foam molecules.  This will clarify how the graphs supporting \lqg\ states, as well as the complexes supporting spin foam amplitudes, appear in this context. 

Next, we shall outline the class of \gft\ models that are standard in the literature. These are based on a single field and generate series catalogued by a specific subset of the unlabelled spin foam molecules $\sfrs$.  Via the interpretation given in Section \ref{sssec:enhancing}, these are associated to $n$--dimensional simplicial structures.

Finally, we shall generalize the \gft\ framework to incorporate broader classes of models. 
There are two main avenues to follow:
\begin{itemize}
  \item[i)] One can stick with unlabelled structures but attempt to directly generate (larger subsets of) $\sfrs$. In this context, the first generalization is effected simply by broadening the type of interaction terms in the theory while keeping a single field. Such models are already common in the \gft\ literature. \cite{Bonzom:2012bg,\cgftrenorm,\cCOR}

    The second generalization involves passing from a single--field to multi--field group field theory.  In this manner, one can generate all of $\sfrs$, albeit in a rather formal manner, with an infinite set of \gft\ fields.  
    
  \item[ii)]
    One moves over to labelled structures, which permit a much simpler class of \gft s, based on a single \gft\ field over a larger data domain. This data domain, inspired by a standard technique in tensor models known as dual--weighting, allows one to generate dynamically the spin foam molecules in $\sfrst_{n,\Simplicialdually}$.  Drawing upon the results of Section \ref{sssec:primitivea}, one has encoded the molecules in $\sfrs$, at least at the combinatorial level.  This sets the scene for Section \ref{sec:sfm}, where we devise a class of \gft\ models that generate weights for the molecules in $\sfrs$ and that effectively assign to the underlying molecules $\sfrs$ the amplitude expected by the 4d \eprl\  quantum gravity spin foam theory. 
    \end{itemize}
The nomenclature and definitions introduced in the previous section will be used extensively in the following.

\subsection{{\sc gft} generalities}
\label{sec:gftmain}

Let us first recount the general definitions and structures of \gft s, as one finds them in the literature \cite{\cGFT}. 
\begin{defin}[{\bf group field}]
  A \emph{group field}, $\phi$, is a function over a group:
  \begin{equation}
    \label{eq:field}
    \phi:G^{\times\copies}\longrightarrow\Rbb\;,
  \end{equation}
  where $G$ is a group, while $n\in\Nbb$. 
\end{defin}
\begin{defin}[{\bf group field theory}]
  \label{def:gft}
A \emph{group field theory} is a quantum field theory for a group field, defined by a partition function: 
\begin{equation}
 \label{eq:partition}
Z_{\gft} = \int \Dcal\phi\; e^{-S[\phi]}\;,
\end{equation}
where $\Dcal\phi$ denotes a (formal) measure on the space of group fields, while the action functional takes the form:
\begin{equation}
   S[\phi] = \frac12\int [\extd g]\; \phi(g_1)\;\Kbb(g_1, g_2)\;\phi(g_2) + \sum_{i\in I}\lambda_i \int [\extd g]\;\Vbb_i\big(\{g_j\}_{J_i}\big)\;\prod_{j\in J_i}\phi(g_{j})\;.
 \label{eq:action}
\end{equation}
$\Kbb$ is the kinetic kernel, $\Vbb_i$ are vertex (interaction) kernels satisfying \emph{combinatorial non--locality},  while $I$ and $J_i$ are finite sets indexing the interactions and the number of fields in the $i$th interaction, respectively.  Meanwhile, $[\extd g]$ represents the appropriate number of copies of the measure on $G$ and $\{\lambda_i\}_{I}$ is the set of coupling constants.\footnote{There is an analogous set of actions for complex group fields and of course, one can define models involving several such fields.}  

\end{defin}
\begin{remark}[{\bf kinetic kernel}]
  The kinetic kernel is a real function with domain $G^{\times 2n}$ that (in some model dependent manner) pairs arguments according to $(g_{1a},g_{2a})$ with $a \in \{1,\dots, n\}$:
  \begin{equation}
    \Kbb(g_1, g_2)  = \Kbb(g_{11},g_{21};\dots;g_{1n},g_{2n})
  \end{equation}
\end{remark}
\begin{remark}[{\bf vertex kernels and combinatorial non--locality}]
  \label{rem:CombinNonLocal}
  {\it Combinatorial non--locality} is a property possessed by \gft\ interaction kernels, effected through  pairwise convolution of the field arguments.  It is the main peculiarity of \gft s with respect to local quantum field theories on space--time. In more detail, the \gft\ interaction kernels do not impose coincidence of the points,  in the group space $G^{\times\copies}$,  at which the interaction fields are evaluated.  Rather, the totality of field arguments from the smaller group space $G$ occurring in a given action term (that is $\copies\times|J|$ for an interaction term with $|J|$ group fields) is partitioned into pairs and the kernels convolve such pairs: 
  \begin{equation}
  \Vbb\big(\{g_j\}_J\big)  = \Vbb\big( \{g_{ja}g_{kb}^{-1}\}\big)
  \end{equation}
  where $j,k\in J$, $a,b\in\{1,\dots,\copies\}$ and $(ja, kb)$ is an element of the pairwise partition of the set $J\times\{1,\dots,\copies\}$. 
  The specific combinatorial pattern of such pairings determines the combinatorial structure of the Feynman diagrams of the theory.  It will be one of the main foci in later discussions, 
both in the standard \gft\ models and, later on, in the generalized class of models.
\end{remark}

Besides this combinatorial peculiarity, one deals with \gft s as one would any other QFT; the main features follow.

\begin{defin}[{\bf quantum observables}]
  \label{def:QuantumObservables}
  \emph{(Quantum) observables}, $O[\phi]$,  are functionals of the group field. 
\end{defin} 
  In particular, the kinetic and interaction terms are quantum observables. Due to their functional form, they motivate interest in a subset of polynomial functionals of the field:
  \begin{defin}[{\bf \trace\ observables}] A \emph{\trace\ observable} is a polynomial functional of the group field that satisfies combinatorial non--locality (since all group elements are traced over pairwise).  Thus, they have the generic form:
\begin{equation}
  \label{eq:obsdef}
  O[\phi] \equiv \int [dg]\; \bbB\big(\{g_j\}_{J}\big)\; \prod_{j\in J} \phi(g_j)\;,\quad\quad\textrm{where}
  \quad\quad \bbB\big(\{g_{j}\}_{J}\big) = \bbB\big(\{g_{ja}g_{kb}^{-1}\}\big)
\end{equation}
and $(ja,kb)$ is an element of the pairwise partition of the set $J\times \{1,\dots,n\}$.
\end{defin}

\begin{remark}[{\bf estimating observables}]
  Expectation values of quantum observables are estimated using perturbative techniques.  For example, the observable $O[\phi]$, expanded with respect to the coupling constants $\{\lambda_i\}_I$, leads to a series of Gaussian integrals evaluated through Wick contraction.  The patterns of contractions are catalogued by Feynman diagrams:
\begin{eqnarray}
  \langle O\rangle_{\gft} &=& \frac{1}{Z_{\gft}}\int \Dcal\phi\; O[\phi]\;e^{-S[\phi]} \nonumber \\
  &=& \frac{1}{Z_{\gft}} \int \Dcal\phi\;O[\phi] \sum_{\{c_i\}_I} \prod_{i\in I}\frac{1}{c_i!}\Bigg[ \lambda_i \int [\extd g]\;\Vbb_i\Big(\{g_j\}_{J_i}\Big)\;\prod_{j\in J_i}\phi(g_{j})\Bigg]^{c_i} e^{-\frac12\int[\extd g]\;\phi(g_1)\;\Kbb(g_1,g_2)\;\phi(g_2)}\nonumber\\
&=& \sum_{\fdiagram} \frac{1}{C(\fdiagram)} A(\fdiagram;\{\lambda_i\}_I)\;,
\label{eq:expand}
\end{eqnarray}
where $C(\fdiagram)$ are the combinatorial factors related to the automorphism group of the Feynman diagram $\fdiagram$  and $A(\fdiagram;\{\lambda_i\}_I)$ is the weight of $\fdiagram$ in the series. The Feynman amplitudes $A(\fdiagram)$ are constructed by convolving (in group space) propagators $\Pbb = \Kbb^{-1}$ and interaction kernels. In this section, however, the focus lies solely on the combinatorial aspects of the \gft\ perturbative expansion. Discussion of specific models is postponed to Section \ref{sec:sfm}. 
\end{remark}
\begin{remark}[{\bf stranded diagrams}]
  The stranded diagram representation of the Feynman diagrams $\fdiagram$ is immediate. With reference to Section \ref{ssec:stranded},  one associates a coil $c$,  with $\copies$ vertices to each field $\phi$.
  
  In an interaction term, the fields represent a set of coils $\Ccal$, while the combinatorial non--locality property of the interaction kernel encodes the set of reroutings $\Rcal$.  Thus, each interaction term represents a stranded atom $\sta = (\Ccal, \Rcal)$. 

The kinetic term, through its involvement in the Wick contraction, is responsible for the bonding of these stranded atoms. Then, the perturbative expansion is quite clearly catalogued by stranded molecules. 

Through the bijection outlined in Section \ref{ssec:stranded}, one could now map to spin foam atoms and molecules. 
\end{remark}
\begin{remark}[{\bf quantum geometric interpretation}]
  In Section \ref{sec:sfm}, we shall concentrate our attention on the \eprl\ quantum gravity \gft.
  However, we provide some interpretation here for \gft s as models of quantum or random geometry.  The components of a \gft\  have already been understood  in terms of topological structures, primarily in two dimensions, but also secondarily in $D$ dimensions (although this enhancement is a subtle issue about which we have made some comments in Section \ref{sssec:enhancing}).  
  
  Keeping to $D$--dimensional language, the group fields correspond to $(D-1)$--dimensional building blocks of $(D-1)$--dimensional topological structures, the \trace\ observables. In a similar manner, the interaction terms in the action correspond to the $D$--dimensional building blocks for $D$--dimensional topological structures cataloguing the terms of the perturbative expansions. 

  Then, the estimating of observables $\langle O_{1}\dots O_{l}\rangle$ via perturbative expansion, yields a sum over $D$--dimensional topological structures, whose boundaries are precisely the $l$ $(D-1)$--dimensional structures encoded by observables. In other words, one is calculating the correlation of the $l$ $(D-1)$--dimensional structures. 

The intention of both the data contained in the group $G$ and the kernels (boundary $\bbB$, kinetic $\Kbb$ and interaction $\Vbb$) is to transform all these topological statements above into quantum geometrical ones.  More precisely, using results from loop quantum gravity, as well as lattice quantum gravity, depending on the precise realization of the data set, it may be interpreted as one of the following: the discrete gravitational connection; the discrete fluxes of the conjugate triad; or the eigenvalues of fundamental quantum geometric operators like areas and volumes. 
\end{remark}

\subsection{Combinatorial correspondence}
\label{ssec:CombCorr}
Let us recast this \gft\ formalism in terms of the combinatorial structures detailed in Section \ref{sec:comb}:  
  \begin{itemize}
    \item[--] The set of group fields is indexed by the set of patches:
      \begin{equation}
	\Phi = \{\phi_{\bp}\}_{\bps}\;,\quad\quad\textrm{where}\quad\quad \phi_{\bp}:G^{\times |\mathcal{E}_{\bp}|} \longrightarrow \Rbb\;,
      \end{equation}
      and $\bp = (\{\vbar\}\cup\Vhat_{\bp},\mathcal{E}_{\bp})$. 
\item[--]  The set of \trace\ observables is indexed by the set of bisected boundary graphs:
  \begin{equation}
    \Ocal = \{O_{\bbg}\}_{\bbgs}\;,\quad\quad\textrm{where}\quad\quad O_{\bbg}[\Phi] = \int [\extd g]\; \bbB_{\bbg}\big(\{g_{\vbar}\}_{\Vbar}\big)\prod_{\vbar\in\Vbar} \phi_{\bp_{\vbar}}(g_{\vbar})
  \end{equation}
  and $\bbg = (\Vbar\cup\Vhat, \mathcal{E}_{\bbg})$.  The patches of $\bbg$ are in correspondence with vertices of $\Vbar$ and one has that $g_{\vbar} = \{g_{\vbar\vhat}: (\vbar\vhat)\in\mathcal{E}_{\bbg}\}$. Combinatorial non--locality is realized using the bisecting vertices $\Vhat$. Each such vertex has a pair of incident edges and thus they encode a pairwise partition of the data set $\{g_{\vbar}\}_{\Vbar}$. Conversely, a pairwise partition of this data set determines a graph $\bbg$.   
  Thus, the graphs in $\bbgs$ catalogue the combinatorially non--local configurations.

\item[--] Likewise, the set of vertex interactions is indexed by $\bbgs$: 
  \begin{equation}
    \lambda_{\bbg}\int [\extd g]\; \Vbb_{\bbg}\big(\{g_{\vbar}\}_{\Vbar}\big)\prod_{\vbar\in\Vbar} \phi_{\bp_{\vbar}}(g_{\vbar})\;.
  \end{equation}
  As a result of the bijection in Proposition \ref{prop:bijectionAtom}, the interaction terms can be interpreted as generating spin foam atoms $\sfa = \alpha(\bbg)$.

\item[--]
  The kinetic term, through its role in the Wick contractions occurring in later perturbative expansions,  is responsible for the bonding of patches compatible according to the compatibility condition of Definition \ref{def:bonding},
$\bp\equiv\bp_1\cong\bp_2$:  
  \begin{equation}
  \label{eq:kinetic}
    \frac12\int[\extd g]\; \phi_{\bp}(g_{\vbar_1})\;\Kbb_{\bp}(g_{\vbar_1}, g_{\vbar_2})\;\phi_{\bp}(g_{\vbar_2})\;,\quad\quad\textrm{where}\quad\quad \Kbb_{\bp}(g_{\vbar_1},g_{\vbar_2}) = \Kbb(\{g_{\vbar_1\vhat}, g_{\vbar_2\vhat} \})\;
  \end{equation}
  is a function of group elements for each $(\vbar_i\vhat)\in\mathcal{E}_{\bp}$.
\item[--]
  Then, generic models are defined via:
  \begin{equation}
    Z = \int \Dcal\Phi \;e^{-S[\Phi]}
  \end{equation}
  with:
  \begin{equation}
    S[\Phi] = \frac12\int[\extd g]\; \phi_{\bp}(g_{\vbar_1})\;\Kbb_{\bp}(g_{\vbar_1}, g_{\vbar_2})\;\phi_{\bp}(g_{\vbar_2}) + 
    \sum_{\bbg\in\bbgs}\lambda_{\bbg}\int [\extd g]\; \Vbb_{\bbg}\big(\{g_{\vbar}\}_{\Vbar}\big)\prod_{\vbar\in\Vbar} \phi_{\bp_{\vbar}}(g_{\vbar})\;.
  \end{equation}
  
\item[--]
Sums and products of \trace\ observables can be estimated perturbatively, generating series of the type:
\begin{eqnarray}
  \langle O_{\bbg_1}\dots O_{\bbg_l}\rangle &=& \frac{1}{Z}\int \Dcal\Phi\; O_{\bbg_1}[\Phi]\dots O_{\bbg_l}[\Phi]\;e^{-S[\Phi]} \nonumber \\
  &=& \sum_{\substack{\sfr\in\sfrs\\[0.1cm] \delta(\sfr) = \sqcup_{i = 1}^{l}\bbg_{i}}} \frac{1}{C(\sfr)} A(\sfr;\{\lambda_{\bbg}\}_{\bbgs})\;.
\end{eqnarray}
Thus, the Feynman diagrams generated by \gft s are actually better characterized as spin foam molecules. 
\end{itemize}
Using the above index, one can catalogue the generalized classes of \gft\ models that make contact with the set of spin foam molecules $\sfrs$. This will be done in later sections.

\begin{remark}[{\bf generalization and control}]
  It is worth noting some motivations for considering such generalized \gft\ models: 
  \begin{itemize}
    \item[--]
As one can see above, there is no technical obstacle whatsoever, within the \gft\ formalism,  to passing from a single--field \gft\ to a multi--field \gft\ (indexed by some set of patches) and/or stimulating new interaction terms (indexed by some set of bisected boundary graphs).  Such choices generate broader classes of spin foam molecules, as one might wish from an \lqg\ perspective. 
  
Given the facility with which such generalized \gft s are defined, a real issue is rather to pinpoint some criterion, for selecting one model over another. Other important issues centre on settling  \textit{i}) whether or not one is able to control analytically or numerically the dynamics of such generalized \gft s  and  \textit{ii}) whether or not such control is improved by one choice of combinatorics over another. Indeed, these issues should also be posed from the spin foam perspective. 

A common choice in the spin foam and \gft\ literature is to restrict to spin foam atoms and molecules with a $D$--dimensional simplicial interpretation.  This choice could be motivated as being more  \lq fundamental\rq, in the sense that one can triangulate more general complexes but not vice versa, and as being simpler than other alternatives.

\item[--]
  Moreover, generalized \gft s already exist in the literature. Indeed, so--called \emph{invariant tensor models}, which are in essence single--field \gft s with a specific subset of generalized interactions \cite{Bonzom:2012bg}, have been the setting for most studies on \gft\ renormalization \cite{\cgftrenorm, \cCOR} and for analysis using tensor model techniques \cite{Gurau:2012hl}. 


\item[--]
Finally, even in models starting with simplicial interactions only, one should expect the quantum dynamics to generate new effective interactions with generalized combinatorics. In turn, these new interaction terms should then be taken into account in the renormalization flow of the simplicial models. Again, the issue is not whether such combinatorial generalizations can be considered, but how one should deal with them in the quantum dynamics of the theory.
\end{itemize}
\end{remark}

\subsection{Simplicial {\sc gft}}
\label{ssec:simplicialgft}
For a moment, let us focus on the \gft\ corresponding to the unlabelled, $n$--regular, simplicial structures: $\bbgs_{n,\Simplicial}$, $\sfas_{n,\Simplicial}$ and $\sfrs_{n,\Simplicial}$ from Section \ref{sssec:special}. As shown in Section \ref{sssec:enhancing}, such structures have a simplicial interpretation. They correspond to a particularly simple choice of combinatorics for the \gft\ action and represent a class of models that are by far the most used in the quantum gravity literature. 

The parameter $n$ is set to the dimension $D$ of the space--time to be reconstructed via the \gft\ dynamics. 
\begin{itemize}
  \item[--] The group field corresponds to the unique unlabelled 
   $D$--patch $\bp_D$:
    \begin{equation}
      \phi\equiv\phi_{\bp_D}:G^D \longrightarrow \Rbb\;.
    \end{equation}
  \item[--] The pairing of field arguments in the interaction kernel is based upon the unique unlabelled simplicial $D$--graph $\bbg\in\bbgs_{D,\Simplicial}$ (that is, $K_{D+1}$, the complete graph over $D+1$ vertices), which allows one to abbreviate notation:
    \begin{equation}
      \label{eq:SimplicialVertex}
      \Vbb_{\bbg}(g) =  \Vbb(\{g_{ij}g_{ji}^{-1}\})\;, \quad\quad \textrm{with}\quad\quad i<j\;.
    \end{equation}
    Henceforth, when dealing with graphs based upon $K_{D+1}$, the markers $i,j\in\{1,\dots,D+1\}$ index the $D+1$ vertices $\Vbar$ and thus the patches of $\bbg$. The bisecting vertices are labelled by $(ij)$.\footnote{The parenthesis signifies that both $(ij)$ and $(ji)$ mark the same bisecting vertex.}  The edge joining the vertex $i$ to the vertex $(ij)$ is denoted by $ij$, while the edge joining the vertex $j$ to the vertex $(ij)$ is denoted by $ji$.    

  \item[--] In the kinetic kernel, the data indices are abbreviated to $g_1\equiv g_{\vbar_1}$ and $g_2\equiv g_{\vbar_2}$. 

  \item[--] The action is therefore specified by:
    \begin{equation}
      \label{eq:SimplicialAction}
S[\phi] = \frac12\int [\extd g]\; \phi(g_1)\;\Kbb(g_1, g_2)\;\phi(g_2) 
+\lambda \int [\extd g]\; \Vbb_{\bbg}(g) \prod_{j=1}^{D+1}\phi(g_j).
    \end{equation}
    up to the precise form of the kinetic and interaction kernels. 
\item[--] There is a distinguished subclass of \trace\ observables indexed by $\bbgs_{D,\loopless}$, the unlabelled $D$--regular loopless graphs. This stems from the property that each graph in $\bbgs_{D,\loopless}$ arises as the boundary of some spin foam molecules in $\sfrs_{D,\Simplicial}$, while the boundary of every spin foam molecule in $\sfrs_{D,\Simplicial}$ is a collection of graphs in $\bbgs_{D,\loopless}$.
\item[--] The perturbative expansion of the partition function (the \gft\ vacuum expectation value) leads to a series catalogued by saturated spin foam molecules $\sfr\in\sfrs_{D,\Simplicial}$, $\delta(\sfr) = \emptyset$. Meanwhile, the evaluation of a generic observable $O_{\bbg}[\phi]$ leads to a series catalogued by spin foam molecules with boundary $\bbg$, that is: $\sfr\in\sfrs_{D,\Simplicial}$ with $\delta(\sfr) = \bbg$. 
\end{itemize} 

The combinatorics of the propagator and the simplicial vertex kernel are illustrated in Figures \ref{fig:gftgraphs} and \ref{fig:dualtet} in the 3--dimensional case. Therein is drawn both the bisected boundary graph realization, alongside the usual stranded diagram representation.

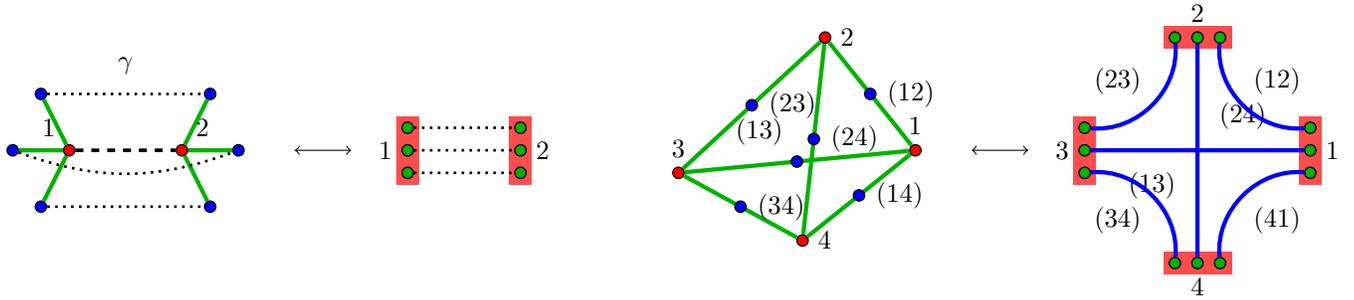
\begin{figure}[htb]
  \centering
  \tikzsetnextfilename{gftgraphs}
  \begin{tikzpicture}[scale=1.5]
  \draw [<->] (1.5,0) -- (2,0);
\node [vb, label=135:1]	(a)	at (-0.5,0)		{}; 
\node [vh]		(1)	at (-.75,.5)		{};
\node [vh]		(2)	at (-.75,-.5)	{}; 
\node [vh]		(3)	at (-1,0)		{}; 
\node [vb, label=45:2]	(b)	at (0.5,0)		{}; 
\node [vh]		(4)	at (.75,.5)		{};
\node [vh]		(5)	at (.75,-.5)		{}; 
\node [vh]		(6)	at (1,0)		{}; 
\foreach \i/\j in {a/1,a/2,a/3,b/4,b/5,b/6}{
 \draw [eb] (\i) --  (\j);
  }
\path	(1) 	edge [bh] node[label=above:$\gamma$] {} (4)
	(2)	edge [bh]		 		(5)
	(3)	edge [bh, bend right=20]	(6)
	(a)	edge [bb]				(b);
  \begin{scope}[xshift=3cm]
  \foreach \i in {0,180}{
    \draw [cs, rotate=\i] (.6,-.3) rectangle (.4,.3);
    }
  \node [vs]			(11)	at (-.5,.2)	{};
  \node [vs, label=180:1]	(12)	at (-.5,0)	{};
  \node [vs]			(13)	at (-.5,-.2)	{};    
  \node [vs]			(21)	at (.5,.2)	{};
  \node [vs, label=0:2]	(22)	at (.5,0)	{};
  \node [vs]			(23)	at (.5,-.2)	{};  
  \path
  \foreach \i in {1,2,3}{
    (1\i) edge [bh]  (2\i)
    };
  \end{scope}

  \begin{scope}[xshift=6cm]
    \draw [<->] (1.5,0) -- (2,0);
  
  \node [vb, label=90:3]	(3)	at (-1.1,-.2){};
  \node [vb, label=0:4]	(4)	at (0,-.8)	{};
  \node [vb, label=90:1]	(1)	at (1,0) 	{};
  \node [vb, label=-0:2]	(2)	at (.2,1)	{};
  \foreach \i/\j in {1/2,2/4,3/4,2/3,1/4}{
    \draw [eb] (\i) -- node[vh, label=0:(\i\j)] {} (\j);
    }
  \draw [eb] (1) -- node[vh, label=120:(13)] {} (3);
  \end{scope}  
  \begin{scope}[xshift=9.5cm] 
  \foreach \i in {0,90,180,270}{
    \draw [cs, rotate=\i] (1.1,-.3) rectangle (.9,.3);
    }
  \node [vs]			(12)	at (1,.2)	{};
  \node [vs, label=0:1]	(13)	at (1,0)	{};
  \node [vs]			(14)	at (1,-.2)	{};    
  \node [vs]			(21)	at (.2,1)	{};
  \node [vs, label=90:2]	(24)	at (0,1)	{};
  \node [vs]			(23)	at (-.2,1)	{};    
  \node [vs]		 	(32)	at (-1,.2)	{};
  \node [vs, label=180:3]	(31)	at (-1,0)	{};
  \node [vs]			(34)	at (-1,-.2)	{};    
  \node [vs]			(41)	at (.2,-1)	{};
  \node [vs, label=270:4]	(42)	at (0,-1)	{};
  \node [vs]			(43)	at (-.2,-1)	{};    
  \path
  (13) edge [es] 	node[label=225:(13)]  {}	 (31)
  (24) edge [es] 	node[label=45:(24)]  {}	 (42)  
  \foreach \i/\j in {14/41,21/12,32/23,43/34}{
  (\i) edge [es,bend right=50] node [auto] {(\j)} (\j)
  };
  \end{scope}
  \end{tikzpicture} 
  \caption{\label{fig:gftgraphs}Equivalent representation of combinatorics of propagator and simplicial interaction for a $D=3$ \gft\ in terms of bisected boundary graphs and in terms of the common stranded diagrams.}
\end{figure}

\begin{figure}[htb]
  \centering
  \tikzsetnextfilename{dualtet}
  \begin{tikzpicture}[scale=2]
 \draw [<->] (1.2,0) -- (1.7,0);
 \draw [<->] (-1.2,0) -- (-1.7,0);

  \begin{scope}[xshift=-3cm]
   \node [vb, label=90:3]	(3)	at (-1.1,-.2){};
  \node [vb, label=0:4]	(4)	at (0,-.8)	{};
  \node [vb, label=90:1]	(1)	at (1,0) 	{};
  \node [vb, label=-0:2]	(2)	at (.2,1)	{};
  \foreach \i/\j in {1/2,2/4,3/4,2/3,1/4}{
    \draw [eb] (\i) -- node[vh, label=0:(\i\j)] {} (\j);
    }
  \draw [eb] (1) -- node[vh, label=120:(13)] {} (3);
  \end{scope}

\begin{scope}
\node [c]					(v)	at (0,-.19)		{};
\node [c,label=0:1]			(1)	at (.5,0)	 	{}; 
\node [c,label=90:2]			(2)	at (-.12,.12)	{}; 
\node [c,label=180:3]		(3)	at (-.32,-.14)	{}; 
\node [c,label=-90:4]			(4)	at (0,-.84)		{};
\node [c,label=0:(12)]		(12)	at (.55,.55)	{};
\node [c,label=90:(13)]		(13)	at (.21,0)		{};
\node [c,label=0:(14)]		(14)	at (.77,-.83)	{};
\node [c,label=135:(23)]		(23)	at (-.73,.31)	{};
\node [c
				] 		(24)	at (-.17,-.51)	{};
\node [c,label=200:(34)]		(34)	at (-.5,-1.06)	{};
\foreach \i/\j in {1/2,1/3,1/4,2/3,2/4,3/4}{
 \path	[f] 	(\i) -- (\i\j) -- (\j) -- (v) -- cycle;
 }
 \foreach \i in {1,2,3,4}{
  \draw [e] (\i) -- (v);
  }
\foreach \i/\j in {1/2,1/3,1/4,2/3,2/4,3/4}{
 \draw 	[eh]	(v)		-- 	(\i\j);
 \draw	[eb] 	(\i) node[vb] {} -- (\i\j) node[vh] {};
 \draw 	[eb] 	(\j) node[vb] {} -- (\i\j) node[vh] {};
 }
\draw [e] (3) node[vb] {} -- (v) node[v] {};

\node [c,label=-90:(24)] 		(24)	at (-.17,-.51)	{};
\end{scope}

\begin{scope}[xshift=3.cm]
\node [c]					(v)	at (0,-.19)		{};
\node [c,label=0:1]			(1)	at (.5,0)	 	{}; 
\node [c,label=90:2]			(2)	at (-.12,.12)	{}; 
\node [c,label=180:3]		(3)	at (-.32,-.14)	{}; 
\node [c,label=-90:4]			(4)	at (0,-.84)		{};
\node [c,label=0:(12)]		(12)	at (.55,.55)	{};
\node [c
				]		(13)	at (.21,0)		{};
\node [c,label=0:(14)]		(14)	at (.77,-.83)	{};
\node [c,label=135:(23)]		(23)	at (-.73,.31)	{};
\node [c
				] 		(24)	at (-.17,-.51)	{};
\node [c,label=200:(34)]		(34)	at (-.5,-1.06)	{};
\foreach \i/\j in {1/2,1/3,1/4,2/3,2/4,3/4}{
 \path	[f] 	(\i) -- (\i\j) -- (\j) -- (v) -- cycle;
 }
 \foreach \i in {1,2,3,4}{
  \draw [e] (\i) -- (v);
  }
\foreach \i/\j in {1/2,1/3,1/4,2/3,2/4,3/4}{
 \draw 	[eh]	(v)		-- 	(\i\j);
 \draw	[eb] 	(\i) node[vb] {} -- (\i\j) node[vh] {};
 \draw 	[eb] 	(\j) node[vb] {} -- (\i\j) node[vh] {};
 }
\draw [e] (3) node[vb] {} -- (v) node[v] {};
\draw (0,1.35)		-- (.43,-1.37);
\draw (0,1.35)		-- (1.11,-.27);
\draw (0,1.35)		-- (-1.44,-.73);
\draw (.43,-1.37)	-- (1.11,-.27);
\draw (.43,-1.37)	-- (-1.44,-.73);
\draw (1.11,-.27)	-- (-1.44,-.73);
\end{scope}

\end{tikzpicture}
\caption{\label{fig:dualtet} In $D=3$, the neighbourhood of a vertex $v$ (a spin foam atom) within a spin foam molecule $\sfr$ and finally seen as dual to a tetrahedron.}
\end{figure}
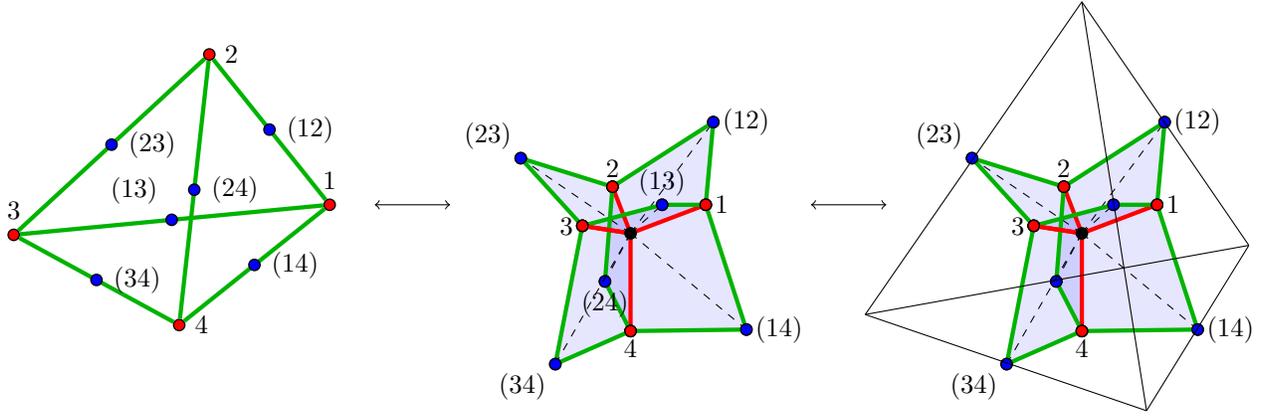

\begin{remark}
  To translate the points made in Section \ref{sssec:enhancing} into \gft\ language, one begins by noting that the spin foam molecules $\sfrs_{D,\Simplicial}$ are interpretable as locally simplicial in $D$ dimensions.  Thus, the group field corresponds to a $(D-1)$--simplex,  the interaction term corresponds to a $D$--simplex, while the kinetic term, through its role in Wick contraction, corresponds to the gluing of $D$--simplices along shared $(D-1)$--simplices.

  Note that the \gft\ action prescribes only the bonding of the spin foam atoms along patches. This corresponds to rules for identifying boundary $(D-1)$-- and $(D-2)$--simplices.  It does not specify uniquely the gluing rules for the full $D$--dimensional information. 
As mentioned in Section \ref{sssec:enhancing}, there are two ways around this limitation.  The first is to add information by hand, which is rather unsatisfactory. The second is to restrict to the so--called \emph{colored} structures $\sfrs_{n,coloured} \subset\sfrs_{n,\Simplicial}$. It is much more natural from the \gft\ point of view, since for any given (simplicial) \gft\ model generating series catalogued by elements of $\sfrs_{n,\Simplicial}$, there is an associated model generating the restricted subclass $\sfrs_{n,coloured}$.  See the review \cite{Gurau:2012hl} for details. 
\end{remark}
 
\begin{remark}
  One could generalize the class of interaction terms to include those based on graphs from the set $\bbgs_{D,\loopless}$. Since these are composed of unlabelled $D$--patches, the \gft\ remains dependent on a single group field:
  \begin{equation}
 \label{eq:SimplicialActionGeneralized}
S[\phi] = \frac12\int [\extd g]\; \phi(g_1)\;\Kbb(g_1, g_2)\;\phi(g_2) 
+\sum_{\bbg\in\bbgs_{D,\loopless}}\lambda_{\bbg} \int [\extd g]\; \Vbb_{\bbg}(g) \prod_{\vbar\in\Vbar}\phi(g_{\vbar})\;,
  \end{equation}
  where $\bbg = (\Vbar\cup\Vhat, \mathcal{E}_{\bbg})$.
Understanding the group field $\phi$ once more as a $(D-1)$--simplex, the spin foam atoms could still be given the interpretation of encoding $D$--dimensional building blocks with locally simplicial $(D-1)$--dimensional boundaries. All spin foam molecules generated by this \gft\ have boundaries in $\bbgs_{D,\loopless}$.  

  Thus, it is clear that the class of models specified by \eqref{eq:SimplicialActionGeneralized} is inadequate for the purposes of generating a dynamics for all \lqg\ quantum states with support in the larger space $\bbgs$. 
\end{remark}

\subsection{Multi--field group field theory}

An obvious strategy for generating series catalogued by (larger subsets of) $\sfrs$ is simply to increase the number of field species entering the model.  Such a scenario was already anticipated at the outset of the group field theory approach to spin foams \cite{Reisenberger:2000fj,Reisenberger:2001hd}. However, from a field theoretic viewpoint, it is a rather unattractive strategy, since the more one wishes to probe quantum states on arbitrary boundary graphs in $\bbgs$, the larger the number of field species and interaction terms required.  Thus, the resulting formalism is not easily controlled using QFT methods.  Having said that, with appropriate kinetic and interaction kernels, multi--field \gft s weight these broader classes of spin foam molecules in the same manner as the generalized constructions one finds in the spin foam literature.  As a result, these \gft\ models are at the same level of formality.
We illustrate multi--field \gft s here simply because we wish to demonstrate the absence of any  impediment  {\it in principle}  to having a \gft\ formulation for the quantum dynamics of all \lqg\ states.

A multi--field group field theory is devised in the following manner:
\begin{itemize}
  \item[--] A subset of group fields $\Phi_{\sub}\subseteq \Phi$ is indexed by a subset of patches $\bps_{\sub}\subseteq \bps$:  
    \begin{equation}
      \label{eq:multifield}
      \Phi_{\sub} = \{\phi_{\bp}\}_{\bps_{\sub}}
    \end{equation}
  \item[--]
    A distinguished class of \trace\ observables $\Ocal_{\sub}\subseteq\Ocal$ is indexed by $\bbgs_{\sub} = \sigma(\bps_{\sub})\subseteq \bbgs$, the bisected boundary graphs generated by $\bps_{\sub}$:
    \begin{equation}
      \label{eq:multiobservable}
      \Ocal_{\sub}  =\{O_{\bbg}\}_{\bbgs_{\sub}}
    \end{equation}
    In particular, observables of this type can be utilized as interaction terms in the action.

  \item[--]
A class of action functionals is then specified by:
\begin{equation}
  \label{eq:multiaction}
  S[\Phi_{\sub}] 
  = \sum_{\bp\in \bps_{\sub}} \int [\extd g]\; \phi_{\bp}(g_1)\; \Kbb_{\bp}(g_1,g_2)\;\phi_{\bp}(g_2) 
  + \sum_{\bbg\in\bbgs_{\sub}}  \lambda_{\bbg} \int [\extd g]\;\Vbb_{\bbg}\Big(\{g_{\vbar}\}_{\Vbar}\Big)\;\prod_{\vbar\in\Vbar}\phi_{\bp}(g_{\vbar})\;,
\end{equation}
where $\bbg = (\Vbar\cup\Vhat,\mathcal{E}_{\bbg})$.

\item[--]The expectation value of an arbitrary product of observables takes the form:
\begin{equation}
  \label{eq:multipartition}
  \langle O_{\bbg_1}\dots O_{\bbg_l}\rangle_{\mfgft} =  \int \Dcal\Phi_{\sub}\; O_{\bbg_1}[\Phi]\dots O_{\bbg_l}[\Phi]\; e^{-S[\Phi_{\sub}]}  = \sum_{\substack{\sfr\in\sfrs_{\sub}\\ \delta(\sfr) = \sqcup_{i = 1}^l \bbg_i}} \frac{1}{C(\sfr)} A(\sfr;\{\lambda_{\bbg}\}_{\bbgs_{\sub}})\;,
\end{equation}

\end{itemize}
\begin{remark}[{\bf higher--dimensional interpretation}]
  In this multi--field setting, one has lost the natural connection between a class of models and a particular value of $D$, the dimension of the reconstructed space--time.  Without doubt, it is difficult to identify precisely generalized classes of spin foam molecules, such that the reconstruction of a D--complex is always possible (and unique). In this non--simplicial setting, the restriction to \emph{coloured} structures is not available (to the best of our knowledge). Moreover, the set of gluing rules that one would need to specify at the outset grows with the generality of the boundary graphs and spin foam atoms.
\end{remark}
\begin{remark}[{\bf 3--dimensional example}]
  \label{rem:3dexample}
  Let us consider a particular multi--field \gft\ model and attempt to provide it with a 3--dimensional interpretation:
  \begin{itemize}
    \item[--]
      It is based on unlabelled $n$--patches, with $3\leq n \leq L$:
      \begin{equation}
      \bps_{\sub} = \{\bp_{n}:3\leq n\leq L\}\;.
    \end{equation} 
    Then, the group field $\phi_{\bp_n}$ could be viewed as representing an 2--dimensional $n$--gon.  
  
  \item[--] 
    A distinguished class of \trace\ observables is indexed by $\bbgs_{\sub} = \sigma(\bps_{\sub})$ and they may be interpreted as surfaces composed of polygons (as we have already stressed, reconstructing these surfaces is a subtle topic and extra information must be put in by hand).
    As a specific example, consider the following \trace\ observable:
    \begin{equation}
\label{eq:pyramid}
O_{\bbg}[\phi_{\bp_3},\phi_{\bp_4}] 
= \int [\extd g]\; \bbB_{\bbg}(g)\; \phi_{\bp_4}(g_1)\,\phi_{\bp_3}(g_2)\,\phi_{\bp_3}(g_3)\,\phi_{\bp_3}(g_4)\,\phi_{\bp_3}(g_5) 
\end{equation}
where
\begin{equation}
\bbB_{\bbg}(g) = \bbB(g_{12}g_{21}^{-1},g_{23}g_{32}^{-1},g_{34}g_{43}^{-1},g_{14}g_{41}^{-1},g_{15}g_{51}^{-1},g_{25}g_{52}^{-1},g_{35}g_{53}^{-1},g_{45}g_{54}^{-1})\;.
\end{equation}
As illustrated in Figure \ref{fig:pyramid}, one could associate a pyramid with a square base to the graph $\bbg$. In this case, the spin foam atom is simply a 3--ball.

\item[--] An action functional of the type given in \eqref{eq:multiaction}, along with the partition function \eqref{eq:multipartition} generate spin foam molecules that may be interpreted as 3--dimensional objects composed of such building blocks.  
\end{itemize}

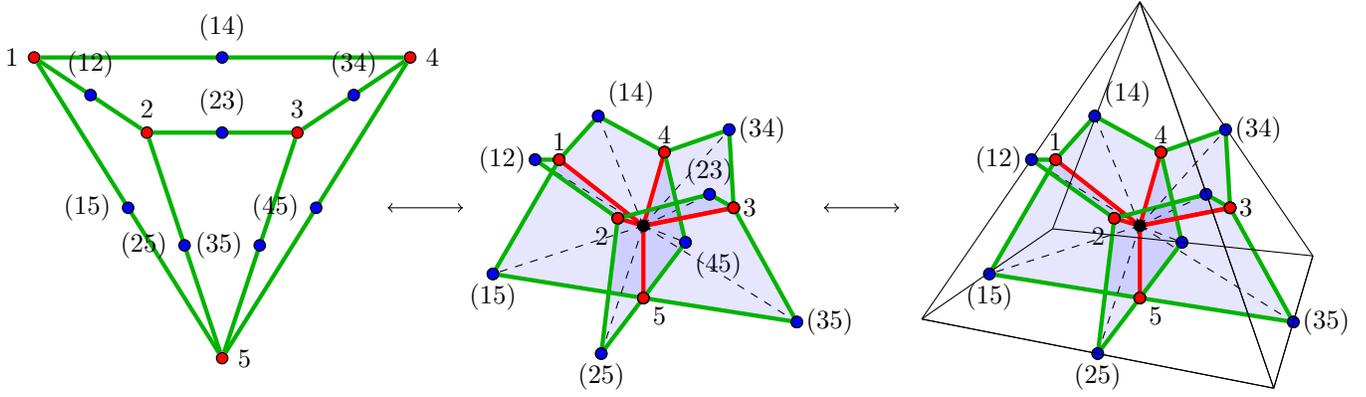
\begin{figure}[htp]
    \centering
     \tikzsetnextfilename{pyramid}
      \begin{tikzpicture}[scale=2]
 \draw [<->] (1.2,0) -- (1.7,0);
 \draw [<->] (-1.2,0) -- (-1.7,0);
\begin{scope}[xshift=-2.8cm]
\node [vb,label=180:1]	(1)	at (-1.25,1)	{}; 
\node [vb,label=90:2]		(2)	at (-0.5,0.5)	{}; 
\node [vb,label=90:3]		(3)	at (0.5,0.5) 	{}; 
\node [vb,label=0:4]		(4)	at (1.25,1)		{}; 
\node [vb,label=0:5]		(5)	at (0,-1)		{};
\foreach \i/\j in {1/4,1/2,2/3,3/4}{
 \draw [eb] (\i) -- node[vh,label=above:(\i\j)] {}  (\j);
  }
  \foreach \i/\j in {1/5,2/5,3/5,4/5}{
 \draw [eb] (\i) -- node[vh,label=left:(\i\j)] {}  (\j);
  }
\end{scope}

\begin{scope}
\node [c]					(v)	at (0,-.12)		{};
\node [c,label=88:1]			(1)	at (-0.56,0.32) 	{}; 
\node [c,label=185:2]			(2)	at (-0.17,-0.07)	{}; 
\node [c,label=0:3]			(3)	at (0.6,0) 		{}; 
\node [c,label=90:4]			(4)	at (.14,.37)	{}; 
\node [c,label=-10:5]			(5)	at (0,-.6)		{};
\node [c,label=left:(12)]		(12)	at (-.72,.32)	{};
\node [c,label=above:(23)]		(23)	at (.44,.09)	{};
\node [c,label=0:(34)]		(34)	at (.57,.52)	{};
\node [c,label=5:(14)]		(14)	at (-.3,.61)		{};
\node [c,label=below:(15)]		(15)	at (-1,-.44)		{};
\node [c,label=below:(25)]		(25)	at (-.28,-.97)	{};
\node [c,label=0:(35)]		(35)	at (1.02,-.76)	{};
\node [c,label=-10:(45)]		(45)	at (.28,-.23)	{};
\foreach \i/\j in {1/2,1/4,2/3,3/4,1/5,2/5,3/5,4/5}{
 \path	[f] 	(\i) -- (\i\j) -- (\j) -- (v) -- cycle;
 }
 \foreach \i in {1,2,3,4,5}{
  \draw [e] (\i) -- (v);
  }
\foreach \i/\j in {1/2,1/4,2/3,3/4,1/5,2/5,3/5,4/5}{
 \draw 	[eh]	(v)		-- 	(\i\j);
 \draw	[eb] 	(\i) node[vb] {} -- (\i\j) node[vh] {};
 \draw 	[eb] 	(\j) node[vb] {} -- (\i\j) node[vh] {};
 }
\draw [e] (3) node[vb] {} -- (v) node[v] {};

\end{scope}
\begin{scope}[xshift=3.3cm]
\node [c]					(v)	at (0,-.12)		{};
\node [c,label=88:1]			(1)	at (-0.56,0.32) 	{}; 
\node [c,label=185:2]			(2)	at (-0.17,-0.07)	{}; 
\node [c,label=0:3]			(3)	at (0.6,0) 		{}; 
\node [c,label=90:4]			(4)	at (.14,.37)	{}; 
\node [c,label=-10:5]			(5)	at (0,-.6)		{};
\node [c,label=left:(12)]		(12)	at (-.72,.32)	{};
\node [c
				]		(23)	at (.44,.09)	{};
\node [c,label=0:(34)]		(34)	at (.57,.52)	{};
\node [c,label=5:(14)]		(14)	at (-.3,.61)		{};
\node [c,label=below:(15)]		(15)	at (-1,-.44)		{};
\node [c,label=below:(25)]		(25)	at (-.28,-.97)	{};
\node [c,label=0:(35)]		(35)	at (1.02,-.76)	{};
\node [c
				]		(45)	at (.28,-.23)	{};
\foreach \i/\j in {1/2,1/4,2/3,3/4,1/5,2/5,3/5,4/5}{
 \path	[f] 	(\i) -- (\i\j) -- (\j) -- (v) -- cycle;
 }
 \foreach \i in {1,2,3,4,5}{
  \draw [e] (\i) -- (v);
  }
\foreach \i/\j in {1/2,1/4,2/3,3/4,1/5,2/5,3/5,4/5}{
 \draw	[eh]	(v)		-- 	(\i\j);
 \draw	[eb] 	(\i) node[vb] {} -- (\i\j) node[vh] {};
 \draw 	[eb] 	(\j) node[vb] {} -- (\i\j) node[vh] {};
 }
\draw [e] (3) node[vb] {} -- (v) node[v] {};

\draw (-0.58,-0.14)-- (0,1.37);
\draw (-0.58,-0.14)-- (1.15,-0.32);
\draw (-0.58,-0.14)-- (-1.45,-0.74);
\draw (0.89,-1.2)-- (1.15,-0.32);
\draw (1.15,-0.32)-- (0,1.37);
\draw(0,1.37)--  (0.89,-1.2);
\draw (-1.45,-0.74)-- (0.89,-1.2);
\draw(0.89,-1.2)-- (1.15,-0.32);
\draw (-1.45,-0.74)-- (0,1.37);
\draw(0,1.37)--  (0.89,-1.2);
\draw (0.89,-1.2)-- (-1.45,-0.74);
\end{scope}

\end{tikzpicture}
  \caption{\label{fig:pyramid}The pyramid graph $\bbg$, the associated spin foam atom $\sfa = \alpha(\bbg)$ and finally the pyramid 3--cell constructed around it.}
\end{figure}

\end{remark}

\begin{remark}[{\bf group field set $\Phi$}]
  \gft\ models based upon (in)finite subsets $\Phi_{\sub}\subset \Phi$ probe only subsets $\bbgs_{\sub}\subset\bbgs$ and $\sfrs_{\sub}\subset\sfrs$ and thus, only subsets of the \lqg\ states and spin foam dynamics. One could consider examining a model based on all of $\Phi$ and all of $\Ocal$. In this manner, one would probe all of $\bbgs$ and $\sfrs$,  as one might expect in the traditional \lqg\ context. 
 
The resulting construction, however, is likely to remain at a formal level. In fact, the multi--field \gft\ realization depends upon infinitely many fields and, in order to have non--trivial dynamics for each field, infinitely many interaction terms. This likely renders any field theoretic analysis rather impracticable.

Having said that, with appropriate choices for kinetic and interaction kernels, the multi--field \gft\ based upon $\Phi$ and $\Ocal$ generates series probing all the spin foam molecules of $\sfrs$,  weighted by amplitudes coinciding with the \kkl\ extension of the \eprl\ quantum gravity model and propagating \lqg\ states on graphs in $\bbgs$.
\end{remark}

 To the extent that \gft s are currently analytically tractable, one is motivated to repackage the structures generated above and devise a class of \gft\ models that encode the quantum dynamics of arbitrary \lqg\ states, while remaining more practically useful. This means managing to encode arbitrary boundary graphs using a single or at least a (small) finite number of \gft\ fields and interactions. The key to achieving this result, which we now illustrate, lies in the use of labelled structures. 
 
\subsection{Dually weighted group field theories}\label{duwegft}

This section focusses on the labelled structures $\bgst_{n,\Simplicial}$, $\bbgst_{n,\Simplicial}$, $\sfast_{n,\Simplicial}$ and $\sfrst_{n,\Simplicialdually}\subset\sfrst_{n,\Simplicial}$.  The reason is that while the first three sets of building blocks are finite, the set of dually--weighted molecules $\sfrst_{n,\Simplicialdually}$ is rich enough to encode all of $\sfrs$. Moreover, this translates to a \gft, based on a finite number of fields and interactions, that generates sets of spin foam molecules large enough to propagate arbitrary \lqg\ states.

\subsubsection{Labelled simplicial {\sc gft}}
\label{sssec:PrimitiveGFT}

Utilizing the labelled simplicial structures $\bpst_{n}$, $\bbgst_{n,\Simplicial}$ and $\sfast_{n,\Simplicial}$ to generate spin foam molecules $\sfrst_{n,\Simplicial}$ is a simple generalization of the simplicial model presented in Section \ref{ssec:simplicialgft}:

\begin{itemize}
  \item[--] The set of group fields is indexed by the set of labelled $n$--patches:
    \begin{equation}
      \Phit = \{\phit_{\bpt}\}_{\bpst_{n}}\;,\quad\quad\textrm{where}\quad\quad \phit_{\bpt}:G^{\times |\mathcal{E}_{\bpt}|} \longrightarrow \Rbb\;,
    \end{equation}
    Note that this is a finite set of fields: $|\Phit| = |\bpst_{n}| = 2^{n}$. Also, $|\mathcal{E}_{\bpt}| = n$. 
  \item[--]  The set of \trace\ observables is indexed by the set of labelled $n$--regular, loopless graphs $\bbgst_{n,\loopless}$:
    \begin{equation}
      \Obt = \{\obt_{\bbgt}\}_{\bbgst}\;,\quad\quad\textrm{where}\quad\quad \obt_{\bbgt}[\Phit] = \int [\extd g]\; \bbBt_{\bbgt}\big(\{g_{\vbar}\}_{\Vbar}\big)\prod_{\vbar\in\Vbar} \phi_{\bpt}(g_{\vbar})\;,
    \end{equation}
    where $\bbBt_{\bbgt}$ implicitly depends on the edge labels drawn from $\{real,virtual\}$. 

  \item[--] The set of vertex interactions is indexed by labelled simplicial $n$--graphs $\bbgst_{n,\Simplicial}$.  Since these are all based on the complete graph over $n+1$ vertices, one can utilize the vertex labelling seen in equation \eqref{eq:SimplicialVertex}: 
    \begin{equation}
      \lambda_{\bbgt}\int [\extd g]\; \Vbbt_{\bbgt}\big(g\big)\prod_{j = 1}^{n+1} \phit_{\bpt}(g_j)
    \end{equation}
    This is a finite set of interactions: $|\bbgst_{n,\Simplicial}| = 2^{{n+1\choose 2}}$.
    
    Of course, the set of interaction terms can be extended to those indexed by $\bbgst_{n,\loopless} = \sigma(\bpst_{n})$, and one should probably expect them to be generated during the renormalization process.  However, the point is that the small set $\bbgst_{n,\Simplicial}$ is rich enough to generate spin foam molecules that could provide non--trivial correlations for all of $\bbgst_{n,\loopless}$, and so is a well--chosen minimal model to take at the outset. 

    Again, using  the bijection in Proposition \ref{prop:bijectionAtom}, the interaction terms can be interpreted as generating spin foam atoms $\sfat = \alpha(\bbgt)$.

  \item[--]
    The kinetic term 
    is responsible for the bonding of patches just as in the unlabeled case \eqref{eq:kinetic}:  
    \begin{equation}
      \frac12\int[\extd g]\; \phit_{\bpt}(g_1)\;\Kbbt_{\bpt}(g_1, g_2)\;\phit_{\bpt}(g_2)\;,\quad\quad\textrm{where}\quad\quad \Kbbt_{\bpt}(g_1,g_2) = \Kbbt(g_{\vbar_1}, g_{\vbar_2} )\;.
    \end{equation}

  \item[--]
    Then, the class of labelled simplicial \gft s is defined via:
    \begin{equation}
      Z_{\pgft} = \int \Dcal\Phit \;e^{-S[\Phit]}
    \end{equation}
    with:
    \begin{equation}
      S[\Phit] = \frac12\sum_{\bpt\in\bpst_{n}} \int[\extd g]\;\phit_{\bpt}(g_1)\;\Kbbt_{\bpt}(g_1, g_2)\;\phit_{\bpt}(g_2) + 
      \sum_{\bbgt\in\bbgst_{n,\Simplicial}}\lambda_{\bbgt}\int [\extd g]\; \Vbbt_{\bbgt}\big(g\big)\prod_{j = 1}^{n+1} \phit_{\bpt}(g_j)
    \end{equation}

  \item[--]
    \trace\ observables can be estimated perturbatively, generating series of the type:
    \begin{eqnarray}
      \langle \obt_{\bbgt_1}\dots \obt_{\bbgt_l}\rangle_{\pgft} &=& \frac{1}{Z}\int \Dcal\Phit\; \obt_{\bbgt_1}[\Phit]\dots\obt_{\bbgt_l}[\Phit]\;e^{-S[\Phit]} \nonumber \\
      &=& \sum_{\substack{\sfrt\in\sfrst_{n,\Simplicial}\\[0.1cm] \widetilde{\delta}(\sfrt) = \sqcup_{i = 1}^{l}\bbgt_{i}}} \frac{1}{C(\sfrt)} A(\sfrt;\{\lambda_{\bbgt}\}_{{\bbgst}_{n,\Simplicial}})\;,
    \end{eqnarray}
\end{itemize}

\begin{remark}[{\bf reducibility}]
  As pointed out in Remark \ref{rem:moleculered}, not all molecules in $\sfrst_{n,\Simplicial}$ reduce to a molecule in $\sfrs$. It is rather the dually--weighted subset $\sfrst_{n,\Simplicialdually}\subset\sfrst_{n,\Simplicial}$ that possesses this property.  As a result, one needs a mechanism at the \gft\ level that isolates this subset.  This mechanism is known as dual--weighting.  
  \end{remark}

\subsubsection{Dually--weighted {\sc gft}}
\label{sssec:DuallyWeightedGFT}

It emerges that employing a simple technique at the field theory level allows one to extract directly the subclass of structures $\sfrst_{n,\Simplicialdually}\subset\sfrst_{n,\Simplicial}$. This technique, dubbed \emph{dual--weighting} in the the matrix model literature, assigns parameterized weights to the vertices $\Vhat$ of the spin foam atoms $\sfat\in\sfast_{n,\Simplicial}$ and, through the bonding mechanism, of the spin foam molecules $\sfrt\in\sfrst_{n,\Simplicial}$.\footnote{The \emph{dual--weighting} moniker stems from the fact that in 2d these vertices are in one--to--one correspondence with the vertices of the dual topological structure. In that context, these parameterized weights can be interpreted as coupling parameters for dual vertices.} 
These weights can be tuned so that only virtual interior/boundary vertices in $\Vhat$ with precisely two virtual faces/one virtual face incident survive. This is precisely the condition pinpointing the configurations in $\sfrst_{n,\Simplicialdually}$.

The dual--weighting mechanism begins by enlarging the elementary data set from $G$ to $G\times\Mcal$, where $\Mcal = \{0,1,\dots, M\}$.
The integer $M$ can be regarded as a free parameter of the theory.
Since these data sets are associated to edges of both patches and boundary graphs, they permit a new encoding of the edge labels $\{real, virtual\}$. The \emph{real} label is encoded as the zero element $0\in\Mcal$, while the \emph{virtual} label is encoded by the non--zero elements $m\in\Mcal-\{0\}$.  
\begin{itemize}
  \item[--] This in turn allows one to repackage the $2^{n}$ fields $\phit_{\bpt}$ ($\bpt\in\bpst_{n,\Simplicial}$) into a single field 
    \begin{equation}
      \phi:(G\times\Mcal)^{n}\longrightarrow\Rbb
    \end{equation}
    based on the unique \emph{un}labelled $n$--patch $\bp_n\in\bps_{n,\Simplicial}$.
    This stems from the fact that these patches have the same combinatorics, differing only in the choice of labels $\{real, virtual\}$ assigned to their edges.
  
  \item[--] In principle, the \trace\ observables are indexed once again by labelled, $n$--regular, loopless graphs $\bbgt\in\bbgst_{n,\loopless}$. Encoding the labelling as above, one can re--index observables by \emph{un}labelled, $n$--regular, loopless graphs $\bbgs_{n,\loopless}$:
   \begin{equation}
      \Ocal = \{O_{\bbg}[\Phi]\}_{\bbgs_{n,\loopless}}\;,\quad\quad\textrm{where}\quad\quad O_{\bbg}[\Phi] = \int [\extd g]\sum_{[m]} \bbB_{\bbg}\big(\{g_{\vbar}; m_{\vbar}\}_{\Vbar}\big)\prod_{\vbar\in\Vbar} \phi_{\bp}(g_{\vbar};m_{\vbar})\;,
    \end{equation}
    where the \emph{combinatorial non--locality} extends to the $\Mcal$ variables. In effect, the observable $O_{\bbg}$ incorporates all $2^{|\mathcal{\bbg}|}$ labelled observables with support on that graphical structure.

    However, combinatorial non--locality, in conjunction with this novel label--encoding, places a restriction on $\bbB_{\bbg}(\{g_{\vbar};m_{\vbar}\}_{\Vbar})$. To detail this, 
one uses the same indexing of vertices and edges as in \eqref{eq:SimplicialVertex} with an extra label $(a)$ to number multi--edges. A simple illustration for a bisected edge between two vertices $i,j\in\Vbar$ looks like:
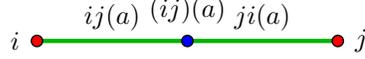
\begin{center}
\tikzsetnextfilename{edge}
\tikz \draw [eb]  (-2,0) node [vb,label=left:$i$] {}  --  node[auto] {$ij(a)$} 
			(0,0) node[vh,label=above:$(ij)(a)$] {} -- node[auto] {$ji(a)$}  
			(2,0) node[vb,label=right:$j$] {};
\end{center}
Then the graph $\bbg$ dictates that the boundary kernel has the form:
\begin{equation}
  \bbB_{\bbg}(\{g_{\vbar};m_{\vbar}\}_{\Vbar}) = \bbB(\{g_{ij(a)}g_{ji(a)}^{-1}; m_{\vbar}\})\;
\end{equation}  
For labelled boundary graphs, both edges $ij(a)$ and $ji(a)$ are marked by the same label $\{real, virtual\}$ (see Remark \ref{rem:generalization}).  This translates to the restriction that $\bbB_{\bbg}(\{g_{\vbar};m_{\vbar}\}_{\Vbar}) = 0 $ when $m_{ij(a)} = 0$, $m_{ji(a)}\in\Mcal-\{0\}$ or vice versa. Alternatively, $\bbB_{\bbg}(\{g_{\vbar};m_{\vbar}\}_{\Vbar}) \neq 0$ only when both $m_{ij(a)} = 0 = m_{ji(a)}$ or both $m_{ij(a)}, m_{ji(a)} \in \Mcal - \{0\}$.    

 \item[--]   The $2^{{n+1\choose 2}}$ interaction terms are indexed by labelled simplicial graphs $\bbgt\in\bbgst_{n,\Simplicial}$.  As with the boundary kernels, one may re--encode the labelling in terms of the new data set.  As a result, one can capture all the interaction terms using the unique \emph{un}labelled, $n$--regular, loopless graph $\bbg\in\bbgs_{n,\Simplicial}$.:  
    \begin{equation}
      \lambda\int [\extd g]\sum_{[m]} \Vbb_{\bbg}(g; m)\prod_{j = 1}^{n+1} \phi(g_j; m_j)\;,
      \end{equation}
      where (analogously to \eqref{eq:SimplicialVertex}): 
     \begin{equation}
       \Vbb_{\bbg}(g;m) = \Vbb\big(\{g_{ij}g_{ji}^{-1};m_{\vbar}\}\big)\;, \quad\quad \textrm{with}\quad\quad i<j\;,
     \end{equation} 
     and the markers $i,j\in\{1,\dots,n+1\}$ index the $n+1$ vertices $\Vbar\subset \bbg$ and thus the patches of $\bbg$. Meanwhile, the pair $ij$ (with $j\neq i$) indexes the edge joining the vertex $i$ to the bisecting vertex $(ij)$.
Combinatorial non--locality imposes an analogous constraint on this interaction kernel.

    Just as for labelled simplicial \gft s,  the set of interaction terms could be extended to those indexed by $\bbgs_{n,\loopless} = \sigma(\bps_{n,\Simplicial})$, while still invoking the dual weighting mechanism. In terms of labelled structures, this means that one could isolate $\sfrst_{n,\looplessdually}\subset\sfrst_{n,\loopless}$. 

   \item[--] The kinetic term takes the form:
     \begin{equation}
       \frac12\int[\extd g]\sum_{[m]}\phi(g_1;m_1)\;\Kbb(g_1,g_2;m_1,m_2)\;\phi(g_2;m_2)\;,
     \end{equation}
     where:
   \begin{equation}
     \Kbb(g_1,g_2;m_1,m_2) = \Kbb\big(g_{\vbar_1},g_{\vbar_2};m_{\vbar_1},m_{\vbar_2}\big)\;,
     \end{equation}
    Since the kinetic term is responsible for the bonding of the patches and bonding respects labelling, then $\Kbb \neq 0$ only when both $m_{\vbar_1\vhat} = 0 = m_{\vbar_2\vhat}$ or both $m_{\vbar_1\vhat},m_{\vbar_2\vhat}\in\Mcal - \{0\}$.  
 \end{itemize}
 Before stating the class of dually--weighted \gft s, we specify the precise form of the $\Mcal$--sector of the various kernels. We shall leave the $G$--sector unspecified for the moment, dealing with specific cases in Section \ref{sec:sfm}.
 \begin{defin}[{\bf dual--weighting matrix}]
   The \emph{dual--weighting matrix sequence}, $\{\Acal_M\}_{M>0}$, is a sequence of invertible matrices (where $M$ denotes the size of $\Acal_M$) which satisfy the condition:%
\footnote{
The trace invariant information contained within an $M\times M$ matrix $\Acal$ can be characterized in a number of ways,   perhaps most familiarly through its $M$ eigenvalues, which arise as the roots of the \emph{characteristic equation} $\chi_{\Acal}(t) = 0$, where:
\begin{equation}
\chi_{\Acal}(t) = \det(t\Ical - \Acal)\;.
\end{equation}
However, a less succinct way to package this information is in the traces of matrix powers: $\tr(\Acal^k)$, where $0\leq k\leq M$. To see this, notice that one can rewrite the characteristic polynomial as:
\begin{equation}
  \chi_{\Acal}(t) = \sum_{k = 0}^M (-1)^k\,t^{M-k}\,\tr(\wedge^k \Acal)\;, \quad\quad\textrm{where}\quad\quad \wedge^k\Acal = \frac{1}{k!}
  \begin{vmatrix} 
    \tr(\Acal) & k-1 & 0 & \cdots & 0 \\[0.2cm]
    \tr(\Acal^2) & \tr(\Acal) & k-2 &\cdots & 0 \\
    \vdots & \vdots & \vdots &  & \vdots \\[0.2cm]
    \tr(\Acal^{k-1}) & \tr(\Acal^{k-2}) & \tr(\Acal^{k-3}) & \cdots & 1\\[0.2cm]
    \tr(\Acal^k) & \tr(\Acal^{k-1}) & \tr(\Acal^{k-2}) & \cdots & \tr(\Acal)
  \end{vmatrix}\;.
\end{equation}
The eigenvalues are determined in terms of the traces and vice versa. In the large--$M$ limit, it is clear therefore that one may impose an infinite number of conditions on the matrix traces.

For concreteness, let us consider as a specific sequence $\{\Acal_M\}$ the diagonal matrices 
$\left(\Acal_M\right)_{mm'} = (-1)^m M^{-1/2} \delta_{mm'}$. 
These fulfill the conditions Eq.\,\eqref{eq:dwproperty} since for odd $k$ 
\begin{equation}
\tr \left( (\Acal_M)^k \right) = - M^{-k/2} \us{M\ra\infty}\longrightarrow 0
\end{equation}
and for even $k$ 
\begin{equation}
\tr \left( (\Acal_M)^k \right) = M^{1-k/2} 
\end{equation}
which equals one for $k=2$ and tends to zero in the large-$M$ limit for $k>2$.

 For more uses of the dual--weighting mechanism see \cite{Benedetti:2012ed} and references therein.
}
\begin{equation}
\label{eq:dwproperty}
  \lim_{M\rightarrow\infty}
  \tr\left( (\Acal_M)^k \right) = \delta_{k,2}.
\end{equation}
\end{defin}
In the remainder we suppress the index $M$ in the dual-weighting matrices.

\begin{remark}[{\bf dual--weighting mechanism}]
  \label{rem:DWmechanism}
 The implementation 
  of the dual--weighting mechanism places certain restrictions on the kinetic, interaction and boundary kernels:
 \begin{itemize}
   \item[--] The kinetic kernel takes the form:
 \begin{equation}
   \Kbb(g_1,g_2;m_1,m_2) = \Kbbb(g_1,g_2;m_1,m_2)\;\Dbb^{-1}(m_1,m_2)\;,
 \end{equation}
 where $\Kbbb$ is constant across $m_{1j}, m_{2j}\in\Mcal-\{0\}$, for each $j\in\{1,\dots,n\}$.  In other words, $\Kbbb$ only depends on whether the edges are \emph{real} or \emph{virtual}.
 Meanwhile, $\Dbb$ factorizes across the edges:
 \begin{equation}
   \Dbb(m_1,m_2) = \prod_{j = 1}^n \dbb_{m_{1j},m_{2j}}\;,
   \quad\quad\textrm{with}\quad\quad 
   \dbb = \left(
    \begin{array}{c|c}
      1&0\\ \hline
      0&\Acal
    \end{array}
    \right)\;.
 \end{equation}
 The condition on $\Kbbb$ means that the value it attains only depends on whether the edges are \emph{real} or \emph{virtual}. 
The zero entries in the $\dbb$--matrix are the manifestation of non--mixing of \emph{real} and \emph{virtual} edges. The dual--weighting matrix $\Acal$ is the truly significant player, as it will be responsible for restricting the spin foam molecules in the large--$M$ limit. 

 \item[--]
The interaction kernel takes the form:
\begin{equation}
  \Vbb_{\bbg}(g;m) = \Vbbb_{\bbg}(g;m)\;\Ibb(m)\;,
  \quad\quad\textrm{where}\quad\quad 
  \Ibb(m) = \prod_{(ij)} \ibb_{m_{ij},m_{ji}} 
  \quad\quad\textrm{and}\quad\quad
  \ibb =  \left(
    \begin{array}{c|c}
      1&0\\ \hline
      0&\Ical
    \end{array}
    \right)\;.
\end{equation}
$\Ical$ is the $M\times M$ identity matrix and the function $\Vbbb$ only depends on whether the edges are \emph{real} or \emph{virtual}.   

\item[--]
Meanwhile, the boundary kernels take the similar form:
\begin{equation}
  \bbB_{\bbg}(\{g_{\vbar};m_{\vbar}\}_{\Vbar}) = \bbBb_{\bbg}(\{g_{\vbar};m_{\vbar}\}_{\Vbar})\;\Ibb(\{m_{\vbar}\}_{\Vbar})\;, 
  \quad\quad\textrm{where}\quad\quad 
  \Ibb(\{m_{\vbar}\}_{\Vbar}) = \prod_{(ij)(a)} \ibb_{m_{ij(a)},m_{ji(a)}} 
\end{equation}
and $\bbBb$ only depends on whether the edges are \emph{real} or \emph{virtual}.
\end{itemize}
Of course, in true \gft\ style, one could shift the dual--weighting matrix to the interaction kernels, that is, swapping $\Ical$ for $\Acal$ in the interaction kernel, while simultaneously swapping $\Acal$ for $\Ical$ in the kinetic kernel. 
Accordingly, in this realization, the boundary kernels should also contain the dual--weighting matrix to ensure the correct propagation of virtual edges. 
\end{remark}
Now, back to the definition of the dually--weighted \gft s:
 \begin{itemize}
  \item[--]
    Then, the class of dually--weighted \gft s is defined via:
    \begin{equation}
      Z_{\dwgft} = \int \Dcal\Phi \;e^{-S[\Phi]}
    \end{equation}
    with:
      \begin{equation}
	S[\Phi] = \frac12\int[\extd g]\sum_{[m]}\phi(g_1;m_1)\;\Kbb(g_1, g_2;m_1,m_2)\;\phi(g_2;m_2)
	+ \lambda\int [\extd g]\sum_{[m]} \Vbb_{\bbg}(g;m)\prod_{j =1}^{n+1} \phi(g_j;m_j)
    \end{equation}

  \item[--]
   A \trace\ observable can be estimated perturbatively, generating series of the type:
    \begin{eqnarray}
    \label{eq:dwperturbative}
      \langle O_{\bbg_1}\dots O_{\bbg_l}\rangle_{\dwgft} &=& \frac{1}{Z_{\dwgft}}\int \Dcal\Phi\; O_{\bbg_1}[\Phi]\dots O_{\bbg_l}[\Phi]\;e^{-S[\Phi]} \nonumber \\
      &=& \sum_{\substack{\sfr\in\sfrs_{n,\Simplicial}\\[0.1cm] {\delta}(\sfr) = \sqcup_{i = 1}^{l}\bbg_{i}}} \frac{1}{C(\sfr)} A(\sfr;\lambda)\;,
    \end{eqnarray}
    As has been stated repeatedly, these series are in principle catalogued by labelled simplicial $n$--molecules $\sfmt\in\sfmst_{n,\Simplicial}$.  However, the label--dependence\footnote{For example, one could tweak the coupling constants to depend more sensitively on the reality/virtuality of the various edges of the bisected graph.  In that case, one would really have to catalogue the sum explicitly in terms of the elements of $\sfmst_{n,\Simplicial}$.} permits the collation and repackaging of the amplitudes attached to the various labellings of each \emph{un}labelled simplicial $n$--molecule $\sfm\in\sfms_{n,\Simplicial}$. In the large--$M$, as one can see in a moment, the label information plays an important role and must be made explicit once more. 
\end{itemize}

\fbox{
  \begin{minipage}[c][][c]{0.97\textwidth}
\begin{proposition}[{\bf large--$M$ limit}]
\label{prop:mlimit}
In the large--$M$ limit of the \dwgft, the observable expectation values  possess perturbative expansions in terms of simplicial $n$--molecules within the dually--weighted subclass $\sfrst_{n,\Simplicialdually}$ (Remark \ref{rem:dually})
  \begin{eqnarray}
    \lim_{M\rightarrow\infty} \langle O_{\bbgt_1}\dots O_{\bbgt_l}\rangle_{\dwgft} &=&  \sum_{\substack{\sfrt\in\sfrst_{n,\Simplicialdually}\\[0.1cm] {\delta}(\sfrt) = \sqcup_{i = 1}^{l}\bbgt_{i}}} \frac{1}{C(\sfrt)} A(\sfrt;\lambda)\;.
    \end{eqnarray}
\end{proposition}
\end{minipage}
}
\begin{proof}
  According to the definition of the amplitudes $A(\sfr;\lambda)$, with $\sfr\in\sfrs_{n,\Simplicial}$, in the perturbative sum \eqref{eq:dwperturbative}, the dual--weighting part of the amplitude factorizes across the vertices $\vhat\in\Vhat$ and for each vertex, denoted by $A_{\vhat}$, it takes one of two values:
  \begin{equation} 
    \begin{array}{cl}
      1 & \textrm{for $\vhat$ real.}\\[0.3cm]
   \displaystyle \tr\Big(\prod_{(\vbar\vhat)}\Acal\Big) & \textrm{for $\vhat$ virtual.}
  \end{array}
\end{equation}
More precisely, the amplitude $A(\sfr;\lambda)$ contains contributions from all the labelled counterparts $\sfrt$ of $\sfr$. In such a counterpart, if $\vhat$ is real, then the dual--weighting amplitude is unity, while if $\vhat$ is virtual, the contribution is the trace of a power of $\Acal$, with one $\Acal$--factor for each edge $(\vbar\vhat)$ incident at $\vhat$. 
Due to the dual--weighting property \eqref{eq:dwproperty},  the second contribution vanishes in the limit $M\longrightarrow\infty$ unless there are precisely two/one edge(s) incident at this internal/boundary vertex $\vhat$.  This is exactly the defining property of $\sfrst_{n,\Simplicialdually}$ (Remark \ref{rem:dually}).
\end{proof}

  \begin{remark}[{\bf molecular interpretation}]
	  Proposition \ref{prop:mlimit} allows one to recast the perturbative series generated by the \dwgft\ in the large--$M$ limit as a series catalogued by more generic molecules.  This follows directly from the reduction accomplished by the projection map $\Pi_{n,\Simplicialdually}:\sfrst_{n,\Simplicialdually}\longrightarrow\sfrs$ (Remark \ref{rem:refmolreduction}). From Proposition \ref{prop:limitation}, $\Pi_{n,\Simplicialdually}$ does not cover the whole of $\sfrs$, but only a subset $\sfrs' = \Pi_{n,\Simplicialdually}(\sfrst_{n,\Simplicialdually})$.
The perturbative series can indeed be rewritten as a sum over spin foam molecules:
  \begin{eqnarray}
    \lim_{M\rightarrow\infty} \langle O_{\bbgt_1}\dots O_{\bbgt_l}\rangle_{\dwgft} &=&  \sum_{\substack{\sfr\in\sfrs'\subset\sfrs\\[0.1cm] {\delta}(\sfr) = \sqcup_{i = 1}^{l}\pi_{n,\Simplicialdually}(\bbgt_{i})}} \frac{1}{C^\textbf{eff}(\sfr)} A^\textbf{eff}(\sfr;\lambda)\;.
    \end{eqnarray}   
    Having said that, note that \emph{every} collection of boundary graphs drawn from $\bbgs$ can be evolved within this $\dwgft$. Should one wish to include a larger set of effective molecules from $\sfrs$, this could be easily obtained by incorporation of more interaction terms with support on $\bbgs_{n,\loopless} = \sigma(\bps_{n,\Simplicial})$. 

    In the quantum gravity context, this means that \dwgft s are effectively \gft s describing physical inner products and correlations of various quantum gravity states with support on arbitrary graphs, which are estimated using perturbative series catalogued by spin foam molecules of the most general combinatorics.
\end{remark}

\ 

\begin{remark}[{\bf 3d example revisited}]

  Let us look again at the scenario of Remark \ref{rem:3dexample}. From the dually--weighted viewpoint:
  \begin{itemize}
    \item[--] The interpretation is 3--dimensional, thus one should take $n=3$, with a field $\phi:(G\times\Mcal)^3\longrightarrow\Rbb$.
    \item[--]
      The pyramid observable detailed in \eqref{eq:pyramid} has, among its realizations in the dually--weighted model, the following one based on the graph $\bbg$,  illustrated in Figure \ref{fig:dw3dgraph} (cf. the corresponding atom, Figure \ref{fig:atommoves}) and composed out of six fields:
\begin{equation}
\label{eq:pyramidAgain}
O_{\bbg}[\phi] = \int [\extd g]\sum_{[m]} \bbB_{\bbg}(g;m)\,
\phi(g_1;m_1)\,\phi(g_2;m_2)\,\phi(g_3;m_3)\,\phi(g_4;m_4)\,\phi(g_5;m_5)\;\phi(g_6;m_6)\,.
\end{equation}
where:
\begin{multline}
  \bbB_{\bbg}(g;m) = \bbB(g_{12}g_{21}^{-1},g_{23}g_{32}^{-1},g_{34}g_{43}^{-1},g_{14}g_{41}^{-1},g_{15}g_{51}^{-1},g_{25}g_{52}^{-1},g_{36}g_{63}^{-1},g_{46}g_{64}^{-1},g_{56}g_{65}^{-1}; m_{12},m_{21},\\
   \quad m_{23},m_{32},m_{34},m_{43},m_{14},m_{41},m_{15},m_{51},m_{25},m_{52},m_{36},m_{63},m_{46},m_{64},m_{56},m_{65})\;.
\end{multline}
This packages together a number of observables, depending on the labels $\Mcal$. The observable of interest is precisely the configuration where the labels $m_{56}$ and $m_{65}$ are non--zero (indicating a virtual edge), while the rest are zero (indicating real edges).  Upon reduction of this virtual edge, one arrives at a graph in $\sfrs$ with one 4--valent patch and four 3--valent patches just as in \eqref{eq:pyramid}. One has represented the square base of the observable \eqref{eq:pyramid} in terms of two triangles in the dually--weighted model.

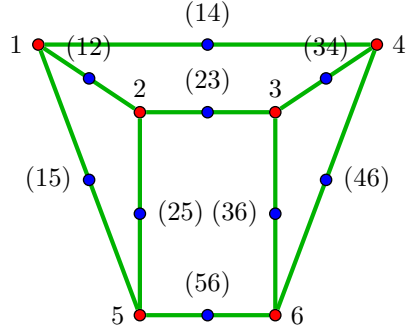
\begin{figure}[htb]
  \centering
  \tikzsetnextfilename{pyramiddw}
  \begin{tikzpicture}[scale=1.8]

\node [vb,label=180:1]	(1)	at (-1.25,1)	{}; 
\node [vb,label=90:2]		(2)	at (-0.5,0.5)	{}; 
\node [vb,label=90:3]		(3)	at (0.5,0.5) 	{}; 
\node [vb,label=0:4]		(4)	at (1.25,1)		{}; 
\node [vb,label=180:5]	(5)	at (-0.5,-1)		{};
\node [vb,label=0:6]		(6)	at (0.5,-1)		{};
\foreach \i/\j in {1/4,1/2,2/3,3/4}{
 \draw [eb] (\i) -- node[vh,label=above:(\i\j)] {}  (\j);
  }
\foreach \i/\j in {1/5,3/6}{
  \draw [eb] (\i) -- node[vh,label=left:(\i\j)] {}  (\j);
  }
\foreach \i/\j in {2/5,4/6}{
  \draw [eb] (\i) -- node[vh,label=right:(\i\j)] {}  (\j);
  }
\draw [eb] (5) -- node[vh,label=above:(56)] {}  (6);
%
%

\end{tikzpicture}
  \caption{\label{fig:dw3dgraph}The graph $\bbg$ composed of six unlabelled 3--patches.}
\end{figure}

\end{itemize}

\end{remark}

\begin{remark}[{\bf extensions}]
 The dual--weighting mechanism can be applied to other classes of models:
 \begin{itemize}
   \item[--]
The construction above works for arbitrary valence $\copies$, and can clearly be extended to multiple \gft\ fields, if so wished. In models for $D$-dimensional gravity, we would like the valence of the graphs associated to quantum states to be $D$, as in simplicial models, and for the same reasons. One should note, however, that in even dimensions $D$, only effective nodes of even valency are then obtained after the dynamical contraction of virtual links. We have already noticed this combinatorial restriction in the previous section (Proposition \ref{prop:surjectors}).
If one wants to generate graphs of truly arbitrary valence, using the same mechanism, one can easily do so by incorporating a single odd--valent field species, endowed with a $D$--dimensional interpretation. Again, this doubling of fields does not change the general features of the construction. 

\item[--]

Also, notice that the {\it coloured extension} of the \gft\ formalism can be directly applied to the dually--weighted model, provided the valence $\copies$ is chosen to be the space--time dimension $D$. This can be done, as we have seen, either by choosing also a simplicial \gft\ interaction, which brings one back to the standard simplicial setting, or by choosing as \gft\ interactions only the tensor invariant ones. The result of doing so is, in both cases, a set of \gft\ Feynman diagrams dual to combinatorial complexes whose full homological structure can be reconstructed from the colour information. 
\end{itemize}
\end{remark}

\begin{remark}[{\bf model building}]
      The choice of group field determines that the kinematical state space of the theory is populated by quantum states with support on arbitrary $n$--valent labelled graphs. The dual--weighting mechanism is the part of the dynamics that ensures that those states are evolved by spin foam molecules that reduce properly to arbitrary spin foam molecules. 
The precise choice of \gft\ action, incorporating the dual--weighting mechanism, is then a matter of model building. In particular, depending on the choice of interaction kernels, some classes of graphs, present in the kinematical Hilbert space, can be suppressed dynamically. 
  \begin{itemize}
    \item[--]
Therefore, one possible criterion for model building stems from the wish to suppress or enhance specific combinatorial structures. Conversely, one may want to start from the simplest set of \gft\ interactions that ensures that {\it all} kinematical states participate to the quantum dynamics. We have proven that simplicial spin foam atoms, in the context of the dually--weighted theories, satisfy this criterion.  

\item[--]
Another criterion that might determine the choice of \gft\ interaction combinatorics emerges from the correspondence between the interaction kernels and the matrix elements of a canonical \lqg\ projector operator in the Fock representation, emphasized in \cite{Oriti:2013vv}. Prescribing the latter implies a choice for the former. In general though, one should expect there to be infinitely many non--trivial matrix elements, meaning infinitely many \gft\ interaction kernels, unless these are restricted by very strong symmetry requirements. As a result, the real quest centres on pinpointing the subsets of interactions that are {\it physically} relevant at different scales, in particular,  to define the theory in some deep UV or IR regime. In other words, the problem becomes that of {\it \gft\ renormalization} \cite{\cgftrenorm}. In fact, one should expect that the renormalization group flow will select a finite set of \gft\ interactions to define a renormalizable \gft\ model. Moreover, this  dictates which new terms are relevant for the quantum dynamics at different scales. In turn, this prescribes a renormalizable \lqg\ dynamics.
\end{itemize}
\end{remark}

To sum up this section, we have defined a class of \gft s admitting boundary states with higher--valent, node structure.  Rather than increasing the number of fields and interactions, we have extended the data set of the usual simplicial \gft. Utilizing these extra arguments to invoke a dual--weighting matrix, in a limit of the theory, we acquire the effective dynamical content of arbitrary spin foams and boundary states.


\newpage


\section{Spin foam models}
\label{sec:sfm}
The previous section explored the space of \gft s, concentrating on the development of a class whose perturbative expansions were effectively catalogued by general spin foam molecules.  From our point of view, the next important step is to demonstrate that such dually--weighted models are compatible with quantum gravitational dynamics, in particular, spin foam quantum gravity models.

There is a class of gravitational spin foam models in 4--dimensions ($D=4$) in the framework of simplicial molecules $\sfrs_{D,\primitive}$, of which the Engle--Pereira--Rovelli--Livine (\eprl) \cite{\ceprl}, the Freidel--Krasnov \cite{\cfk} and the Baratin--Oriti \cite{\cbo} models are members.  Moreover, all these models permit an extension to arbitrary spin foam molecules $\sfrs$  \cite{\ckkl}, which we call \kkl--extension in the case of the \eprl\ model.

In this section, we shall first review gravitational spin foam models, in their generalized context 
, giving explicit details of the \eprl\ model. We shall show that a multi--field \gft\ straightforwardly assigns these amplitudes in the perturbative expansion.  We shall then define a dually--weighted \gft\ model that assigns the same as effective amplitudes in the large--$M$ limit.

\subsection{4--dimensional spin foam quantum gravity}

\begin{defin}[{\bf spin foam model}]
A \emph{spin foam model} is a quantum theory defined by a partition function of the following form:
\begin{equation}\label{eq:fullsf}
  Z_{\SF} = \sum_{\substack{\sfm\in\sfms\\[0.1cm] \delta(\sfm)=\emptyset}} W(\sfm)\, A(\sfm)\;,
\end{equation}
where:
\begin{description}
  \item[--] $\sfms$ is the set of spin foam molecules  from Definition \ref{def:molecule}, or one of its subsets;  
  \item[--] $A(\sfm)$ is the {\it spin foam amplitude} associated by the model to $\sfm$; and 
  \item[--] $W(\sfm)$ is the {\it spin foam measure factor} that weights $\sfm$ in the sum over such molecules. 
\end{description}
\end{defin}
\begin{remark}
The distinction between $W(\sfm)$ and $A(\sfm)$ might appear quite arbitrary. However, they are distinguished to highlight the fact that within the spin foam formalism, while derivations for certain amplitudes $A(\sfm)$ can be given,  one must prescribe $W(\sfm)$ by hand.  As one might expect, group field theory provides a complete prescription for both $A(\sfm)$ and $W(\sfm)$. 
\end{remark}

\begin{remark}
	In the operator spin foam formalism \cite{\cOperator}, the spin foam amplitude is specified by the sets of variables and operators that it associates to the components (vertices, edges, faces) of $\sfm$.  Here, those components may be identified through combinations of the various vertices $v\in\Vcal$, $\vbar\in\Vbar$, $\vhat\in\Vhat$, where $\Vcal_{\sfm} = \Vcal\cup\Vbar\cup\Vhat$ 
	(Remark \ref{rem:bondingexample}). The variables are often drawn from some group--related structures, namely,  group/algebra elements or group representations.  Sets of variables and operators are denoted by $\sfls$ and $\sfos$, respectively, while individual variables and operators are denoted by $\sfl$ and $\sfo$.
\end{remark}

\begin{remark}[{\bf gravitational spin foam model}]
  A \emph{gravitational spin foam model} is a spin foam model related to the Holst-Plebanski action, that is, it includes some quantum version of the simplicity constraint.  An interesting feature at this stage is their \emph{locality} property. They are local in the sense that the amplitude associated to a spin foam molecule $\sfm$ factorizes into operators associated to each of the vertices in $\Vcal_\sfm$, which depend only on the labels attached to \lq\lq nearby\rq\rq\ components, more precisely, components that contain a given vertex of $\Vcal_\sfm$.  Explicitly: 
\begin{equation}
\label{eq:skelsf} 
A(\sfm) = \sum_{\sfls}\prod_{\vhat\in\Vhat} \sfo_{\vhat}(\sfls_{\vhat}) \prod_{\vbar\in\Vbar} \sfo_{\vbar}(\sfls_{\vbar}) \prod_{v\in\Vcal}\sfo_v(\sfls_v) \;.
\end{equation}
\end{remark}
As stated above, we shall focus on a particular 4--dimensional gravitational spin foam model: the Engle--Pereira--Rovelli--Livine (\eprl) model \cite{\ceprl} and its \kkl--extension. 
Once more, we emphasize that the same construction, including the combinatorial generalization, applies to the other spin foam models as well.

\begin{remark}[{\bf variables}]
In the group representation, 
  the relevant variables 
are the group elements $g_{v\vbar\vhat} \in G$, $G = \SO(4)$, where $(v\vbar\vhat)\in\Fcal_\sfm$ is a face of the molecule $\sfm$.
 It is convenient to identify certain subsets of variables:
\begin{equation}
  \label{eq:defs}
\begin{array}{rcrclcrcl}
  \sfls_v &\equiv& g_v &=& \displaystyle \bigcup_{\vbar,\vhat}\{g_{v\vbar\vhat}: (v\vbar\vhat)\in\Fcal_\sfm \}\;,\\[0.3cm] 
  \sfls_{\vbar} &\equiv& g_{\vbar} &=& \displaystyle\bigcup_{v,\vhat}\{g_{v\vbar\vhat}: (v\vbar\vhat)\in\Fcal_\sfm \}\;,\\[0.3cm]
  \sfls_{v\vbar} &\equiv& g_{v\vbar} &=& g_{v}\cap g_{\vbar}\;, 
\end{array}
\end{equation}
\end{remark}
The operators are functions of these group elements. For the \eprl\ model in the group element realization, only the edge and vertex operators are non--trivial, that is, $\sfos = \sfos_v\cup\sfos_{\vbar}$.

Consider a vertex $\vbar\in\Vbar$ that arises in the molecule after the bonding of two patches $\bp_1\cong\bp_2\equiv\bp$ and denote the two edges in $\Ecal$ incident at $\vbar$ by $(v_1\vbar)$ and $(v_2\vbar)$.  The operator associated to $\vbar$, which is usually called the edge operator in the spin foam literature, is defined as follows: 
\begin{defin}[{\bf edge operator}]
  The \eprl\ \emph{edge operator} associated to $\vbar$ is:
\begin{equation}
  \label{eq:edgeop} 
  \sfo_{\vbar}(g_{\vbar}) \equiv \Pbb_{\bp}(g_{v_1\vbar}, g_{v_2\vbar})
  = \int_{G^{\times 2}} \extd h_{v_1\vbar}\, \extd h_{v_2\vbar}\;\prod_{e = (\vbar\vhat)\in\mathcal{E}_{\bp}} \left[\sum_{J_{\vhat}\in \mathcal{J}} \tr_{J_{\vhat}}\left(g_{v_1\vbar\vhat}\;h_{v_1\vbar}^{-1}\; \Sbb_{J_{\vhat},N_0}\; h_{v_2\vbar}\; g_{v_2\vbar\vhat}^{-1}\right)\right]\;,
\end{equation}
where:
\begin{description}
	\item[--] $\bp$ is the boundary patch associated to the vertex $\vbar$ (it may contain loops);
  \item[--] $\mathcal{J}$ is the set of $\gamma$--simple representations of $\mathfrak{g} = Lie(G) = \so(4) \cong \su(2)_+\times\su(2)_-$:
    \begin{equation}
      \mathcal{J} = \left\{J\in \textrm{Irrep}(\mathfrak{g}): J = (j_+,j_-) \;\textrm{with}\; \tfrac{j_-}{j_+} = \tfrac{|1-\gamma|}{1+\gamma}\;\textrm{and} \; j_\pm\in \Nbb/2\right\}\;.
    \end{equation}
  \item[--] $\Sbb_{J,N}$ is the gauge--covariant simplicity operator:
\begin{equation}
  \Sbb_{J,N} = \left\{
    \begin{array}{ll}
      \displaystyle d_{J} \int_{\mathcal{S}_{N}^2}\extd \vec{n}\; |j_+ \vec{n}\rangle {|j_- \vec n\rangle}\, \langle j_+ \vec n| {\langle j_- \vec n |}\;,& (\gamma > 1)\\[0.4cm]
      \displaystyle d_{J}\int_{\mathcal{S}_{N}^2}\extd \vec{n}\; |j_+ \vec{n}\rangle \overline{|j_- \vec n\rangle}\, \langle j_+ \vec n| \overline{\langle j_- \vec n |}\;,& (\gamma < 1)
    \end{array}
    \right.
  \end{equation} 
  where $d_J = (2j_+ +1)(2j_- +1)$, $|j_\pm \vec{n}\rangle$ are $\su(2)$ coherent states,\footnotetext{$\su(2)$ coherent states are defined as:
    \begin{equation}
      |j\vec{n}\rangle = n|jj\rangle\;, 
    \end{equation}
    where $|jj\rangle$ is the highest weight state in the irrep $j$ of $\su(2)$ and $n = exp(\vec{n}\cdot\vec{\sigma})\in \SU(2)$, $\vec{\sigma}$ are the generators of $\su(2)$.}
    $\mathcal{S}^2_{N}$ is the unit 2--sphere in the 3--dimensional hypersurface perpendicular to the 4--vector $N$ and $N_0 = (1,0,0,0)$.
\end{description}
\end{defin}
\begin{remark}
Note that the edge operator factorizes across the edges of the intermediary patch $\bp$. 
Moreover, these factors are independent of the edge being part of a loop.

It is also worth elaborating on the gauge--covariance of the simplicity operator, since this is not usually emphasized in the literature.  The operator $\Sbb_{J,N}$ transforms as:
\begin{equation}
	\Sbb_{J,h\triangleright N} = h\;\Sbb_{J,N}\;h^{-1}\;,
\end{equation}
where $h\in\SO(4)$ and $h\triangleright N$ denotes the rotated 4--vector.
\end{remark}

\begin{defin}[{\bf vertex operator}]
  Consider a vertex $v\in\mathcal{V}$.  The \eprl\ \emph{vertex operator} associated to $v$ is:
  \begin{equation}
    \label{eq:vertexop}
    \sfo_v(g_v)  = \Vbb_{\bbg}(g_v) =   \prod_{(\vb_1,\vh),(\vb_2,\vh)\in\Ecal_{\bbg} }
    \delta(g_{v\vbar_1\vhat},g_{v\vbar_2\vhat})\;,
  \end{equation}
  where $\bbg=(\Vcal_{\bbg},\Ecal_{\bbg})$ is the bisected boundary graph associated to the spin foam atom in $\sfr$ containing $v$. 
  
Notice that this is just the vertex operator of a BF spin foam model. This confirms the general fact that the ingredients of a spin foam amplitudes can be freely shifted from vertex to edges, corresponding to a parallel shift from interaction to kinetic term in the corresponding group field theory formulation.
\end{defin}

\begin{remark}[{\sc eprl} {\bf spin foam amplitude}]
  These operator kernels are the constitutents of the \eprl\ spin foam amplitudes.  Their convolution, guided by the connectivity of the spin foam molecule $\sfm$, to which they are assigned, produces the amplitude $A(\sfm)$ for that molecule : 
  \begin{equation}
  \label{eq:qgmodelsa}
  A(\sfm) = \int [\extd g]\prod_{\vbar\in\Vbar} \sfo_{\vbar}(g_{\vbar}) \prod_{v\in\Vcal}\sfo_v(g_v)\;.
\end{equation}
\end{remark}

\begin{remark}
	A direct quantum gravity interpretation follows after attaching a 4--dimensional reference frame to each vertex $v$, $\vbar$ and $\vhat$. Then, one thinks of each $h_{v\vbar}$ as the parallel transport matrix from $v$ to $\vbar$, while the $g_{v\vbar\vhat}$ are the parallel transport matrices from $\vbar$ to $\vhat$ \emph{before} explicit bonding of the spin foam atoms; hence, the $v$--dependence arises.  The ordered product of elements $h_{v\vbar}$ arising in faces containing $\vhat$, constitutes a holonomy representation of the curvature tensor.  The group elements transport pre--geometric quantities,  which are encoded in the elements of the representation modules labelled by $J_{\vhat}$. At the vertices $\vbar$, simplicity constraints are applied to these elements (via the operator $\Sbb_{J,N}$), to ensure the propagation of a geometric subset of information.
 \end{remark}
 \begin{remark}
   The advantages of the strand diagram realization comes to the fore at this juncture, since the edges and vertex operators factorize over the strands, namely:
   \begin{eqnarray}
	   \label{eq:strandfac}
	   \pbb (g_{v_1\vbar\vhat},g_{v_2\vbar\vhat}; h_{v_1\vbar},h_{v_2\vbar})
	   &=& \sum_{J_{\vhat}\in \mathcal{J}} \tr_{J_{\vhat}}\left(g_{v_1\vbar\vhat}\;h_{v_1\vbar}^{-1}\; \Sbb_{J_{\vhat},N_{\vbar}}\; h_{v_2\vbar}\; g_{v_2\vbar\vhat}^{-1}\right)\;,\\
	   \vbb(g_{v\vbar_1\vhat},g_{v\vbar_2\vhat}) 
	   &=&  \delta(g_{v\vbar_1\vhat},g_{v\vbar_2\vhat})\;.
   \end{eqnarray}


\end{remark}
 \begin{remark}
	 Since the edge operator is not a projector, in order to have the functional form of the kinetic kernel, one should explicitly invert the propagator: $\Kbb_\bp = \Pbb_\bp^{-1}$.  We do not engage in this task for two reasons: \textit{i}) the kinetic operator is of secondary interest to the propagator, since the propagator determines the Feynman amplitudes; \textit{ii}) within a field theory approach, one can transfer the simplicity constraints from the propagator to the the interaction kernels, leaving a projective propagator that may be directly incorporated as the kinetic kernel.   
 \end{remark}
Thus, we have sufficient information already to lay out a multi--field \gft\ for the \kkl--extension of the \eprl\ model.

\fbox{
  \begin{minipage}[c][][c]{0.97\textwidth}
\begin{defin}[{\bf multi--field} {\sc eprl gft}]
A \emph{multi--field \eprl\ group field theory} is defined by a partition function of the form:
\begin{equation}
	Z_{\textsc{mf-}\eprl} = \int [\extd\Phi] \; e^{-S_{\textsc{mf-}\eprl}[\Phi]} 
\end{equation}
where:
\begin{equation}
\label{eq:sfmultiaction}
 S_{\textsc{mf-}\eprl}[\Phi] 
 = \sum_{\bp\in\mathfrak{P}} \int [\extd g]\; \phi_{\bp}(g_{v_1\vbar})\; \Kbb_{\bp}(g_{v_1\vbar},g_{v_2\vbar})\;\phi_{\bp}(g_{v_2\vbar}) 
 + \sum_{\bbg\in \mathfrak{B}} \lambda_{\bbg} \int [\extd g]\; \Vbb_{\bbg}(g_v)\;\prod_{\vbar\in\Vbar}\phi_{\bp}(g_{\vbar})\;.
\end{equation}
\end{defin}
\end{minipage}
}

\subsection{Dually--weighted {\sc eprl gft}}

At this point, we have at our disposal all the tools necessary to incorporate the \eprl\ model within the \dwgft\ formalism.
In the following, we shall deal exclusively with $4$--regular simplicial structures  $\bbgs_{4,\primitive}$, $\sfas_{4,\primitive}$ and $\sfrs_{4,\primitive}$, as well as their labelled counterparts.
\begin{defin}[{\bf dually--weighted variables}]
  The \emph{dually--weighted \eprl\ variables} are:
  \begin{description}
	  \item[--] the group elements $g_{v\vbar\vhat}\in G=\SO(4)$, where $(v\vbar\vhat)\in\Fcal_\sfm$;
	  \item[--] the dual--weighting indices $m_{v\vbar\vhat}\in \{0,1,\dots, M\}$, where $(v\vbar\vhat)\in\Fcal_\sfm$;
  \end{description}
\end{defin}
With the distinction of real and virtual structures comes the responsibility of designing amplitudes that assign the correct effective amplitude to the underlying real spin foam molecule.  To this end, we shall first state the operators and later show their efficacy. 
\begin{defin}[{\bf dually--weighted edge operator}]
  The \emph{dually--weighted \eprl\ edge operator} is:
  \begin{eqnarray}
	  \sfo_{\vbar}(g_{\vbar}, m_{\vbar}) = \Pbb(g_{v_1\vbar},g_{v_2\vbar};{m_{v_1\vbar},m_{v_2\vbar}})
  \end{eqnarray}
  where:
    \begin{eqnarray}
  \Pbb(g_{v_1\vbar},g_{v_2\vbar};{m_{v_1\vbar},m_{v_2\vbar}}) &=&  \int \extd h_{v_1\vbar}\, \extd h_{v_2\vbar}\;\prod_{e=(\vbar\vhat)\in\Ecal_\bp} 
    \Big[
      \pbb_{real}(g_{v_1\vbar\vhat}, g_{v_2\vbar\vhat}; h_{v_1\vbar}, h_{v_2\vbar})\;
      \dbb_{real}(m_{v_1\vbar\vhat},m_{v_2\vbar\vhat}) \notag \\
  &&\hspace{3cm} + \;
\pbb_{virtual}(g_{v_1\vbar\vhat}, g_{v_2\vbar\vhat}; h_{v_1\vbar}, h_{v_2\vbar})\;
\dbb_{virtual}(m_{v_1\vbar\vhat},m_{v_2\vbar\vhat})\Big], \label{eq:extprop}
  \end{eqnarray}
 and:
\begin{itemize}
	  \item[--] $\bp=\bp_4
	  $, the unique unlabelled 4--patch;
	  \item[--] the gravitational factors are:
  \begin{eqnarray}
    \pbb_{real}(g_{v_1\vbar\vhat}, g_{v_2\vbar\vhat}; h_{v_1\vbar}, h_{v_2\vbar})
 &=& \sum_{J_{\vhat}\in \mathcal{J}} \tr_{J_{\vhat}}\left(g_{v_1\vbar\vhat}\;h_{v_1\vbar}^{-1}\; \Sbb_{J_{\vhat},N_{0}}\; h_{v_2\vbar}\; g_{v_2\vbar\vhat}^{-1}\right)\;,\\
     \pbb_{virtual}(g_{v_1\vbar\vhat}, g_{v_2\vbar\vhat}; h_{v_1\vbar}, h_{v_2\vbar})\;
 &=& \delta(g_{v_1\vbar\vhat}\;h_{v_1\vbar}^{-1})\;\delta(h_{v_2\vbar}\;g_{v_2\vbar\vhat}^{-1})\;,
  \end{eqnarray}
\item[--] the dual--weighting factors are:
  \begin{equation}
    \dbb_{real}= 
    \left(
    \begin{array}{c|c}
      1&0\\ \hline
      0&0
    \end{array}
    \right)
    \quad\quad\quad\quad
    \dbb_{virtual} = 
    \left(
    \begin{array}{c|c}
      0&0\\ \hline
      0&\Acal
    \end{array}
    \right)\;.
  \end{equation}
  \end{itemize}
\end{defin}
\begin{remark}
Notice that dual--weighting factors satisfy $\dbb = \dbb_{real} + \dbb_{virtual}$.  The \emph{real} and \emph{virtual} labels encode the conditions of the dual--weighting mechanism outlined in Remark \ref{rem:DWmechanism}.  Moreover, the real gravitational factor coincides with the original $\pbb$ from Equation \eqref{eq:strandfac}.  Thereby, the virtual strand factor decouples the information assigned to the edge in one atom, from that assigned with the other.   
\end{remark}

\begin{defin}[{\bf dually--weighted vertex operator}]
  Consider a bulk vertex $v$.  
  The \emph{dually--weighted \eprl\ vertex operator} is:
  \begin{equation}
    \label{eq:sfvertex2}
    \sfo_v(g_v;m_v) = \Vbb_{\bbg}(g_v,m_v)\;,
    \end{equation}
    where $\bbg$ is the unique unlabelled 4--simplicial bisected boundary graph and:
    \begin{equation}
	    \Vbb_{\bbg}(g_v,m_v) =
\prod_{\substack{f\in\Fcal_\sfm\\ f\supset v}} \left[ \vbb_{real}(g_{v\vbar_1\vhat},g_{v\vbar_2\vhat}) \;  \ibb_{real}(m_{v\vbar_1\vhat},m_{v\vbar_2\vhat}) + \vbb_{virtual}(g_{v\vbar_1\vhat},g_{v\vbar_2\vhat}) \;  \ibb_{virtual}(m_{v\vbar_1\vhat},m_{v\vbar_2\vhat})\right]\;.
    \end{equation}
 The factors are: 
    \begin{equation}
    \vbb_{real}(g_{v\vbar_1\vhat},g_{v\vbar_2\vhat}) = \vbb_{virtual}(g_{v\vbar_1\vhat},g_{v\vbar_2\vhat})=  \delta(g_{v\vbar_1\vhat},g_{v\vbar_2\vhat}) \;, 
    \end{equation}
    and:
    \begin{equation}
	    \ibb_{real} =
	    \left(
    \begin{array}{c|c}
      1&0\\ \hline
      0&0
    \end{array}
    \right)
    \quad\quad\quad\quad
    \ibb_{virtual} = 
    \left(
    \begin{array}{c|c}
      0&0\\ \hline
      0&\Ical
    \end{array}
    \right)\;.
  \end{equation}
\end{defin}

\fbox{
  \begin{minipage}[c][][c]{0.97\textwidth}
\begin{defin}[{\bf dually--weighted} {\sc eprl gft}]
A \emph{dually--weighted \eprl\ group field theory} is defined by a partition function of the form:
\begin{equation}
	Z_{\textsc{dw-}\eprl} = \int [\extd\Phi] \; e^{-S_{\textsc{dw-}\eprl}[\Phi]} 
\end{equation}
where:
\begin{equation}
\label{eq:sfdwaction}
 S_{\textsc{dw-}\eprl}[\Phi] 
 = \frac12\int[\extd g]\sum_{[m]}\phi(g_1;m_1)\;\Kbb(g_1, g_2;m_1,m_2)\;\phi(g_2;m_2)
	+ \lambda\int [\extd g]\sum_{[m]} \Vbb_{\bbg}(g;m)\prod_{j =1}^{n+1} \phi(g_j;m_j)
%
\end{equation}
\end{defin}
\end{minipage}
}

Now, it is time to confirm that these operator assignments lead to the correct effective amplitude. 
\begin{proposition}
In the large--$M$ limit, the effective amplitude assigned by the dually--weighted \eprl\ model to the underlying real spin foam molecule coincides with that of the original \eprl\ model.  
\end{proposition}
\begin{proof}
	Utilizing Proposition \ref{prop:mlimit}, in the large--$M$ limit, the contributing molecules are restricted to those, for which their virtual vertices $\vhat\in\Vhat\subset\Vcal_\sfm$ lie in precisely four virtual faces $f\in\mathcal{F}_{\sfm}$. 
	The amplitude is then:
	\begin{equation}
		\label{eq:dwAmp}
		A(\sfm) = \int [\extd g]\sum_{[m]}\prod_{\vbar\in\Vbar} \sfo_{\vbar}(g_{\vbar};m_{\vbar}) \prod_{v\in\Vcal}\sfo_v(g_v;m_v)\;,
	\end{equation}
  Then, the key calculation examines the effect of integrating out the variables associated to components in the neighbourhood of this vertex $\vhat$. More precisely, the amplitude assigned by the dually--weighted \eprl\ model to a molecule containing such a vertex has the following factors:
  \begin{multline}
	  \vbb_{virtual}(g_{v_1\vbar_{12}\vhat},g_{v_1\vbar_{21}\vhat})\;  \pbb_{virtual}(g_{v_1\vbar_{12}\vhat}, g_{v_2\vbar_{12}\vhat}; h_{v_1\vbar_{12}}, h_{v_2\vbar_{12}})\\[0.3cm]
	  \vbb_{virtual}(g_{v_2\vbar_{12}\vhat},g_{v_2\vbar_{21}\vhat})\; \pbb_{virtual}(g_{v_1\vbar_{21}\vhat}, g_{v_2\vbar_{21}\vhat}; h_{v_1\vbar_{21}}, h_{v_2\vbar_{21}})\;,
	  \label{eq:eval}
\end{multline}
where the configuration is illustrated in Figure \ref{fig:contraction}.

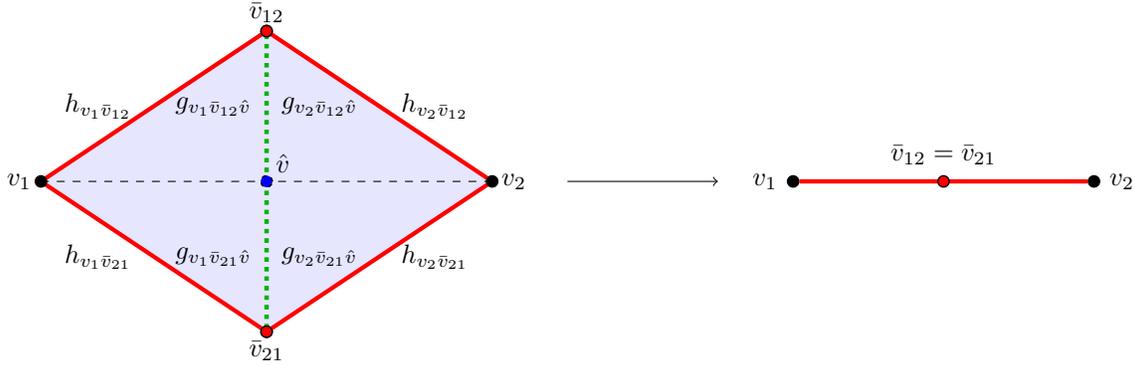
\begin{figure}[htb]
  \centering
  \tikzsetnextfilename{contraction}
  \begin{tikzpicture}[scale=2]

\begin{scope}
\node [c,label=left:$v_1$]		(1)	at (0,0)		{};
\node [c,label=right:$v_2$]	(2)	at (3,0)	 	{};  
\node [c,label=45:$\vh$]		(vh)	at (1.5,0)		{};
\node [c,label=above:$\vb_{12}$](12)at (1.5,1)		{};
\node [c,label=below:$\vb_{21}$](21)	at (1.5,-1)		{};
\foreach \i/\j/\k in {1/12/21,2/12/21}{
\path		[f] 	(\i) 	--  (\j) -- (vh) --  (\k) -- cycle;
\draw	[eh]	(\i)  	-- (vh);
\draw[dotted,eb](vh) 	-- (\j);
\draw[dotted,eb](vh) node[vh]{}	-- (\k);
 }
\draw [e] 	(1) 			-- node [label=left:$h_{v_{1}\vb_{12}}$,label=right:$g_{v_{1}\vb_{12}\vh}$]{} (12) node[vb]{};
\draw [e] 	(1) node[v]{}	-- node [label=left:$h_{v_{1}\vb_{21}}$,label=right:$g_{v_{1}\vb_{21}\vh}$]{} (21) node[vb]{};
\draw [e] 	(2) 			-- node [label=right:$h_{v_{2}\vb_{12}}$,label=left:$g_{v_{2}\vb_{12}\vh}$]{} (12) node[vb]{};
\draw [e] 	(2) node[v]{}	-- node [label=right:$h_{v_{2}\vb_{21}}$,label=left:$g_{v_{2}\vb_{21}\vh}$]{} (21) node[vb]{};
\end{scope}

\draw [->] (3.5,0) -- (4.5,0);

\begin{scope}[xshift=5cm]
\node [v,label=left:$v_1$]		(1)	at (0,0)		{};
\node [v,label=right:$v_2$]	(2)	at (2,0)	 	{};  
\draw [e] (1) -- node[vb,label=above:${\vb_{12} = \vb_{21}}$] {} (2);
\end{scope}
\end{tikzpicture}
  \caption{\label{fig:contraction} Integrating out a virtual face.}
\end{figure}

Manipulating the amplitude, one may simplify the factors in \eqref{eq:eval} by integrating with respect to the elements the set $g_{\vhat}$:
\begin{equation}
	\label{eq:evalResult}
\int dg_{\vhat} \left[ \textrm{Expression \eqref{eq:eval}} \right] 
= \delta(h_{v_1\vbar_{12}}\;h_{v_1\vbar_{21}}^{-1})\;\delta(h_{v_2\vbar_{12}}\;h_{v_2\vbar_{21}}^{-1})\;.
  \end{equation}
   This integration can be completed since the elements in $g_{\vhat}$ only occur within the four factors \eqref{eq:eval}.
   Now, one is free to use these two $\delta$--functions to integrate out the variables $h_{v_1\vbar_{21}}$ and $h_{v_2\vbar_{21}}$, setting $h_{v_1\vbar_{21}}=h_{v_1\vbar_{12}}$ and $h_{v_2\vbar_{21}}=h_{v_2\vbar_{12}}$ in the remaining factors within \eqref{eq:dwAmp} to arrive at \eprl\ amplitude assigned to that molecule obtained from $\sfm$ by molecular reduction along the virtual structure. We illustrate the reduction in Figure \ref{fig:contraction}. 
\end{proof}

\begin{remark}[{\bf imposing greater simplicity}]
	Another tempting proposal is to impose the simplicity constraints on both real \emph{and} virtual structures.
The motivation is that, assuming a polyhedral interpretation is available, the polyhedra corresponding to the states of the model will now be decomposed into \emph{geometric} simplices, the geometricity of each being ensured by the imposition of the simplicity constraints.	
In the \dwgft\ above, this amounts to altering the propagator \eqref{eq:extprop}, using:
	\begin{equation}
		\label{eq:strandfinal}
		\pbb_{virtual}(g_{v_1\vbar\vhat}, g_{v_2\vbar\vhat}; h_{v_1\vbar}, h_{v_2\vbar})
		= \pbb_{real}(g_{v_1\vbar\vhat}, g_{v_2\vbar\vhat}; h_{v_1\vbar}, h_{v_2\vbar})
		= \sum_{J_{\vhat}\in \mathcal{J}} \tr_{J_{\vhat}}\left(g_{v_1\vbar\vhat}\;h_{v_1\vbar}^{-1}\; \Sbb_{J_{\vhat},N_{0}}\; h_{v_2\vbar}\; g_{v_2\vbar\vhat}^{-1}\right)\;.
	\end{equation}
This defines, a priori, a \emph{different-}spin foam model, with an expected higher degree of geometricity.	
A motivation for this change stems from the logic that polytopes that are constructed from geometric simplices are likely to be more physically viable than polytopes constructed from simplices that are only partially geometric (in the sense that the simplicity constraints are not imposed on some of their virtual sub--facets). 

	Of course, it is worth clarifying that the resulting model can be interpreted in two equivalent ways:
	\begin{itemize}
		\item[--] The perturbative series are catalogued by molecules in $\sfmst_{4,\Simplicial}$. In the large--$M$ limit, the surviving molecules are again those of $\sfmst_{4,\Simplicialdually}$.   Within this model, reduction does not lead to effective amplitudes that coincide with those assigned by the \eprl\ model to generic spin foam molecules.
		\item[--] The perturbative series are catalogued by molecules in $\sfms_{4,\Simplicial}$. Due to the coincidence of the strand factors in \eqref{eq:strandfinal}, the dual--weighting part of the amplitude factorizes completely from the gravitational part, as well as over the vertices $\vhat$:
			\begin{equation}
				\label{eq:factorfinal}
				\tr\Big(\prod_{(\vbar\vhat)}\dbb\Big) = 1 =\tr\Big(\prod_{(\vbar\vhat)}\Acal\Big)\;,
			\end{equation}
			where the product is over those edges $(\vbar\vhat)$ incident at $\vhat$. Thus, one get the original simplicial \eprl\ model, with a slight modification of the weights by a factor \eqref{eq:factorfinal} for each vertex $\vhat$.
	\end{itemize}
\end{remark}


\section{Conclusions}

The main purpose of this work has been to show that it is possible to define \gft s compatible with \lqg\ in its full combinatorial generality, that is for quantum states defined on arbitrary boundary graphs, in particular with vertices of arbitrary valence.  

In order to set the ground for our \gft\ construction, we gave a precise and exhaustive classification of the combinatorial structures entering both group field theories and spin foam models, along with their associated boundaries.
To this end we used a physicochemical dictionary, with spin foam molecules obtained as bondings of atoms which are in one-to-one correspondence with these general boundary graphs. In particular, we believe that our classification complements, but also clarifies and completes the one in \cite{\ckkl}, which formed the basis for the first combinatorial generalization of spin foam models.
Moreover, our spin foam molecules turn out to be combinatorial 2--complexes in the precise sense of abstract polyhedral complexes, settling the question of determining the kind of spin foam complexes a theory should be based on when in an abstract, non--embedded context.  (This is at least a starting point, given that considerations about physical symmetries may require extra data to encode D--dimensional topologies).
We argued that these are the relevant combinatorial objects, in terms of which the most general  \gft s and spin foam models are defined.

With the ground properly set, it is straightforward to define a generalization of the well known simplicial \gft\ using arbitrary atoms. We presented explicitly how this can be obtained by a multi-field \gft. 
Since it is extremely difficult to turn such a formally defined field theory, with a potentially infinite number of fields, into an analytically manageable one, with the elaboration of concrete calculations and physical insights, we argued that there is need for an alternative.

Indeed, we introduced dually--weighted \gft s, which generate arbitrary structures, at the expense of a slight modi\-fication of the easiest simplicial \gft.  Therefore, they are as controllable as the latter.
The definition of \dwgft\ has been based firstly on the combinatorial possibility: \textit{i}) to decompose arbitrary boundary graphs into simplicial ones --  this permits their realization with single group field; and furthermore \textit{ii}) to decompose arbitrary spin foam atoms and molecules into simplicial atoms that correspond to a simplicial \gft\ interaction. These facts were proven in the combinatorics section in every detail.

Secondly, \dwgft\ is based on the possibility to implement such a definition at the dynamical level. To this end, we provided an example of a useful application of tensor model techniques to \lqg. We realized a dynamical mechanism for this decomposition of spin foam molecules in terms of a dual--weighting on a simplicial \gft. The effect is that in the large--$M$ limit only those molecules (still built from labelled, simplicial atoms)  that can be canonically reduced to arbitrary spin foam molecules survive. In this way, the \dwgft\ gives rise to an effective perturbative series over arbitrary molecules with the corresponding generalized spin foam amplitudes as dynamical quantum weights.

Finally, we showed that in both cases the implementation of the dynamics of gravitational spin foam models, generalized to arbitrary complexes, is possible. While the implementation along the lines we illustrate is generic,  we provided as an explicit example the spin foam operators in the case of the \eprl\ amplitude, thus obtaining a dually weighted \gft\ whose Feynman amplitudes match the \kkl\ spin foam amplitudes. Moreover, we have given also a modification of the same model, resulting from a better justified imposition of geometricity conditions, as suggested by our dually--weighted construction.

\

There are several tasks one might wish to tackle, on the basis of our results.

Concerning the geometry of gravitational models, the obvious first issue is the implementation of simplicity constraints.
We showed that their implementation in the known models can be straightforwardly applied to the \gft s generating, directly or effectively, arbitrary spin foam molecules. This is in the same spirit as \cite{\ckkl}.
However, known models are all derived from arguments resting on the classical geometry of simplices. 
A spin foam atom with arbitrary combinatorics, on the other hand, corresponds rather to a polytope. Taking the more general combinatorics of \lqg\ in earnest, it follows that a version of the simplicity constraints related to the classical geometry of polytopes is needed. 
Of course, one is then left to deal with the independent matter of quantizing any such geometricity constraints. 

The topological structure of arbitrary molecules should also be considered more carefully. From the simplicial case, it is well known that the good behaviour of a spin foam model of quantum gravity may rest upon the spin foam molecules possessing an extension to a full $D$-dimensional topological structure.  This is important for the definition of a differential structure \cite{Calcagni:2013ku} and geometric quantities such as curvature, the control of divergences \cite{\cgftrenorm, \cCOR},  as well as for diffeomorphism symmetry \cite{Baratin:2011bk,Bonzom:2012tg,Bonzom:2012gwa,Freidel:2003bh,Dittrich:2008pw}. 
A first question is therefore how these issues translate to the case of polyhedral complexes.
A straightforward solution to the issue might be to pass over to  coloured  \gft s, which generate simplicial pseudo--$D$--manifolds. As we have shown, the \dwgft\ can be based on the coloured model without obstacle, and such formalism will then generate effectively all combinatorial $D$--complexes in terms of their triangulations. Still, one may want an encoding of the topology of general $D$--complexes directly at the level of generalized 2--complexes, and this remains an interesting open problem.

Besides these conceptual issues, the most important task is surely the investigation of the field theoretic properties of \dwgft. Among them, one would like to understand the large--$N$ \cite{\clargeN} and double scaling \cite{\cdouble} limits of our (coloured) \dwgft\ and how it compares to the \gft\  theory without dual--weighting. This would extend the results obtained in the context of tensor models. Next, the most important question is probably renormalizability. As mentioned, there is no obstacle preventing the extension of the \dwgft\ from simplicial interactions to a sum over tensor invariant or bubble interactions. Investigating the renormalizability of such models can therefore be carried out using the same techniques already applied in the \gft\ literature \cite{\cgftrenorm, \cCOR}.  
 
Lastly, utilizing a recently proposed strategy  based upon \gft\ condensates \cite{\cGFC}, one can extract effective cosmological dynamics directly from the fundamental \gft\ formulation.  One should expect that a modification of the combinatorial structures entering the microscopic dynamics would percolate directly to such effective macroscopic dynamics. This may lead to interesting modification and could give an alternative way, alongside renormalization analysis, to check the physical relevance and necessity of generalizing the combinatorics of fundamental quantum gravity states and histories. 


\section*{Acknowledgements}
D.O.\ acknowledges financial support through a Sofia Kovalevskaja Award. J.T.\ acknowledges support from Evangelisches Studienwerk Villigst and the Andrea von Braun Foundation.

\newpage


\begin{appendix}

\section{Polyhedral complexes}
\label{sec:PCs}

The Feynman diagrams generated by \gft s are abstract combinatorial objects. 
It is therefore appropriate and necessary
to relate them to abstract combinatorial categories instead of the piecewise linear category.
In this appendix we will provide the definition for combinatorial complexes and show that spin foam molecules are a certain subclass of these.


A generalization of the notion of finite abstract simplicial $\m$-complex, briefly reviewed in  (\ref{sec:simpcomp}), to finite abstract polyhedral $\m$-complex is necessary to account for diagrams of more general \gft s. Providing such a definition (\ref{sec:polycomp}) based on the notion of abstract polytopes (\ref{sec:polytopes}) and proving the relation with spin foam molecules (\ref{sec:SFstructure}) is the main goal of this appendix.

To be clear, the goal is not to show that diagrams in any \gft\ are dual to some polyhedral $\m$-complex which is certainly not true in general. 
The aim is rather to identify the diagrams themselves, that is spin foam molecules and their subclasses (Section \ref{sec:comb}), as polyhedral 2-complexes (\ref{sec:SFstructure}).
Then one could further specify subclasses of polyhedral 2-complexes which allow for an extension to higher $\m$-complexes or for dual complexes of a certain type.

\subsection{Finite abstract simplicial complexes}\label{sec:simpcomp}

To remind the reader on what is meant by a combinatorial complex and for the sake of a self-contained appendix we provide the well known definitions \cite{Kozlov:2008wc} for the simplicial case in this section.

\begin{defin}[{\bf combinatorial simplicial complex}]
A \emph{finite abstract simplicial complex} $\SC$ 
is a collection (multiset) of ordered subsets $\sigma$ of a set of vertices
$\SC_{(0)}=\{v_{1},v_{2},\dots ,v_{N_{0}}\}$ such that
\begin{itemize}
\item[(C1)] for every $\sigma\in \SC$ and $\sigma'\In\sigma$ also $\sigma'\in \SC$.
\end{itemize} 
Such a $\sigma'\In\sigma$ is called a \emph{(boundary) face} of $\sigma$. 
All subsets of cardinality $p+1$ are called \emph{$p$-simplices} $\sigma_{p}\in \SC_{(p)}$ and the dimension $\m$ of $\SC$ is defined as the maximal cardinality of simplices in $\SC$. 
Thus
$\SC=\bigcup_{p=-1}^{\m} \SC_{(p)}$,\footnote{
Every non-empty $\SC$ contains the empty set which is considered as the unique (-1)-simplex, thus $\SC_{(-1)}=\{\emptyset\}\ne\emptyset$.
} 
 and it is also referred to as a \emph{simplicial $\m$-complex}.
\end{defin}

\begin{remark}[{\bf intersection property}]
For piecewise linear cell complexes \cite{Baez:2000kp} a second defining property is that all intersections of cells (simplices in this case) are again part of the complex, in the language of $\SC$:
\begin{itemize}
\item[(C2)] if $\sigma, \sigma'\in\SC$, then $\sigma\cap\sigma'\in\SC$.
\end{itemize}
In the case of abstract simplicial complexes (C2) follows trivially since such intersections are subsets and thus boundary faces which are in $\SC$ due to property (C1).
\end{remark}

A special class of interest are complexes which are pseudo-manifolds. For this the definition common in the context of  simplicial complexes in the topological sense \cite{Seifert:1980uo} extends directly to the combinatorial context (where again cells are simplices)\cite{\cguraulost}:
\begin{defin}[{\bf simplicial pseudo-manifold}]
\label{def:Manifold}
A finite abstract simplicial $\m$-complex $\SC$ is a (finite abstract simplicial) \emph{$\m$-dimensional pseudo-manifold} if it has the following three properties: It is
\begin{itemize}
\item[(M1)] \emph{dimensional homogeneous} (also referred to as \emph{pure}): for each cell in the complex there is a $\m$-cell in the complex which it is a face of.
\item[(M2)] \emph{strongly connected}: Any two $\m$-cells can be joined by a chain of $\m$-cells in which each pair of neighbouring cells has a common $\m$-1-face.
\item[(M3)] \emph{non-branching}: Each $\m$-1-cell is face of at most two $\m$-cells. In the latter case the $\m$-1-cell and all its faces are called \emph{boundary faces} of the complex.%
\footnote{The notion of boundary face defined in this way applies to any combinatorial complex, not necessarily fulfilling (M3).
}
If there are no boundary faces the complex is called a \emph{closed} pseudo-manifold.
\end{itemize}
\end{defin}

The natural ansatz for a generalization from simplicial to polyhedral is to consider a complex built from collections of abstract polytopes instead of simplices.
This poses a twofold challenge. An abstract $p$-simplex defined by an ordered set of its $p+1$ vertices implies at the same time subsimplices given by all its subsets. 
For an abstract polytope the subcell structure has 
to be specified in a different way.
There is a well known alternative way in terms of a partially ordered set (poset) \cite{Kozlov:2008wc}:
\begin{defin}[{\bf poset representation}]
For a finite abstract simplicial complex $\SC$ the \emph{face poset} $\mathcal{F}(\SC)$ is the poset whose elements consist of all nonempty simplices of $\SC$ and whose partial order relation is the inclusion relation on the set of simplices.
\end{defin}
It will turn out in the following that this is the appropriate conceptual framework to extend from simplicial to polyhedral.


\subsection{Abstract polytopes}\label{sec:polytopes}

Fortunately there exists a combinatorial definition of abstract polytopes \cite{McMullen:2009ff,Danzer:1982dp}:

\begin{defin}
\label{def:Polytope}
An \emph{abstract $\m$-polytope}, i.e. an abstract polytope of finite dimension $\m\ge -1$, is a poset $(P,<)$ obeying the properties (P0) - (P3) below. 

Elements of $P$ are called \emph{faces}.
Totally ordered subsets (called \emph{chains}) have length $p$ if they contain exactly $p+1$ faces. If they are maximal they are referred to as \emph{flags} of $P$.
Then the first two defining properties are
\begin{enumerate}
\item[(P0)]{$P$ contains a least and a greatest face, denoted $f_{-1}$ and $f_{\m}$.}
\item[(P1)]{Each flag has length $\m+1$ (which defines the dimension).}
\end{enumerate}
For the statement of the second two defining properties a few more definitions are needed.

The \emph{section} of two faces $f, g$ of $P$ is defined as 
\begin{equation}
f/g := \{h|h\in P, f\le h\le g\}.
\end{equation}
Each section of $P$ is itself a poset obeying the first two properties, with an appropriate dimension (it turns out that it is even an abstract polytope if $P$ is). Thus, identifying each face $f$ with the section over the least face $F\equiv f/f_{-1}$ each face can be attributed a dimension as well. 
Faces different from $f_{-1}$ and $f_{\m}$ are called \emph{proper} faces of $P$. 
As usual one calls 0-faces \emph{vertices} and 1-faces \emph{edges}.

A poset $P$ of dimension $\m$ with properties (P0) and (P1) is defined to be \emph{connected} if either $\m\le 1$, or $\m\ge 2$ and for any two proper faces $f, g$ of $P$ there is a finite sequence of proper faces $f=h_0,h_1,...,h_{k-1},h_k=g$ of $P$ such that $h_{i-1}$ and $h_i$ are incident for $i=1,...,k$.
In this context incidence means that $h_{p-1}\le h_p$ or $h_{p-1}\ge h_p$.

Furthermore $P$ is called \emph{strongly connected} if each section of $P$ (including itself) is connected.

With this we can state the remaining two defining properties:
\begin{enumerate}
\item[(P2)]{$P$ is strongly connected.}
\item[(P3)]{All one-dimensional sections of $P$ are diamond-shaped; that is for every $p=0,1,...\m-1$, if $f$ and $g$ are incident faces of $P$ of dimension $p-1$ and $p+1$, then there are exactly two $p$-faces $h$ of $P$ such that $f<h<g$.}
\end{enumerate}
\end{defin}


\begin{remark}[{\bf Low dimensional polytopes}]
\label{rem:polyexamples}
Up to $\m=2$ there is a very manageable amount of abstract polytopes:
\begin{itemize}
\item Every 0-polytope is a single vertex, having the form $P=\{\emptyset, v\}$ with $\emptyset<v$.
\item Because of (P3), every 1-polytope consists of a single edge, $P=\{\emptyset, v_1, v_2, e\}$ with $\emptyset<v_i<e, i=1,2$.
\item Every finite 2-polytope is a polygon \cite{McMullen:2009ff} of the form shown in Figure \ref{fig:hasse}.\footnote{There is only one infinite 2-polytope \cite{McMullen:2009ff}.}
\end{itemize}
\end{remark}

\begin{remark}[{\bf Hasse diagrams}]
A good way to visualize posets $P$ are Hasse diagrams
(graphs drawn on the plane where vertices represent the elements of $P$  and and edges the transitivity reduced ordering relations, i.e. there is an edge for every two faces $f<g$ in $P$ for which there is no $h$ in $P$ such that $f<h<g$ which goes upwards from $f$ to $g$).  
In particular, since posets $P$ obeying (P0) and (P1) are graded posets $P=\bigcup_{p=-1}^{\m} P_{(p)}$, a canonical way to draw the Hasse diagram is with all elements of each $P_i$ on the same height in the plane (Figure. \ref{fig:hasse}).
\end{remark}

\begin{remark}[{\bf vertex representation}]
\label{rem:VertexRep}
The face set of a graded poset $P=\bigcup_{p=-1}^{\m} P_{(p)}$ (if countable) can be represented by a collection of (ordered) sets in analogy to abstract simplicial complexes in the following way: 
Vertices are labelled in an arbitrary way by natural numbers, $P_{(0)}=\{v_1,v_2,...\}$.
Then, every face $f$ is represented by the ordered set $(v_{i_{1}},v_{i_{2}},\dots)$ consisting of all vertices $v_{i_j}\le f$.%
\footnote{Note again that different faces might have the same vertex set, which is the reason why this representation is a collection, i.e. a multiset. To distinguish explicitly an extra label is needed.}
In particular, the least face $f_{-1}$ is represented by $\emptyset$.

Obviously, the representation in terms of vertices of the face poset of a simplicial complex is just the simplicial complex itself.
For a polytope, the crucial difference to a simplex is that its $p$-face sets are not necessarily of cardinality $p+1$, and in particular (C1) does not hold.
\end{remark}

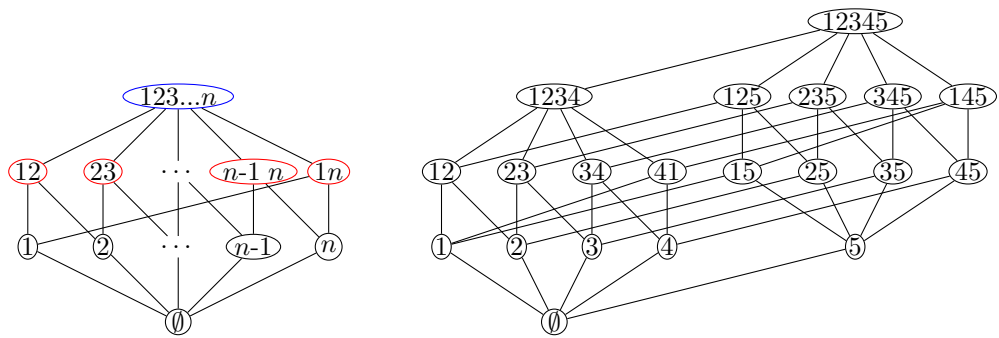
\begin{figure}
\begin{center}
\tikzsetnextfilename{hasse}
\begin{tikzpicture}
\node (empty)	[h, circle]				{$\emptyset$};
\node (3)		[above of=empty]	{\dots};
\node (2)		[h,left of=3]		{2};
\node (1)		[h,left of=2]		{1};
\node (4)		[h,right of=3]		{$n$-1};
\node (5)		[h,right of=4, circle,inner sep=1pt]	{$n$};
\node (12)		[hv,above of=1 ]		{12};
\node (23)		[hv,right of=12 ]		{23};
\node (34)		[right of=23 ]		{\dots};
\node (45)		[hv,right of=34 ]		{$n$-1$\;n$};
\node (51)		[hv,right of=45 ]		{$1n$};
\node (poly)	[he,above of=34]		{123...$n$};
\foreach \from/\to in {12/2,23/3,34/4,45/5,51/1}{
	\draw (poly) 
	 -- (\from);
	\draw (\from)
	 -- (\to);
	\draw (\to) 
	 -- (empty);
	}
\foreach \from/\to in {12/1,23/2,34/3,45/4,51/5}{
	\draw (\from) 
	 -- (\to);
	}
\begin{scope}[xshift=5cm]
\node (0)		[h, circle]					{$\emptyset$};
\node (a)		[h,above of=0,xshift=-15mm]	{$1$};
\node (b)		[h,right of=a]				{$2$};
\node (c)		[h,right of=b]				{$3$};
\node (d)		[h,right of=c]				{$4$};
\node (ab)		[h,above of=a]				{$12$};
\node (bc)		[h,right of=ab]				{$23$};
\node (cd)		[h,right of=bc]				{$34$};
\node (ad)		[h,right of=cd]				{$41$};
\node (f)		[h,above of=ab, xshift=15mm]	{$1234$};
\node (x)		[h,above of=0, xshift=40mm]	{$5$};
\node (ax)		[h,above of=x, xshift=-15mm]	{$15$};
\node (bx)		[h,right of=ax]				{$25$};
\node (cx)		[h,right of=bx]				{$35$};
\node (dx)		[h,right of=cx]				{$45$};
\node (abx)	[h,above of=ax]				{$125$};
\node (bcx)	[h,right of=abx]				{$235$};
\node (cdx)	[h,right of=bcx]				{$345$};
\node (adx)	[h,right of=cdx]				{$145$};
\node (p)		[h,above of=abx,xshift=15mm] 	{$12345$};
\foreach \from/\to in 
{a/0,b/0,c/0,d/0,x/0,
ab/b,bc/c,cd/d,ad/a,
f/ab,f/bc,f/cd,f/ad,
ax/x,bx/x,cx/x,dx/x,
abx/bx,bcx/cx,cdx/dx,adx/ax,
p/abx,p/bcx,p/cdx,p/adx,
abx/ab,bcx/bc,cdx/cd,adx/ad}
    \draw (\from) 
     -- (\to);
\foreach \from/\to in 
{ab/a,bc/b,cd/c,ad/d,
abx/ax,bcx/bx,cdx/cx,adx/dx,
ax/a,bx/b,cx/c,dx/d,
p/f}
    \draw (\from) 
     -- (\to);
\end{scope}
\end{tikzpicture} 
\caption{Hasse diagram of the $n$-polygon (left) and of the pyramid (right) in a representation of faces in terms of vertices (Remark \ref{rem:VertexRep}). 
}
\label{fig:hasse}
\end{center}
\end{figure}
\begin{remark}[{\bf duality}]
A nice property of abstract polytopes is that they have a natural dualization by flipping around the partial order. The finite graded structure, connectedness and diamond shape of 1-sections guarantee that the dual poset is in fact an abstract polytope as well \cite{McMullen:2009ff}. 
In terms of Hasse diagrams the dual polytope is represented by the same graph but read from bottom up to top instead of from top down to bottom.
\end{remark}

\begin{figure}
\begin{center}
\tikzsetnextfilename{hassedual}
\begin{tikzpicture}[every node=rectangle]
\node (0)		[h, circle]					{$\emptyset$};
\node (a)		[hv, below of=0,xshift=15mm]	{$1$};
\node (b)		[hv, left of=a]				{$2$};
\node (c)		[hv, left of=b]				{$3$};
\node (d)		[hv, left of=c]				{$4$};
\node (ab)		[he,below of=a]				{$12$};
\node (bc)		[he,left of=ab]				{$23$};
\node (cd)		[he,left of=bc]				{$34$};
\node (ad)		[he,left of=cd]				{$41$};
\node (f)		[h,below of=ab, xshift=-15mm]	{$1234$};
\node (x)		[hv, below of=0, xshift=-40mm]	{$5$};
\node (ax)		[he,below of=x, xshift=15mm]	{$15$};
\node (bx)		[he,left of=ax]				{$25$};
\node (cx)		[he,left of=bx]				{$35$};
\node (dx)		[he,left of=cx]				{$45$};
\node (abx)	[h,below of=ax]				{$125$};
\node (bcx)	[h,left of=abx]				{$235$};
\node (cdx)	[h,left of=bcx]				{$345$};
\node (adx)	[h,left of=cdx]				{$145$};
\node (p)		[h,below of=abx,xshift=-15mm] 	{$12345$};
\foreach \from/\to in 
{a/0,b/0,c/0,d/0,x/0,
ab/b,bc/c,cd/d,ad/a,
f/ab,f/bc,f/cd,f/ad,
ax/x,bx/x,cx/x,dx/x,
abx/bx,bcx/cx,cdx/dx,adx/ax,
p/abx,p/bcx,p/cdx,p/adx,
abx/ab,bcx/bc,cdx/cd,adx/ad}
    \draw (\from) -- (\to);
\foreach \from/\to in 
{ab/a,bc/b,cd/c,ad/d,
abx/ax,bcx/bx,cdx/cx,adx/dx,
ax/a,bx/b,cx/c,dx/d,
p/f}
    \draw (\from) -- (\to);
\end{tikzpicture}
\caption{Hasse diagram of the dual of the pyramid in Figure \ref{fig:hasse} which is itself a pyramid. Moreover this labeling corresponds to the pyramid vertex in Figure \ref{fig:pyramid}.}
\label{fig:hassedual}
\end{center}
\end{figure}
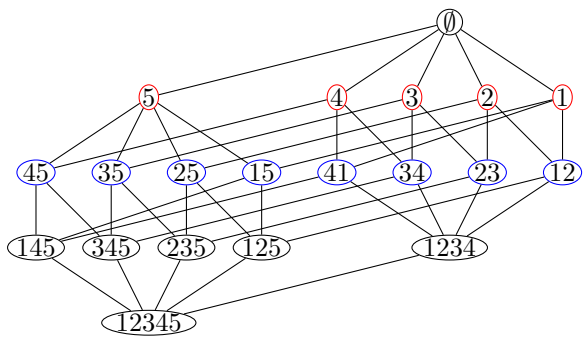


\subsection{Abstract polyhedral complexes}\label{sec:polycomp}

It is now possible to define polyhedral complexes as collections of abstract polytopes in the same spirit as simplicial complexes are collections of simplices. 
To the best of our knowledge, this has not been considered in the literature so far.
Technically, the essential difference between simplicial and polyhedral is the defining condition (C1) which guarantees that the cells are indeed simplices carrying the full structure of subsimplices. While these are just subsets of vertex sets there, for polytopes the subcell structure has to be spelled out explicitly in terms of the partial order relation.

\begin{defin}[{\bf Combinatorial polyhedral complex}]
An \emph{abstract polyhedral complex} $\PC$ is a poset which
\begin{itemize}
\item[(P0')] contains a least face, denoted $f_{-1}$, and
\item[(C1')] for every element $f\in\PC$ the section $F=f/f_{-1}$ is an abstract polytope.
\end{itemize}
\end{defin}

\begin{remark}[{\bf properties of polyhedral complexes}]
\label{rem:polyproperties}
A few comments on the so defined complexes are in order:
\begin{enumerate}
\item Even though there is no single greatest face in an abstract polyhedral complex $\PC$, it is a graded poset $P=\bigcup_{p=-1}^{\infty} P_p$ due to the grading of the polytopes it consists of. If $\PC$ is finite there are polytopes of a maximal dimension $\m$ and $\PC=\PC_{\m}$ can be called an \emph{abstract polytope $\m$-complex}.
\item Therefore a representation of the partial ordering  in terms of Hasse diagrams is possible. 
\item Condition (C1') implies that all faces of a polytope $F\In\PC$ are again polytopes in $\PC$, just because this is true for any abstract polytope (Definition \ref{def:Polytope}).
\item For the same reason the intersection property (C2) is true, i.e. that for any two faces $f,g\in\PC$ their intersection as polytopes is again a polytope in $\PC$,  $f/f_{-1}\cap g/f_{-1}\In\PC$.
\item From the above properties it is obvious that the face poset $\mathcal{F}(\SC_{\m})$ of a simplicial $\m$-complex is an abstract polytope $\m$-complex. Again one can represent the faces of a polyhedral complex by vertex sets as described in Remark \ref{rem:VertexRep}. For $\mathcal{F}(\SC_{\m})$ this gives back the original $\SC_{\m}$ (up to vertex relabeling).
\end{enumerate}
\end{remark}

Now the conditions defining pseudo-manifolds can be directly applied to polyhedral complexes, where cells are now the polytopes:

\begin{defin}[{\bf polyhedral pseudo-manifold}]
An abstract polyhedral $\m$-complex $\SC$ is an abstract polyhedral \emph{$\m$-dimensional pseudo-manifold} if it has the properties (M1-M3) of Definition \ref{def:Manifold}.
\end{defin}
%

In fact, the manifold conditions (M1-M3) are implied by the polytope conditions (P1-P3). Showing this is the crucial part of the following consequence:

\begin{proposition}
Every abstract $\m$-polytope $P$ is an (abstract polyhedral) $\m$-dimensional pseudo-manifold.

The boundary $\br P$ is a closed ($\m$-1)-dimensional pseudo-manifold.
\end{proposition}
\begin{proof}
Let $P$ be an abstract $\m$-polytope. 
For $\m<2$ the proposition is trivial. Therefore let $\m\ge2$ in the following.

The first part is rather straightforward: 
Obviously, $P$ is an abstract polyhedral $\m$-complex.
Since $P=f_{\m}/f_{-1}$ is the single $\m$-polytope in $P$ (because of P0) and thus contains all other polytopes in $P$, (M1) and (M2) follow trivially.
In particular, all $\m$-1-polytopes $F=f/f_{-1}\in P$ are faces only of this single $\m$-polytope and hence are boundary cells of $P$, proving (M3).
Thus $P$ is a pseudo-manifold with boundary $\br P=P-\{f_{\m}\}$.

Since $\br P$ still consist of polytopes (C1') as sections over the unique least face $f_{-1}\in\br P$ (P0') which are of maximal dimension $\m-1$, it follows immediately that $\br P$ is a polyhedral $\m$-1-complex.

The proof, that $\br P$ is further a closed pseudo-manifold, is more illuminating. To this end the properties (M1-M3) for $\br P$ will be shown to follow from the defining properties of $P$, (P1-P3).
\begin{itemize}
\item[(M1):]{Let $f\in \br P$ be an arbitrary face of dimension $0\le p\le\m-1$. 
Since also $f\in P$ and thus $f_{-1}<f<f_{\m}$ it follows (P1) that there is a chain of length $\m$ in $P$ and hence a chain of length $\m-1$ in $\br P$ containing $f$. Hence there is also an $\m$-1-face $g$ with $f<g$.
}
\item[(M2):]{The notion of \emph{strong} connectedness in (M2) is much weaker that in (P2). In fact, (M2) already follows from connectedness in the poset sense: 

Since $P$ is strongly connected (P2) it is also connected. This implies in particular that for $p=\m-1$ every two $p$-faces have a finite sequence of $p$-faces incident along $p-1=\m-2$ dimensional faces.
}
\item[(M3):] {Finally from (P3) it follows that in particular for every $\m$-2 -face $f\in P$ the section $f_{\m}/f$ is diamond shaped; that is, there are exactly two $\m$-1-faces in $P$ which $f$ is a face of.
}
\end{itemize}
Thus, $\br P$ is a closed pseudo-manifold.
\end{proof}


\subsection{Structure of spin foam molecules}
\label{sec:SFstructure}

Now the stage is set to analyse the structure of spin foam molecules, that is what kind of combinatorial complexes they are. Defined as bonding of atoms consisting of triples of vertices, obviously they are simplicial 2-complexes. 
But since these triangular faces are only wedges of actual larger faces they turn out to be simplicial subdivisions of polyhedral complexes and generalizations thereof. 

In any case, with the understanding of the manifold conditions (M1-M3) it is clear that spin foam molecules  are homogeneous of dimension two (M1) and obviously strongly connected (M2). But since they are intended to capture a higher dimensional structure of $D>2$, in all interesting cases of spin foam atoms they are branching.

In this section we discuss these statements in detail.

\begin{proposition}[{\bf spin foam molecules}]
\label{prop:Molecules}
Spin foam molecules are homogeneous, strongly connected simplicial 2-complexes. 
\begin{proof}
Let $\sfr= (\mathcal{V}_{\sfr}, \mathcal{E}_{\sfr}, \mathcal{F}_{\sfr})\in\sfrs_{}$ be a spin foam molecule. 
By definition (\ref{def:molecule}), its vertex set comes with a graded, tripartite structure $\Vcal_{\sfr}=\Vcal\cup\Vb\cup\Vh$ and every face $f\in\mathcal{F}_{\sfr}$ is defined by a triple of vertices $f=(v,\vb,\vh)\in\Vcal\times\Vb\times\Vh$.
Furthermore, according to definition \ref{def:sfatom}, for $f=(v,\vb,\vh)\in\Fcal_{\sfr}$ every pair of vertices is an edge in $\Ecal_{str}$, concluding the proof of the defining condition (C1) of simplicial complexes.

By the same definition \ref{def:sfatom}, for every vertex $v\in\Vcal_{\sfr}$ there is a face $f\in\Fcal_{\sfr}$ such that $v\in f$, proving homogeneity (M1). 
Finally, in (M2) holds since in an atom every pair of triangles is strongly connected and the bonding transfers this property to the whole molecule.
Thus $\sfr$ is a homogeneous and strongly connected simplicial complex.
\end{proof}
\end{proposition}

Nevertheless, spin foam molecules are usually regarded as something more general than simplicial complexes. Indeed one can consider our definition of the molecules as a triangulation of more general complexes. 
This can be made precise in the following way:
A simplicial subdivision%
\footnote{Even in the combinatorial topology context this is often called barycentric subdivision \cite{Kozlov:2008wc}, even though there is no notion of centre in the abstract setting. For this reason, and to highlight that it is a subdivision into simplices, we prefer to call it  "simplicial subdivision".}
of polyhedral complexes can be defined exactly the same way as done in the case of simplicial complexes \cite{Kozlov:2008wc}, by defining vertices for every face and simplices for every chain, effectively subdividing all polytopes into simplices.
Including boundaries we introduce one modification to the standard definition, identifying the subdividing vertices of each boundary $\m$-1-cell with the subdividing vertex of the single $\m$-cell it is a face of:
\begin{defin}[{\bf simplicial subdivision}]
The \emph{simplicial subdivision} of an abstract polyhedral $\m$-complex $\PC$ is the simplicial complex
\[
\ssub\PC := \left\{ \{f_1,f_2,\dots,f_t\} | f_1>f_2>\dots>f_t, f_i\in\PC/\sim , t\ge1 \right\}.
\]
Here $f,g\in\PC$ are equivalent, $f\sim g$, if and only if either $f=g$ or $f\in\PC_{(\m-1)}$ and $g\in\PC_{(\m)}$ is the unique $\m$-cell such that $f<g$.
\end{defin}

\begin{remark}
Spin foam molecules have a very similar structure: Vertices $\vhat\in\Vhat$ correspond to faces and vertices $\vbar\in\Vbar$ to edges,
Therefore they can be regarded as simplicial subdivisions of some 2-dimensional objects.
To determine their structure we define  for a molecule $\sfr = (\Vcal_{\sfr},\Ecal_{\sfr},\Fcal_{\sfr}) \in\sfrs$ with $\Vcal_{\sfr}=\Vcal\cup\Vb\cup\Vh$ the inverse to the subdivision,
$\Cc ^{\sfr} = \{\emptyset\}\cup\Cc ^{\sfr}_{(0)}\cup\Cc ^{\sfr}_{(1)}\cup\Cc ^{\sfr}_{(2)}$,
in the following way:
\begin{description}
  \item[--] $\Cc ^{\sfr}_{(0)} := \Vcal\cup\Vbar_{\delta\sfr}$ is the set of bulk vertices and of boundary graph vertices on the boundary of $\sfr$, $\Vbar_{\delta\sfr}=\Vb\cap\Vcal_{\delta\sfr}$.
    
  \item[--]{ $\Cc ^{\sfr}_{(1)} := \Ecal_{int}\cup\Ecal_{ext}$,
  that is internal edges, either between bulk vertices or a bulk and one boundary vertex in $\Vb$,
  \[
  \Ecal_{int} = \left\{(v_1,v_2) | \exists \vb\in\Vb: \{(v_1\vb),(v_2,\vb)\}\In\Ecal_{\sfr}\right\}
   \cup \left\{(v,\vb)|v\in\Vcal, \vb\in\Vb_{\delta\sfr} \right\} 
  \]
  and boundary edges between two boundary vertices in $\Vb$
  \[
   \Ecal_{ext} = \left\{(\vb_1,\vb_2)|\exists \vh\in\Vh: (\vb_1\vh),(\vb_2,\vh)\in\Ecal_{\sfr}\right\}
   \]
   One can show that indeed the latter are edges on the boundary according to (M3). 
   The internal edges are in one-to-one correspondence to the vertices in $\Vb$.
  }
  \item[--] $\Cc ^{\sfr}_{(2)} := \left\{ (\Cc ^{\sfr}_{(0)}\cap\bigcup_{f\in\Fcal_{\sfr} : \vh\in f} f) | \vh\in\Vh \right\} $ is the set of unions of all triangles sharing a bisection point $\vh\in\Vh$. 
   These are either of the form $(v_1,v_2,...,v_k)$ for $k$ vertices $v_i\in \Vcal$ or, if they contain a boundary edge $(\vb_1,\vb_2)\in\Ecal_{ext}$, of the form $(\vb_1,\vb_2,v_1,v_2,...,v_k)$.    Due to the definition of bondings (remark \ref{rem:bondingexample}) these are the only two possibilities.
\end{description}
As this is a vertex representation, a partial ordering is given by the inclusion relations between cells in $\Cc^{\sfr}$.

It is then straightforward to show that $\ssub\Cc^{\sfr} = \sfr$. 
\end{remark}

\begin{proposition}[{\bf Loopless spin foam molecules}]
\label{prop:LMolecules}
Loopless spin foam molecules without self-bondings are simplicial subdivisions of homogenous, strongly connected polyhedral 2-complexes.
\end{proposition}
\begin{proof}
Let  $\sfr\in\sfrs$ and consider $\Cc^{\sfr} = \{\emptyset\}\cup\Cc ^{\sfr}_{(0)}\cup\Cc ^{\sfr}_{(1)}\cup\Cc ^{\sfr}_{(2)}$. 
Trivially $\emptyset$ is the least face and elements of $\Cc ^{\sfr}_{(0)}=\Vcal\cup\Vbar_{\delta\sfr}$ and $\Cc ^{\sfr}_{(1)}$ are polytopes, being respectively vertices and edges build from those vertices.
For the proof that $\Cc^{\sfr}$ is a polyhedral 2-complex, it remains to show that the elements of $\Cc ^{\sfr}_{(2)}$, together with their subsets in $\Cc ^{\sfr}$, are polygons.

Let $\vh\in\sfr$ and $f_{\vh}=\Cc ^{\sfr}_{(0)}\cap\bigcup_{f\in\Fcal_{\sfr} : \vh\in f} f$ the corresponding face in $\Cc ^{\sfr}_{(2)}$.
Consider first the case in which there is exactly one bulk vertex  $ v\in V$ part of that face, $v\in f_{\vh}$. Then, since there are no self-bondings in $\sfr$, the face has the form $\Cc ^{\sfr}_{(0)}\cap(v,\vb_1,\vh)\cup(v,\vb_2,\vh)=(v, \vb_1,\vb_2)$ for $\vb_1,\vb_2\in \Vcal_{\delta\sfr}$.%
\footnote{A self-bonding which identifies $\vb_1$ and $\vb_2$ leads to a face $(v)$ of a self-loop edge $(vv)$. This is not a polytope because it violates (P3).
}
All the two-element subsets are edges in $\Cc ^{\sfr}_{(1)}$, thus the section $f_{\vh}/\emptyset$ is a polytope in $\Cc ^{\sfr}$. 
One can then show by induction that every bonding taking $\vh$ into account effectively adds another $v_i \in \Vcal$ to $f_{\vh}$ and edges of the so defined polygon are still in $\Cc ^{\sfr}_{(1)}$. Finally, it may then, for  $|f_{\vh}\cap\Vcal|>1$ occur that  $\vb_1$ and $\vb_2$ are bonded to each other and thus are not part of $f_{\vh}$ anymore. But still the section $f_{\vh}/\emptyset\In\Cc ^{\sfr}$.
This concludes the proof of (C1') and of $\Cc ^{\sfr}$ being a polyhedral complex.

Finally, homogeneity and strong connectedness of $\Cc ^{\sfr}$ are directly induced by 
$\sfr$ having these properties as a simplicial complex (Proposition \ref{prop:Molecules}).
\end{proof}

\begin{remark}
From the proof of Proposition \ref{prop:LMolecules} it is clear that spin foam molecules, in their full generality, have to be described by an extension of the polytope concept which includes loops. Loops occur in self-bondings of atoms as well as for patches of boundary graphs with loops, leading to faces with only one boundary edge.
Both cases can be easily included in a definition of \emph{generalized polytopes} by loosening (P3), allowing two or one $p$-faces in sections of $p$+1 with $p$-1 faces.

One can then prove that spin foam molecules $\sfrs$ are simplicial subdivisions of generalized polyhedral complexes.
We are not presenting the details for this here because it is rather straightforward. Moreover, there are good reasons to prefer polyhedral $\m$-complexes to generalized polyhedral $\m$-complexes from a quantum gravity perspective: While the former might have a higher dimensional extension to pseudo-$D$-manifolds ($D>\m$), this is not expected for the latter. It has already been shown in the simplicial ($n$-regular) case that, with the same extension, the loops in self-bondings lead to degeneracies such that there is no interpretation as pseudo-$D$-manifolds \cite{\cguraulost}.

In a \gft\ it is rather straightforward to implement the property that no self-bondings occur in the generation of complexes. A complex field, together with interaction terms as functionals of either the field or its complex conjugate are enough to generate bipartite graphs. In that case, no atom can be bonded to itself in the complexes generated by the \gft. From this perspective, genuine polyhedral complexes are indeed the only combinatorial objects occurring.
\end{remark}

\end{appendix}

\

\bibliographystyle{JHEP}
\bibliography{GFTallLQG}

\end{document}